\documentclass[11pt,a4paper]{article}
\usepackage[margin=2.5cm]{geometry} 

\usepackage{amsfonts,bm}
\usepackage{textcomp}
\usepackage{wrapfig}
\usepackage{algpseudocode}
\usepackage[linesnumbered,ruled,vlined]{algorithm2e}
\usepackage{blkarray}
\usepackage{listings}
\usepackage{inconsolata}
\usepackage{makecell}
\usepackage{pifont}
\usepackage{siunitx}
\usepackage{nicefrac}
\usepackage{xparse}
\usepackage[numbers,sort&compress]{natbib}

\usepackage{amsmath,amssymb,amsthm}
\usepackage{graphicx}
\usepackage{enumitem}
\usepackage{booktabs}
\usepackage{multirow}
\usepackage{hyperref}
\usepackage{caption}
\usepackage{subfig}
\usepackage{tabularx}
\usepackage{xcolor} 
\usepackage{authblk}

\usepackage{array} 

\newcolumntype{C}{>{\centering\arraybackslash}X}
\newcolumntype{L}{>{\raggedright\arraybackslash}X}
\newcolumntype{R}{>{\raggedleft\arraybackslash}X}

\SetCommentSty{mycommfont}

\newcommand{\cmark}{\ding{51}}
\newcommand{\xmark}{\ding{55}}
\newcommand{\pmark}{$\circ$}

\definecolor{codegreen}{rgb}{0,0.6,0}
\definecolor{codegray}{rgb}{0.5,0.5,0.5}
\definecolor{codepurple}{rgb}{0.58,0,0.82}
\definecolor{backcolour}{rgb}{0.95,0.95,0.92}

\lstdefinestyle{mystyle}{
  backgroundcolor=\color{backcolour},
  commentstyle=\color{codegreen},
  keywordstyle=\color{magenta},
  numberstyle=\tiny\color{codegray},
  stringstyle=\color{codepurple},
  basicstyle=\ttfamily\footnotesize,
  breaklines=true,
  numbers=left,
  numbersep=5pt,
  showstringspaces=false,
  captionpos=b
}
\newcommand{\hl}[1]{#1}

\newtheorem{theorem}{Theorem}
\newtheorem{lemma}{Lemma}

\NewDocumentCommand{\subcaptionbox}{m o o m}{%
  \IfNoValueTF{#2}
    {\subfloat[\centering #1]{#4}}
    {\subfloat[\centering #1]{\begin{minipage}{#2}\centering #4\end{minipage}}}%
}

\title{MineRobot: An Actuator-Centered Kinematic Modeling and Solving Framework for Underground Mining Robots}

\author[1]{Shengzhe Hou}
\author[1,2,*]{Xinming Lu}
\author[1]{Tianyu Zhang}
\author[3]{Changqing Yan}
\author[1]{Xingli Zhang}

\affil[1]{\small College of Computer Science and Engineering, Shandong University of Science and Technology, Qingdao 266590, China}
\affil[2]{\small Shandong Lionking Software Co., Ltd., Tai'an 271000, China}
\affil[3]{\small College of Intelligent Equipment, Shandong University of Science and Technology, Tai'an 271000, China}
\affil[*]{\small Correspondence: luxinming@sdust.edu.cn}

\date{} 
\begin{document}
\maketitle

\begin{abstract}
Underground mining robots are increasingly modeled for planning, operator training, and digital-twin workflows, where reliable actuator-level kinematics is needed to reduce hazardous in situ trials.
Unlike typical open-chain industrial manipulators, representative mining machines are often linear-actuator-driven closed-chain mechanisms with planar four-bar linkages, making reusable kinematic modeling and real-time FK/IK solving challenging.
We present \textit{\hl{MineRobot}}, an actuator-centered framework for modeling and solving the kinematics of this representative mechanism class.
MineRobot introduces the Mining Robot Description Format (MRDF), a domain-specific representation that parameterizes mining-robot kinematics with native semantics for actuators and loop closures.
It then contracts planar four-bar substructures into generalized joints and extracts, for each actuator, an Independent Topologically Equivalent Path (ITEP) classified into four canonical types.
Based on this decomposition, per-type solvers are composed into a sequential forward-kinematics (FK) pipeline, while inverse kinematics (IK) is formulated as a bound-constrained actuator-length optimization solved by a Gauss--Seidel-style update scheme.
By converting coupled closed-chain kinematics into small topology-aware solves, MineRobot reduces robot-specific hand derivations and supports efficient repeated FK/IK computation without treating each query as a full coupled constraint-solving problem.
Experiments on representative underground mining robots demonstrate real-time FK performance and robust IK convergence within the tested operating ranges, supporting the use of MineRobot as an actuator-centered kinematic layer for planning, training, and digital-twin workflows.
\end{abstract}
\textbf{Keywords:} actuator-centered kinematics; forward kinematics; inverse kinematics; closed-chain mechanisms; underground mining robots


\section{Introduction}
\label{sec:intro}

Underground mining robots are increasingly used for excavation, extraction, roof support, haulage, ventilation, and other operations in hazardous and confined environments~\cite{peng2019longwall,ralston2017longwall,deshmukh2020roadheader}.
Because these machines must generate large forces, maintain structural stability, and operate safely under harsh working conditions, many representative platforms rely on hydraulic or pneumatic linear actuators rather than only rotary motors.
Typical examples include roadheaders~\cite{deshmukh2020roadheader,Yan2025roadheader}, longwall hydraulic supports, and shearers~\cite{peng2019longwall,ralston2017longwall} as shown in Figure~\ref{fig:robotimgs}.

\begin{figure}[htbp]
    \centering
    \includegraphics[width=0.5\linewidth]{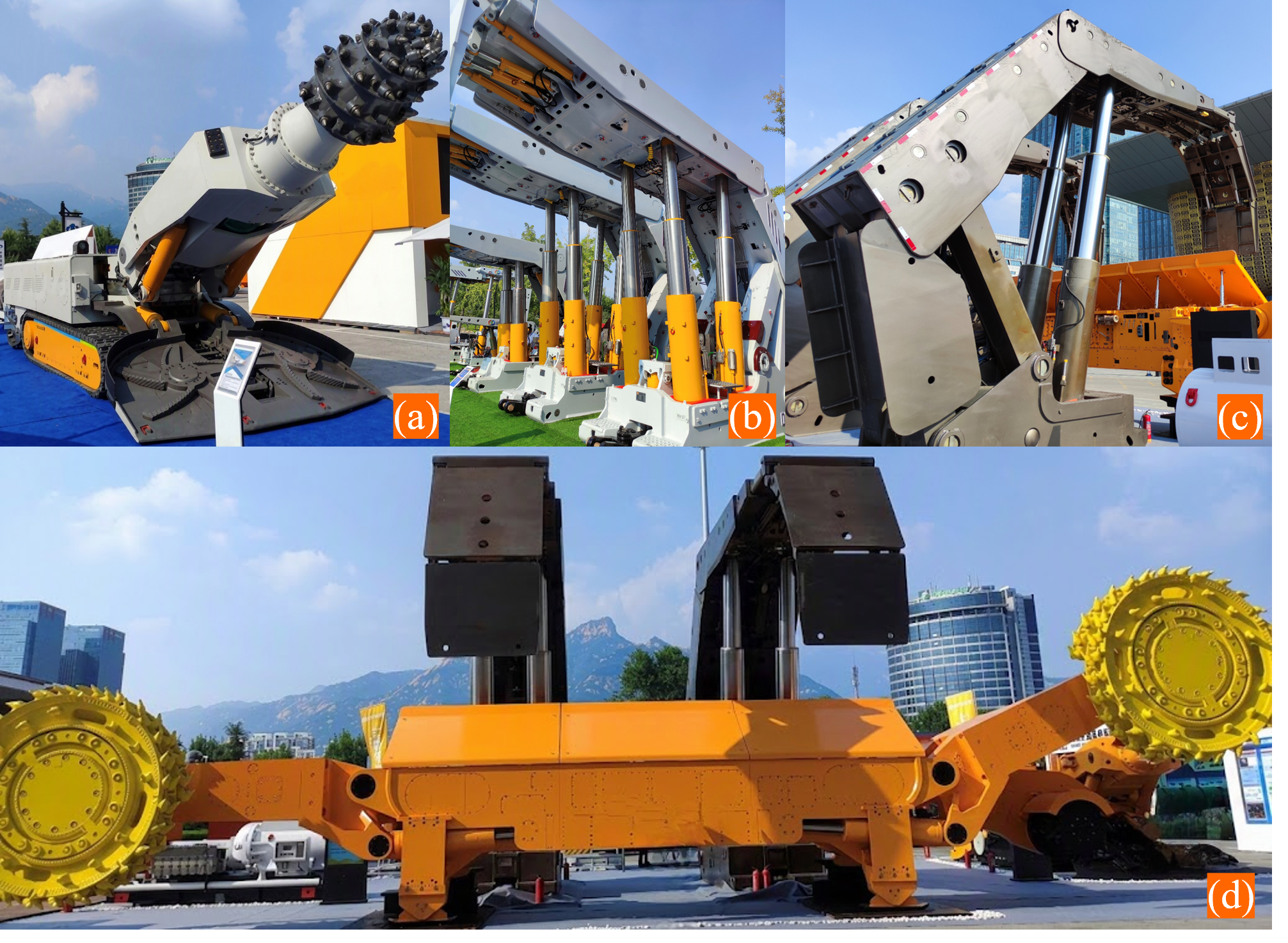}
    \caption{\hl{Representative} mining-robot platforms. (\textbf{a}) Roadheader for tunnel excavation; (\textbf{b},\textbf{c}) hydraulic supports for stabilizing the longwall face; (\textbf{d}) shearer for mineral cutting.}
    \label{fig:robotimgs}
\end{figure} 

Unlike typical open-chain industrial manipulators~\cite{perez2019industrial}, the underground mining machines targeted in this work are often linear-actuator-driven closed-chain mechanisms composed mainly of revolute, prismatic, and fixed joints, with planar four-bar linkages as recurring substructures~\cite{guan2019dynamic,guo2024adaptive,prebil2002synthesis}.
These structures introduce loop-closure constraints and actuator-coupled link motion, making actuator-level forward kinematics (FK) and inverse kinematics (IK) less straightforward than in tree-structured robots~\cite{merlet2006parallel,taghirad2013parallel}.
In FK, actuator lengths must be mapped to configurations that satisfy closed-chain constraints; in IK, target poses must be converted into feasible actuator lengths under stroke limits, including cases where the target lies outside the reachable workspace.
Manual loop-equation derivation is tedious and mechanism-specific, while repeatedly solving the full coupled nonlinear system can be brittle or expensive in planning, training, teleoperation, and digital-twin workflows~\cite{mei2025sensing,ge2020virtual,xie2022framework}.

A reusable actuator-level representation and solver are therefore needed.
Existing virtual-environment or simulation studies for mining robots are usually platform- or application-specific~\cite{ge2020virtual,hou2023interactive,xie2022virtual}, which hinders reuse across different mining robots.
General-purpose robot description formats and robotics frameworks, such as URDF, SDF, MJCF, Drake, and MuJoCo~\cite{quigley2009ros,sdformat,todorov2012mujoco,drake}, provide important foundations, but mining-robot-specific structures are often represented through generic joints, constraints, drives, or external logic.
This creates a gap between the actuator-length variables naturally used by mining equipment and the joint- or constraint-centered formulations commonly used for modeling and solving.

Within this scope, we present \textit{\hl{MineRobot}}, an actuator-centered framework for modeling and solving the kinematics of a representative class of underground mining robots (Figure~\ref{fig:framework}).
The target mechanism class consists of linear-actuator-driven mechanisms composed mainly of revolute, prismatic, and fixed joints, with planar four-bar linkages as recurring substructures.
MineRobot introduces the Mining Robot Description Format (MRDF) to represent links, joints, linear actuators, stroke bounds, redundancy, and loop-closure information in a form consumed directly by the solver pipeline.
It then contracts planar four-bar substructures into generalized joints and extracts an Independent Topologically Equivalent Path (ITEP) for each actuator.
The resulting actuator-level decomposition supports a sequential FK pipeline and an FK-aligned, bound-constrained IK solver over actuator lengths.
We also provide an open-source online interactive demo of MineRobot implemented as a simplified JavaScript \hl{version} 
 (\url{https://github.com/Housz/MineRobot}).

\begin{figure}[htbp]
\centering 
        \includegraphics[width=1.0\linewidth]{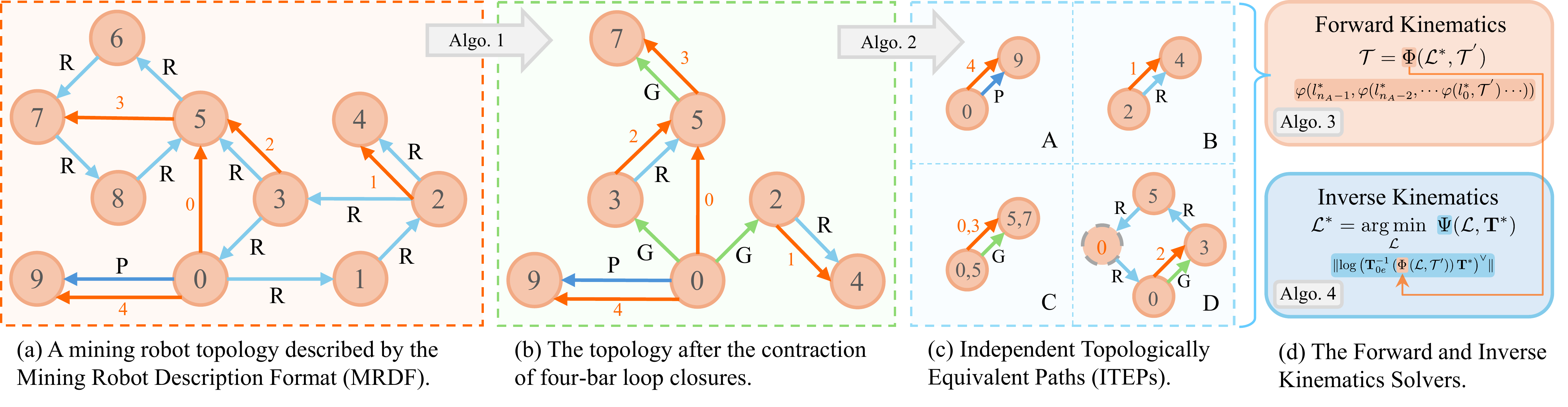}
    \caption{\textbf{\hl{MineRobot}
 } framework.
    (\textbf{a}) Robot topology modeled by the Mining Robot Description Format (MRDF): links (nodes), revolute/prismatic joints (light/dark blue arrows), and linear actuators (orange arrows).
    (\textbf{b}) Four-bar linkages are \emph{\hl{contracted}} into generalized joints (green), after which the joint topology excluding actuators is acyclic.
    (\textbf{c}) For each actuator, we extract its Independent Topologically Equivalent Path (ITEP) and categorize ITEPs into four types.
    (\textbf{d}) Sequentially solving actuator ITEPs yields a forward-kinematics (FK) pipeline. Inverse kinematics is formulated as an FK-aligned optimization problem.}
    \label{fig:framework}
\end{figure}

MineRobot is not intended as a general-purpose framework for arbitrary closed-chain robots.
Its topology construction assumes that planar four-bar linkages can be contracted into generalized joints and that the actuator-free topology is acyclic or satisfies the independent-path conditions required for ITEP extraction.
The present formulation is also limited to the kinematics layer; dynamics, contact, force transmission, hydraulic control, physical calibration, mobile-base motion planning, and field-level decision making should be handled by higher-level modules that use MineRobot's FK/IK routines as an actuator-centered kinematic sublayer.

This paper makes the following contributions:
\begin{itemize}[left=0pt]
    \item \emph{\hl{MRDF.}} A description format that natively models links, joints, and actuators with explicit loop-closure support, forming the basis for actuator-level kinematics computation (Section~\ref{sec:modeling}).
    \item \emph{\hl{Topology pipeline.}} A pipeline comprising four-bar contraction to generalized joints and automatic per-actuator ITEP extraction with a four-type classification (Section~\ref{sec:topo}). 
    \item \emph{\hl{FK solver.}} Four per-type ITEP solvers composed sequentially to realize an actuator-centered forward-kinematics solver for the full robot (Section~\ref{sec:solver}). 
    \item \emph{\hl{IK solver.}} An FK-aligned, bound-constrained inverse-kinematics formulation and a Gauss--Seidel-style optimizer that alternately updates actuator lengths for efficient solution (Section~\ref{sec:solver}).
\end{itemize}

\section{Related Work}\label{sec:related-work}

We review related work from four perspectives: (i) actuator-driven mining robots and closed-chain mechanisms, (ii) robot description formats and actuator semantics, (iii)~kinematics solvers for closed-chain and actuated mechanisms, and (iv) robotics frameworks and applications.

\subsection{Actuator-Driven Mining Robots and Closed-Chain Mechanisms}
Underground mining robots are designed for heavy-load operation in confined, hazardous, and uncertain environments, including longwall mining~\cite{peng2019longwall,ralston2017longwall} and roadheading~\cite{deshmukh2020roadheader}.
Compared with typical industrial manipulators, these robots more often use hydraulic or pneumatic linear actuators~\cite{guan2019dynamic,guo2024adaptive} to obtain high force density, structural stiffness, and safe operation in underground settings.
Hydraulic supports~\cite{guan2019dynamic}, roadheaders~\cite{deshmukh2020roadheader,tian2018kinematic}, shearers, and related equipment commonly contain closed-chain or parallel substructures~\cite{merlet2006parallel}, redundant actuator arrangements~\cite{gosselin2018redundancy}, and planar four-bar linkages~\cite{prebil2002synthesis} that transmit force while constraining motion.
Existing studies have addressed individual mining mechanisms from the perspectives of relative kinematics for coal-mine cutting robots~\cite{liu2024relative}, hydraulic-support attitude adjustment and pose estimation~\cite{ge2020virtual,mei2025sensing,wang2024method}, and mechanism-specific kinematic or structural modeling of hydraulic supports~\cite{guan2019dynamic}.
These works are valuable for specific platforms, but they typically rely on mechanism-specific models or hand-derived equations.
A reusable actuator-centered kinematics framework that captures the common closed-chain structures of multiple underground mining robots remains underexplored.

\subsection{Robot Description Formats and Actuator Semantics}
Robot description formats~\cite{nordmann2014survey,qiu2024describing,ivanou2021robot} act as domain-specific languages that organize geometry, topology, joints, and other information for downstream modeling and computation.
The Unified Robot Description Format (URDF)~\cite{quigley2009ros,tola2023understanding}, originating from the Robot Operating System (ROS)~\cite{macenski2022robot}, has become a widely used standard for robot modeling and has strong support for tree-structured, joint-centered mechanisms.
However, URDF does not natively encode closed-loop mechanisms, linear actuators as first-class kinematic entities, redundant actuator groups, or four-bar abstractions.
Alternative formats and stacks, such as SDFormat (SDF)~\cite{sdformat}, MuJoCo XML (MJCF)~\cite{todorov2012mujoco}, and Drake's modeling pipeline~\cite{drake}, provide broader support for constraints, simulation models, or multibody systems.
Recent URDF extensions, including URDF+~\cite{chignoli2024urdf+} and Extended URDF~\cite{batto2025extended}, also seek to represent parallel mechanisms and loop closures more explicitly.
Nevertheless, these formats are generally intended to be broad and software-stack compatible.
They do not directly expose the mining-robot-specific actuator semantics and topology information required to schedule actuator-centered FK/IK solvers.
MRDF is therefore designed as a kinematics-oriented representation rather than a replacement for general-purpose robot or scene formats: it records links, joints, linear actuators, redundant actuator groups, and loop-closure information in a form that can be consumed directly by the topology-processing and solver pipeline.

\subsection{Kinematics Solvers for Closed-Chain and Actuated Mechanisms}
Forward kinematics (FK) and inverse kinematics (IK) are foundational problems in robotics and articulated mechanism modeling.
For open-chain robots~\cite{siciliano2016robotics}, FK is obtained by propagating transformations along a tree, while IK can be addressed by analytic methods, Jacobian-based methods, nonlinear optimization, or data-driven approaches~\cite{aristidou2018inverse,le2019kinematics}.
Closed-chain and parallel mechanisms~\cite{taghirad2013parallel,merlet2006parallel,staicu2019dynamics} are more difficult because the solution must satisfy loop-closure constraints while handling multiple configurations, singularities, and actuator limits.
A common practice is to derive mechanism-specific closure equations and solve the resulting nonlinear system numerically~\cite{mei2025sensing,wang2024method}.
Recent work on linearly actuated parallel manipulators has similarly derived direct and inverse kinematics for a specific actuator arrangement~\cite{choi2021kinematic}, illustrating the continuing need for mechanism-specific kinematic formulations.
For a single well-defined mechanism, such hand-derived or closed-form analyses can be highly efficient and may provide the best solution for that specific robot.
However, they require expert derivation and implementation for each new mechanism, and the resulting models are usually difficult to reuse across different mining-robot topologies.
MineRobot addresses a different objective: it seeks a reusable actuator-centered formulation for a family of mining robots, rather than a manually optimized derivation for one mechanism.
For IK, feasible targets may be solved locally, whereas targets outside the reachable workspace are often formulated as constrained optimization problems that seek the closest feasible configuration~\cite{zhao1994inverse,erleben2019solving}.
These approaches are general in principle, but they usually treat all loop constraints as a coupled system.
In contrast, MineRobot uses the topology of mining robots to reorganize the problem: four-bar substructures are contracted, an Independent Topologically Equivalent Path (ITEP) is extracted for each actuator, and per-type local solvers are composed into a sequential FK pipeline and an FK-aligned IK optimization.

\subsection{Robotics Frameworks and Applications}
General-purpose robotics frameworks, including Gazebo~\cite{koenig2004gazebo}, MuJoCo~\cite{todorov2012mujoco}, Drake~\cite{drake}, PyBullet~\cite{coumans2021pybullet}, Isaac Sim~\cite{nvidia_isaac_sim}, and the Dynamic Animation and Robotics Toolkit (DART)~\cite{lee2018dart}, provide broad support for robot modeling, multibody computation, dynamics, control, visualization, and testing.
Commercial mechanism-analysis and CAD/multibody tools can also model and analyze complex mechanisms without requiring users to manually derive all governing equations.
These tools are powerful and widely applicable, and they are not the target of replacement by MineRobot.
However, for the mining-robot mechanisms considered here, such workflows are usually organized around mechanism-specific multibody models, mates, constraints, drivers, and simulation setups; repeated actuator-level FK/IK queries may still require additional constraint definitions, external logic, or full coupled-system solves.
This creates a workflow gap between general-purpose robotics frameworks, commercial mechanism-analysis tools, and underground mining robots, where actuator lengths are the natural input/output variables and repeated queries should exploit the structural regularities of the mechanism.
MineRobot addresses this gap by coupling an actuator-centered representation with a topology-aware kinematics pipeline tailored to the dominant closed-chain structures of underground mining robots.
\section{Parametric Modeling of Mining Robots} \label{sec:modeling}

In this section, we formalize a mining robot and its constituent elements—links, joints, and actuators—together with their attributes that collectively encode the robot’s geometric and topological information. We then introduce the Mining Robot Description Format (MRDF), which records this information in a structured form and serves as the primary information carrier of the MineRobot framework. Finally, we outline the matrix-form data structures obtained by parsing MRDF, which will be used to concisely specify the algorithms that follow.

\subsection{Robot Definition}

We consider a mining robot \(\mathcal{R}\) to be composed of links, joints, \hl{and actuators:} 
\begin{equation}
\mathcal{R} = ( \mathcal{L}, \mathcal{J}, \mathcal{A} ),
\end{equation}
where \(\mathcal{L} = \{L_i\}\), \(\mathcal{J} = \{J_i\}\), and \(\mathcal{A} = \{A_i\}\) denote the sets of links, joints, and actuators, respectively. Let \(n_L=|\mathcal{L}|\), \(n_J=|\mathcal{J}|\), and \(n_A=|\mathcal{A}|\). Figure~\ref{fig:joints-and-actuator} illustrates these three categories and their topological relations.

\begin{figure}[htbp]
  \centering
    \includegraphics[width=0.4\linewidth]{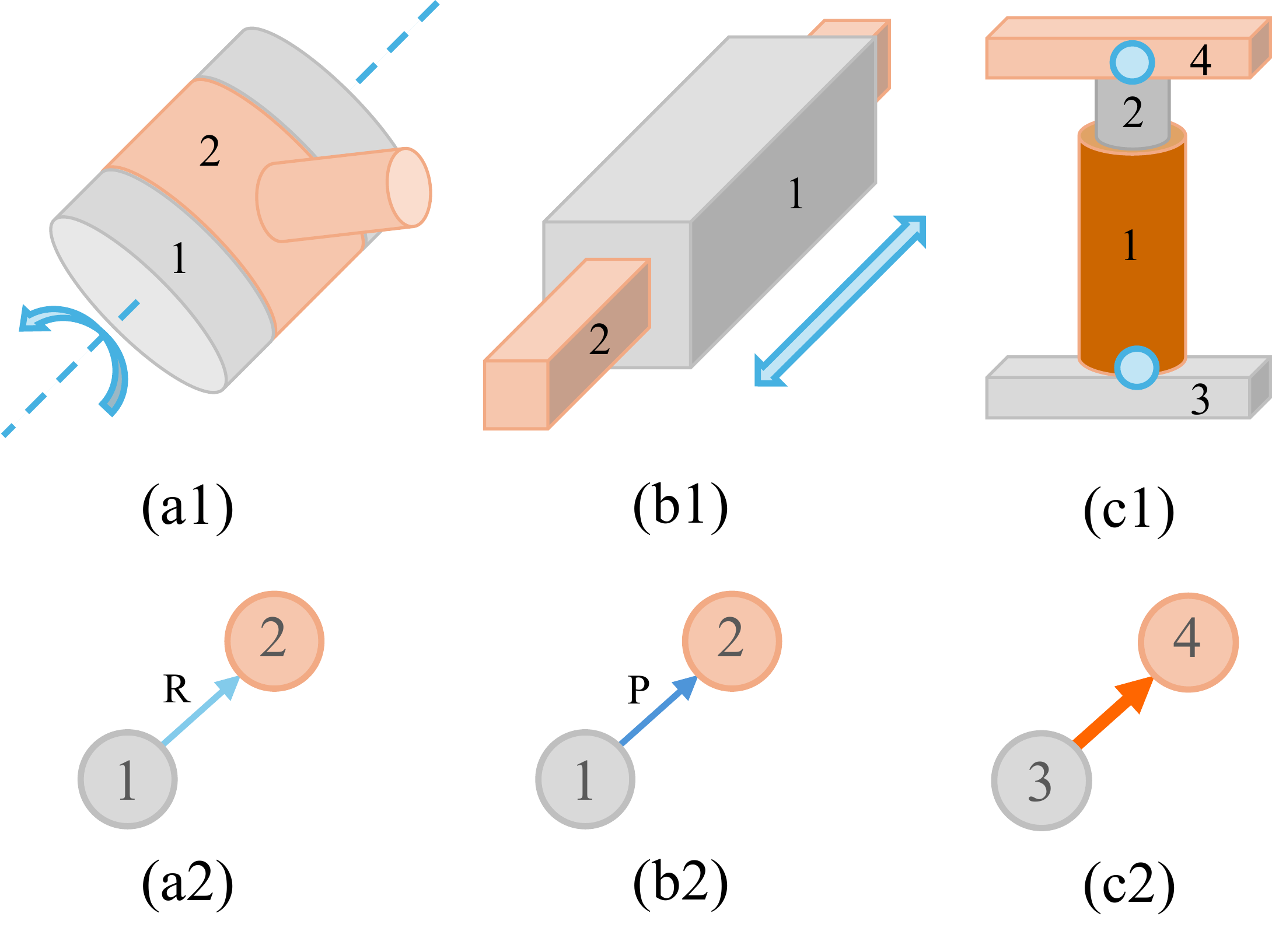}
    \caption{\hl{Illustrations} 
 of links, joints, actuators, and their directed topologies.
    (\textbf{a1},\textbf{a2}) show revolute and prismatic joint examples between links \(L_1\) and \(L_2\), and (\textbf{b1},\textbf{b2}) show their corresponding topological graphs.
    Links are represented as nodes, while revolute and prismatic joints are represented by light-blue and dark-blue directed arrows, respectively.
    (\hl{\textbf{c1},\textbf{c2}}
) shows a linear actuator whose tube \(L_1\) and rod \(L_2\) are mounted to parent links \(L_3\) and \(L_4\) through revolute joints; the actuator is represented by a thick orange directed arrow.}
    \label{fig:joints-and-actuator}
\end{figure}

Based on common underground platforms and industry practice, the mining robots considered in this paper satisfy the following assumptions:
\begin{enumerate}[label=(\roman*),align=parleft,leftmargin=*,labelsep=0.8mm]
  \item \label{asp:1} Actuation is provided by linear actuators, and each actuator redundancy group (defined in \textit{\hl{Actuators}} in this section) corresponds to one degree of freedom (DoF) of the~robot;
  \item \label{asp:2} The robot joints are restricted to revolute, prismatic, and fixed joints;
  \item \label{asp:3} The robot may contain planar four-bar linkages, and after contracting all such four-bars (defined in Section~\ref{sec:fourbar-generalized}), the joint topology becomes acyclic when all actuators are~ignored.
\end{enumerate}

This scope captures the dominant kinematic structures of major underground mining robots. Our contribution is a unified framework that enables automatic topology processing and kinematics computation within this scope, rather than a general-purpose solution intended to cover arbitrary robot mechanisms.

\textbf{\textit{\hl{Links}
}}.
$L_i=(i,\mathbf{T},V)$ denotes a link, modeled as an ideal rigid body with 6-DoF in three-dimensional (3D) space.
The index $i\in\{0,\dots,n_L-1\}$ is the unique identifier; by convention, the base link is $L_0$, whose parent is the world coordinate system. The links are connected in pairs by joints, directionally. 
Topologically, the closer a component is to $L_0$, the smaller its ID is.
We use $\mathbf{T}_i\in \mathrm{SE}(3)$ to denote the transformation of \emph{\hl{link}} $L_i$ relative to its parent.
Figure~\ref{fig:transformations} illustrates the local transformations associated with the directed joint \(J_{ij}\) between \(L_i\) and \(L_j\).
The component $V$ stores the visual model (geometry and its local offset), detailed in Section~\ref{sec:mrdf}.

\textbf{\textit{\hl{Joints}}}.
$J_i = (i, p, c, \tau, \mathbf{a}, \mathbf{T}^o )$ denotes a joint with unique identifier \(i\in\{0,\dots,n_J-1\}\).
A joint is a kinematic constraint that directionally connects a parent link \(L_p\) to a child link \(L_c\). 
For notational convenience we also use \(J_{pc}\) to refer to the joint from \(L_p\) to \(L_c\). 
The joint type \(\tau\in\{P,R,F\}\) specifies prismatic (\(P\)), revolute (\(R\)), or fixed (\(F\)); links connected by a fixed joint can be merged topologically. Figure~\ref{fig:joints-and-actuator}(a1,a2) show examples of \(R\) and \(P\) from \(L_1\) to \(L_2\); their directed topologies are indicated in light/dark blue.

\begin{figure}[htbp]
  \centering
        \includegraphics[width=0.4\linewidth]{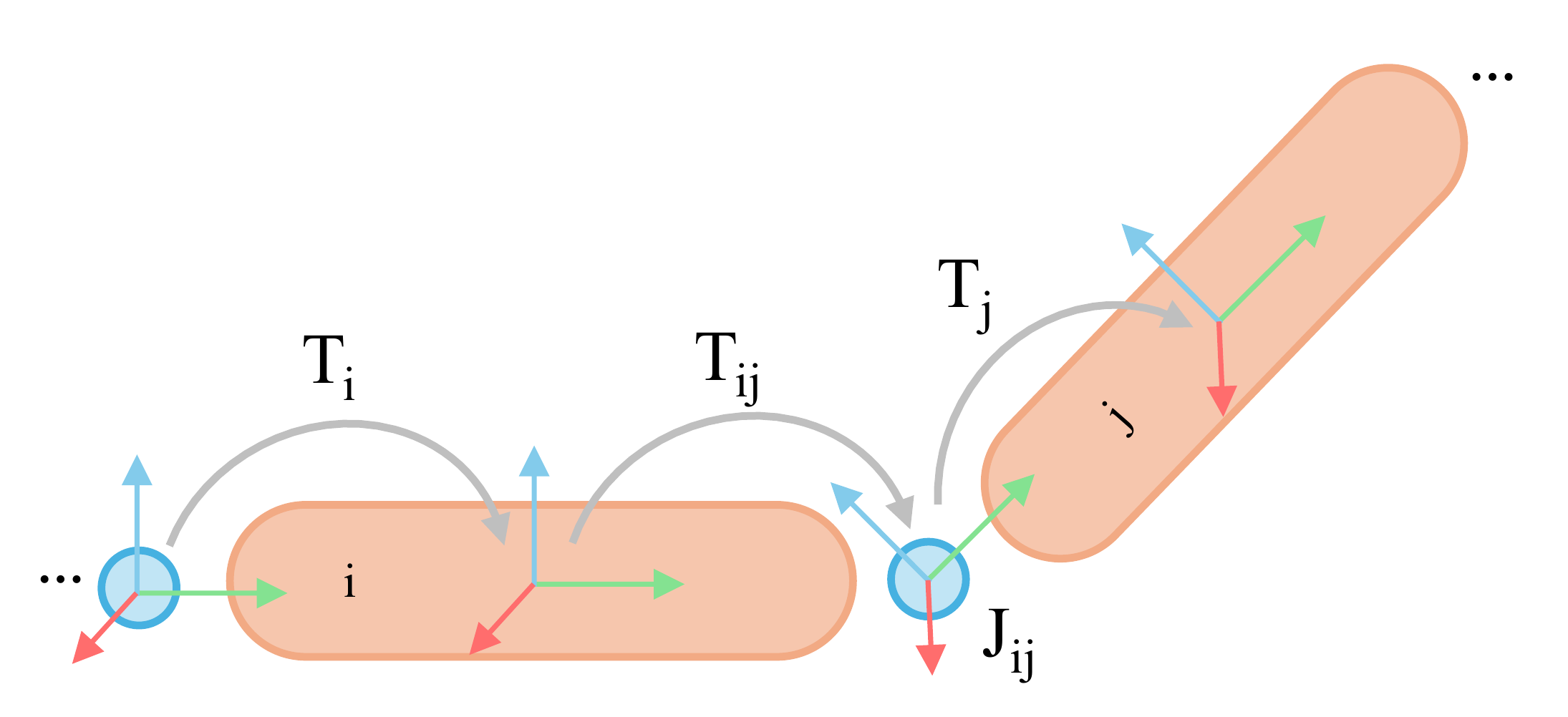}
    \caption{\hl{Transformations} 
 between connected links and joints.
 between connected links and joints.
    Link \(L_i\) and link \(L_j\) are connected by joint \(J_{ij}\).
    The transformation of \(L_i\) relative to its parent joint, the joint-origin transformation of \(J_{ij}\) relative to \(L_i\), and the transformation of \(L_j\) relative to \(J_{ij}\) are denoted by \(\mathbf{T}_i\), \(\mathbf{T}_{ij}\), and \(\mathbf{T}_j\), respectively.}
    \label{fig:transformations}
\end{figure}

$\mathbf{a} \in \mathbb{R}^3$ is a normalized direction vector, which represents the rotation axis of the revolute joint or the direction of the prismatic joint (see the blue dotted line in Figure~\ref{fig:joints-and-actuator}(a1) and the blue arrow in Figure~\ref{fig:joints-and-actuator}(a2)).
The constant \(\mathbf{T}^o\in \mathrm{SE}(3)\) is the joint’s origin transform relative to its parent link \(L_p\). Let \(\theta\) be the joint parameter (angle for \(R\), displacement for \(P\)). 
We adopt the following parent-to-child transform:
\begin{equation}\label{eq:tpc}
\mathbf{T}_{pc}(\theta)=\mathbf{T}(\mathbf{a},\theta)\,\mathbf{T}^o,
\end{equation}
where $\mathbf{T}(\mathbf{a},\theta)=
\begin{cases}
\left[\begin{smallmatrix}\mathbf{I}&\theta\,\mathbf{a}\\ \mathbf{0}^\top&1\end{smallmatrix}\right], & \tau=P,\\
\left[\begin{smallmatrix}\mathbf{R}(\mathbf{a},\theta)&\mathbf{0}\\ \mathbf{0}^\top&1\end{smallmatrix}\right], & \tau=R,
\end{cases}$
and $\mathbf{R}(\mathbf{a},\theta)=\mathbf{I}+\sin\theta\,[\mathbf{a}]_\times+(1-\cos\theta)\,[\mathbf{a}]_\times^2$ (Rodrigues’ formula), with $[\mathbf{a}]_\times$ the skew-symmetric matrix of $\mathbf{a}$.
For \(\tau=F\), \(\mathbf{T}_{pc}=\mathbf{T}^o\).
Figure~\ref{fig:transformations} illustrates the transformation $\mathbf{T}_{ij}$ of the joint $J_{ij}$ between $L_i$ and $L_j$.
Given a path from \(L_s\) to \(L_t\) through intermediate links \(L_{l_0},\dots,L_{l_{k-1}}\) with joints \(J_{s l_0},J_{l_0 l_1},\dots,J_{l_{k-1} t}\), the cumulative transform is
\begin{equation}\label{eq:Tpq}
\mathbf{T}_{s t}=\mathbf{T}_{s l_0}\!\left(\,\prod_{i=0}^{k-2}\mathbf{T}_{l_i l_{i+1}}\,\right)\mathbf{T}_{l_{k-1} t},
\end{equation}
with the product ordered along the path direction.

\textbf{\textit{\hl{Actuators}}}.
\(A_i\ = (i, \mathtt{t}, \mathtt{r}, b, Rd) \) represents an actuator of the robot. 
An actuator produces telescopic motion under hydraulic or pneumatic pressure. 
It consists of a \emph{\hl{tube}} link and a \emph{rod} link:
\(\mathtt{t}=(t,p_t)\), \(\mathtt{r}=(r,p_r)\),
where $t$ and $r$ are the tube/rod link IDs and $p_t$, $p_r$ are their respective parent links. The tube and rod mount to their parents via revolute joints $J_{p_t t}$ and $J_{p_r r}$. 
We depict the actuator topology by an orange arrow from $L_{p_t}$ to $L_{p_r}$ (Figure~\ref{fig:joints-and-actuator}(c1,c2)). 
The component \(b=[\ell_i,u_i]\) specifies the lower and upper bounds of the actuator length \(l_i\), where \(l_i\) is defined as the distance between the two mounting joint origins in world coordinates, i.e., the distance between \(J_{p_t t}\) and \(J_{p_r r}\).

The $Rd$ component is a set of actuator identifiers that represent the redundant actuators of the actuator.
To improve load capacity and stability, multiple actuators may synchronously realize the same motion~\cite{gosselin2018redundancy,merlet2006parallel, maloisel2023optimal}. We define an equivalence relation $\sim$ on $\mathcal{A}$: $A_i\sim A_j$ iff $A_i$ and $A_j$ are mutually redundant. Each equivalence class is a \emph{\hl{redundant actuator group}}. The cardinality of the quotient set \(\mathcal{A}/\!\sim\) equals the system DoF of \(\mathcal{R}\).
In mechanical design, redundant actuators are commonly arranged symmetrically and connected in parallel or series hydraulic circuits—for example, paired columns in hydraulic supports 
($A_0$–$A_1$ and $A_2$–$A_3$, detailed in Section~\ref{sec:experiments}) and paired cylinders for the rotary mechanism of roadheaders ($A_0$–$A_1$, detailed in Section~\ref{sec:experiments}) \cite{tian2018kinematic}.
These relations are recorded at modeling time and used later in topology processing and solver scheduling.

\subsection{Mining Robot Description Format} \label{sec:mrdf}

The tuple \( \mathcal{R} = ( \mathcal{L}, \mathcal{J}, \mathcal{A} ) \) specifies the kinematics content we require. We therefore design a simple, human- and machine-readable format—\emph{\hl{Mining Robot Description Format (MRDF)}}—that records links, joints, and actuators as first-class entities and serves as the primary input to the MineRobot framework.
A comparison with existing robot description formats is provided later in Section~\ref{sec:comparison-expressiveness}.
Unlike general-purpose formats such as URDF and SDF, MRDF explicitly models actuators because our solvers are actuator centered (Section~\ref{sec:solver}). 
We serialize MRDF using JavaScript Object Notation (JSON) because of its concise syntax and straightforward implementation.
Figure~\ref{fig:mrdf} lists representative MRDF pseudocode snippets for a link (Figure~\ref{fig:mrdf}a), a joint (Figure~\ref{fig:mrdf}b), an actuator (Figure~\ref{fig:mrdf}c), and the overall file layout (Figure~\ref{fig:mrdf}d).
The snippets are simplified and annotated for readability rather than provided as complete valid JSON files.
Detailed MRDF conventions are given in Appendix~\ref{appendix:mrdf-schema}.

\begin{figure}[htbp]
  \centering
        \includegraphics[width=0.8\linewidth]{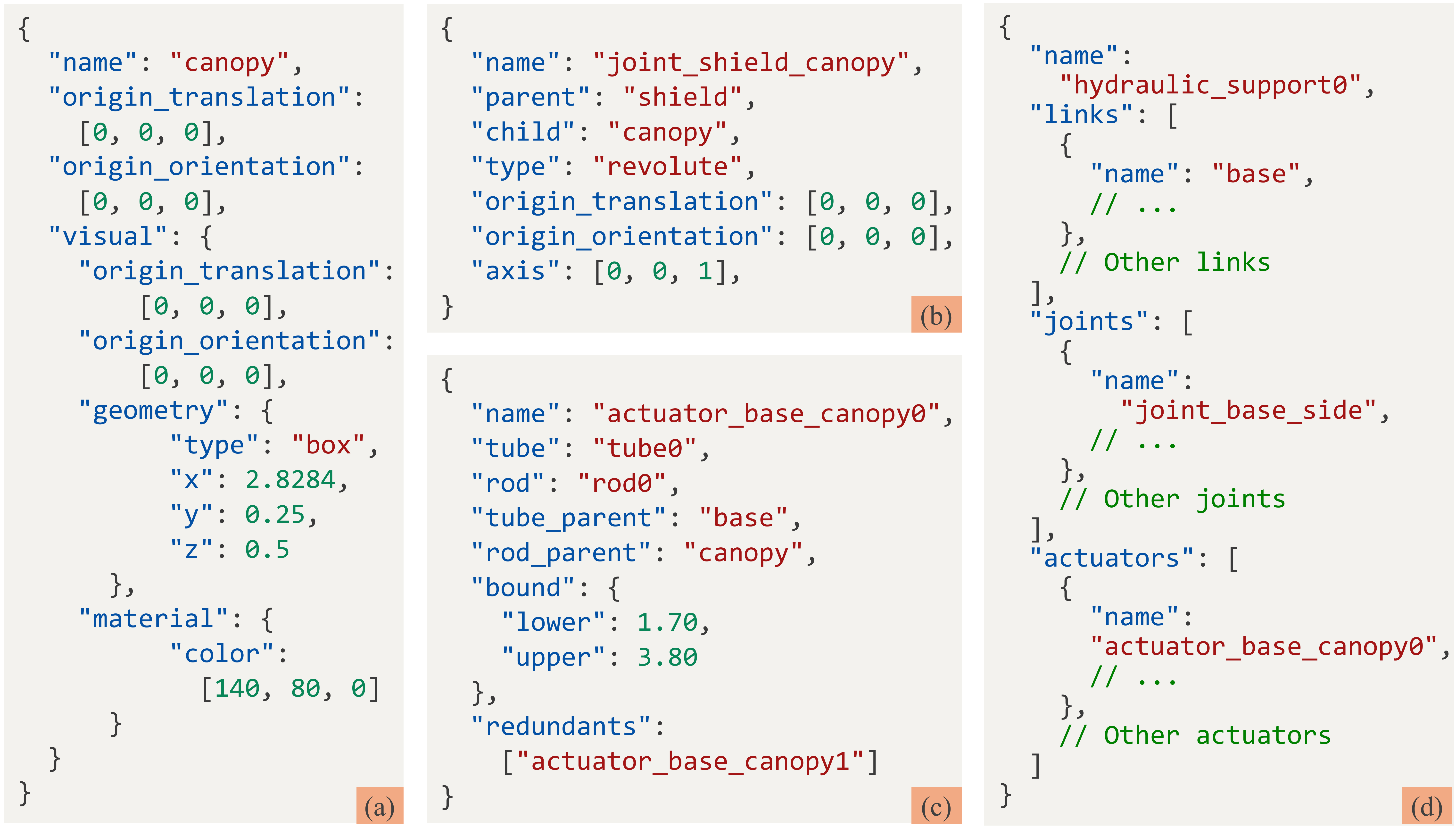}
    \caption{Representative MRDF pseudocode snippets and overall file structure. The snippets are simplified for readability and are not complete valid JSON files. (\textbf{a}) Link entry with name, pose, and visual fields. (\textbf{b}) Joint entry with parent/child, type, pose, and axis fields. (\textbf{c}) Actuator entry with tube/rod-side mounting information, stroke bounds, and redundancy information. (\textbf{d}) Overall MRDF file layout.}
    \label{fig:mrdf}
\end{figure}

A link entry stores the link name, origin transformation, and optional visual information. The orientation is represented by roll--pitch--yaw (RPY) angles~\cite{sciavicco2012modelling}. The visual field may specify primitive geometry and material parameters or refer to an external model asset such as OBJ or glTF. These fields map to the link tuple \(L_i=(i,\mathbf{T},V)\).
A joint entry stores the joint name, parent/child link names, joint type (\(P/R/F\)), origin transformation, and motion axis for \(P/R\) joints. These data correspond to \(J_i=(i,p,c,\tau,\mathbf{a},\mathbf{T}^o)\). During parsing, fixed joints (\(F\)) enable merging adjacent links before downstream processing.
An actuator entry records actuator identity, tube/rod-side mounting information, admissible stroke bounds \([\ell_i,u_i]\), and redundant-actuator relations. This mirrors \(A_i=(i,\mathtt{t},\mathtt{r},b,Rd)\), where \(b\) stores the actuator length bounds and \(Rd\) stores redundancy information. The bounds constrain the telescopic length \(l_i\), defined in world coordinates as the distance between the tube/rod mounting joint origins.

Although MRDF follows the common link--joint graph abstraction used by URDF/SDF-like robot descriptions, it is not intended as a JSON translation of those formats.
Its distinguishing feature is that linear actuators are modeled as first-class kinematic entities, with tube/rod-side mounting information, stroke bounds, and redundancy relations that are consumed directly by four-bar contraction, ITEP extraction, FK scheduling, and actuator-length IK optimization.
The current MRDF version is therefore a domain-specific description format for the actuator-centered mining-robot pipeline; arbitrary spatial closed chains, general multi-loop mechanisms, and unsupported joint or constraint types remain outside the current scope.

In preprocessing, the framework reads MRDF, assigns unique identifiers to all names, merges links connected by fixed joints, and converts translation/orientation fields into \(\mathrm{SE}(3)\) matrices.
For concise algorithm specification, we maintain the following matrices. \\
\noindent\textit{\hl{Joint matrix.}
} \(\boldsymbol{J} \in \mathbb{R}^{n_L \times n_L}\) encodes directed joint types from \(L_i\) to \(L_j\):
\begin{equation}
    \boldsymbol{J}_{i,j}=
        \begin{cases}
        1, & \text{$J_{i,j}$ is revolute},\\
        2, & \text{$J_{i,j}$ is prismatic},\\
        3, & \text{$J_{i,j}$ is fixed},\\
        4, & \text{$J_{i,j}$ is generalized},\\
        0, & \text{otherwise}.
        \end{cases}
\end{equation}

\hl{We use} \(1/2/3/4\) for \(R/P/F/G\), and \(0\) when no joint exists from \(L_i\) to \(L_j\). The value \(4\) indicates a \emph{generalized joint}, which represents a contracted four-bar linkage (Section~\ref{sec:topo}). \\
\noindent\textit{\hl{Actuator incidence.}} \(\boldsymbol{At},\,\boldsymbol{Ar}\in \mathbb{R}^{n_A\times n_L}\) mark tube/rod links per actuator (rows index actuators, columns index links):
\begin{equation}
    \boldsymbol{At}_{i,j},\ \boldsymbol{Ar}_{i,j}=
        \begin{cases}
        1, & \text{$L_j$ is the tube/rod link of $A_i$},\\
        0, & \text{otherwise}.
        \end{cases}
\end{equation} \\

\noindent\textit{\hl{Redundancy matrix.}} 
\(\boldsymbol{Rd}\in \mathbb{R}^{n_A\times n_A}\) records actuator redundancy:
\begin{equation}
    \boldsymbol{Rd}_{i,j}=
        \begin{cases}
            1, & \text{$A_j$ is redundant to $A_i$},\\
            0, & \text{otherwise}.
        \end{cases}
\end{equation}

$\boldsymbol{Rd}_{i,j}$ is $1$ if actuator $A_j$ is redundant for $A_i$ and $0$ otherwise. By convention, $\boldsymbol{Rd}_{i,i}=1$.
These matrices, together with per-link transforms and joint origins, constitute the minimal kinematics state used by our topology processing (four-bar contraction and ITEP extraction) and by the kinematics solvers described later.

\section{Topology Construction Algorithms} \label{sec:topo}

In this section, we introduce a two-stage topology-construction pipeline.
We first detect four-bar linkages and \emph{\hl{contract}} each into a \emph{\hl{generalized joint}} while preserving its connectivity to the rest of the robot (Algorithm~\ref{alg:generalized-joints}).
Under the scope assumptions in Section~\ref{sec:modeling}, after contraction, the joint topology becomes acyclic with actuators ignored, ensuring that subsequent path searches are well defined.
On this contracted topology, for every actuator we extract its \emph{\hl{Independent Topologically Equivalent Path (ITEP)}}—the ordered set of links and joints influenced when all other actuators are fixed (Algorithm~\ref{alg:topo})—and assign each ITEP to one of four canonical types.
These ITEPs serve as inputs to the kinematics solver in Section~\ref{sec:solver}.

\subsection{Four-Bar Linkage and Generalized Joints}
\label{sec:fourbar-generalized}


A four-bar is a closed-loop mechanism composed of a fixed ground link and three moving links connected in sequence by revolute joints~\cite{norton2007design,mccarthy2010geometric}.
From a mechanical-design perspective, four-bar linkages are classical transmission modules for motion transformation, path generation, and compact force transmission in machinery~\cite{norton2007design,mccarthy2010geometric,ebrahimi2015efficient}.
In this work, the planar assumption is made at the nominal rigid-body mechanism-design level, rather than as a claim that physical mining equipment is perfectly planar under load.
Its engineering basis is that many representative underground mining mechanisms, especially hydraulic-support linkages and other heavy-duty linkage modules, are designed with approximately parallel revolute axes and a dominant working plane~\cite{prebil2002synthesis,guan2019dynamic}.
In such mechanisms, planar four-bar linkages provide a practical trade-off among motion-path shaping, load capacity, stiffness, manufacturability, maintainability, and predictable operation in confined environments~\cite{prebil2002synthesis,guan2019dynamic}.
Actual equipment may still deviate from this idealization because of factors such as manufacturing tolerances, pin clearance and wear, elastic deformation, asymmetric loads, and sensing or calibration errors~\cite{ge2020virtual,mei2025sensing,wang2024method}.
These effects are not captured by the present nominal kinematic model and require calibration, sensing compensation, higher-fidelity multibody modeling, or physical validation.
Mechanisms with significant spatial multi-loop motion, non-parallel joint-axis coupling, or a large deformation that invalidates this planar rigid-body abstraction fall outside the current MineRobot scope.
Accordingly, we focus on the planar, single-DoF, reciprocating four-bars that are prevalent in our target mechanism class, and exclude atypical spatial closed-chain designs, link-interference cases, and singular or branch-switching configurations.

Formally, a four-bar is the ordered quadruple
\[
Fb=(L_a,L_b,L_c,L_d),
\]
where links are connected by directed revolute joints \((J_{ab},J_{bc},J_{cd})\), and the loop-closure between \(L_d\) and \(L_a\) is encoded by a fixed pair \((J_{da},J_{ad})\).
The same directed convention is used in MRDF.
Figure~\ref{fig:four-bar} illustrates the structure and its corresponding topological graph (with \(J_{ad}\) omitted for clarity).
By the Chebychev--Grübler--Kutzbach criterion~\cite{mccarthy2010geometric}, the DoF of a four-bar is 1.
Without loss of generality, we designate \(J_{ab}\) as the \emph{\hl{designated input joint}} with parameter \(\theta\).
Given \(\theta\), and the remaining joint variables \(\alpha\) and \(\beta\) (at \(J_{bc}\) and \(J_{cd}\)) are uniquely determined; their computation is detailed in Section~\ref{sec:solver}.

\begin{figure}[htbp]
  \centering
        \includegraphics[width=0.7\linewidth]{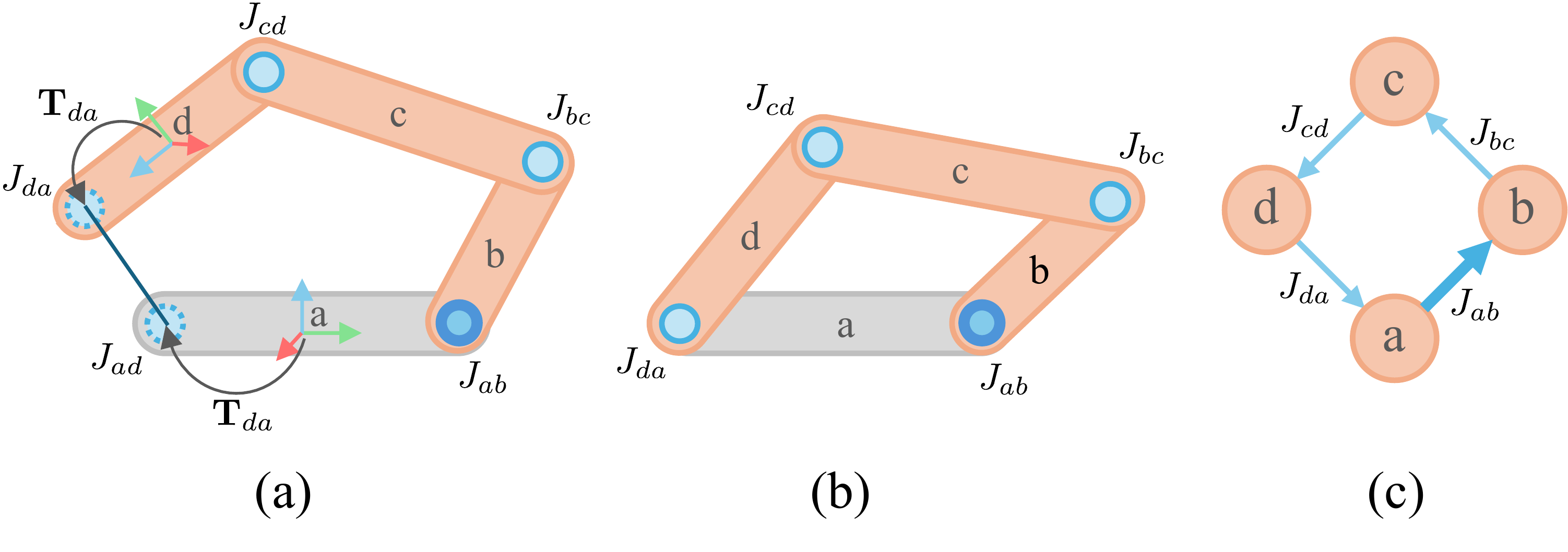}
    \caption{\hl{Four-bar} 
 linkage.
    (\textbf{a}) Links \(L_a, L_b, L_c, L_d\) are connected in sequence by revolute joints \(J_{ab}, J_{bc}, J_{cd}\);
    the pair \(J_{da}\)–\(J_{ad}\) encodes the closure between \(L_d\) and \(L_a\).
    (\textbf{b}) Same geometry with \(J_{ad}\) omitted.
    (\textbf{c}) Corresponding topological graph.
    The ground link \(L_a\) is stationary. We treat \(J_{ab}\) (bold) as the \emph{\hl{designated input joint}}; prescribing its rotation \(\theta\) uniquely determines the linkage configuration.
    }
    \label{fig:four-bar}
\end{figure}

A four-bar acts as a single-DoF transmission from the ground link \(L_a\) to the remaining members of \(Fb\).
We abstract this transmission by a \emph{\hl{generalized joint}} (type~4 in the joint matrix \(\boldsymbol{J}\); abbreviated as \(G\)), which denotes a single-DoF map from \(L_a\) to a target link within the same four-bar.
This abstraction operation is called \emph{\hl{contraction}}.
Algorithm~\ref{alg:generalized-joints} contracts each four-bar into generalized joints while preserving its incidence to the rest of the mechanism; it takes the joint matrix \(\boldsymbol{J}\) and actuator incidence \(\boldsymbol{At},\boldsymbol{Ar}\) as input and updates \(\boldsymbol{J}\) in place.

We first compute strongly connected components (SCCs) using Tarjan’s classic algorithm~\cite{tarjan1972depth} to obtain all four-bars $\mathcal{F}$ (Step~1), and discard trivial two-node SCCs introduced solely by auxiliary closure joints (e.g., \(J_{da}\)–\(J_{ad}\) in Figure~\ref{fig:four-bar}).
Because four-bars are modeled with a consistent ordering, each extracted \(Fb\) follows the sequence in Figure~\ref{fig:four-bar}c.

This contraction is a topology-level abstraction, not a removal of the four-bar kinematic relation.
A generalized joint \(G\) records that a target member of the four-bar is connected to the ground member through a planar 1-DoF transmission.
During FK/IK solving, this transmission is evaluated by the corresponding four-bar closure solver; therefore, the internal loop constraint is retained in the kinematic computation rather than discarded.
The contraction preserves the external incidence and actuator-relevant motion transmission required by the MineRobot solver, but it is not claimed to be a universal kinematic equivalence for arbitrary spatial or multi-DoF closed chains.
Under the planar 1-DoF four-bar assumption, prescribing the designated input joint uniquely determines the remaining four-bar configuration within the selected assembly branch. Therefore, replacing the four-bar topology by a generalized joint preserves the single-DoF transmission relation needed to propagate motion from the ground member to the exposed member.

\begin{algorithm}[htbp]
\small
\caption{Contracting four-bar linkages into generalized joints.}
\label{alg:generalized-joints}
\KwIn{Joint matrix $\boldsymbol{J}$; actuator incidence $\boldsymbol{At},\boldsymbol{Ar}$}
\KwOut{Updated $\boldsymbol{J}$ with generalized joints ($G$)}
$\mathcal{F}\gets \textsc{getFourbars}(\boldsymbol{J})$\;
\ForEach{$Fb=[L_a,L_b,L_c,L_d]\in\mathcal{F}$}{
  \tcp{\hl{remove} 
 internal four-bar joints}
  $\boldsymbol{J}_{a,b},\boldsymbol{J}_{b,c},\boldsymbol{J}_{c,d},\boldsymbol{J}_{d,a},\boldsymbol{J}_{a,d}\gets 0$;

  \tcp{add $G$ joints to keep external incidence}
  \For{$x\in\{b,c,d\}$}{
    \If(\tcp*[f]{has external children}){$\sum_j \boldsymbol{J}_{x,j}>0$}{ $\boldsymbol{J}_{a,x}\gets 4$ }
    \If(\tcp*[f]{is tube/rod parent}){$\exists i:\boldsymbol{At}_{i,x}=1 \vee \boldsymbol{Ar}_{i,x}=1$}{ $\boldsymbol{J}_{a,x}\gets 4$ }
  }

  \tcp{if an actuator spans two non-ground links}
  \ForEach{actuator $A_i$ with parents $(p_t,p_r)\in\{b,c,d\}$}{
    $\boldsymbol{J}_{p_t,p_r}\gets 4$\;
  }
}
\end{algorithm}

Given \(Fb=(L_a,L_b,L_c,L_d)\), the contraction operator acts on \(\boldsymbol{J}\) via the following minimal yet sufficient rules:
(i) \textbf{\hl{Remove loop-internal joints.} 
} Delete all joints among \(\{L_a,L_b,L_c,L_d\}\) (Step~3).
(ii) \textbf{\hl{Expose ground-to-member transmission.}} For each non-ground member \(L_x\in\{L_b,L_c,L_d\}\), if \(L_x\) has any child outside \(Fb\) (Steps~5--6) or \(L_x\) is a parent of any actuator endpoint (tube or rod; Steps~7--8), insert a generalized joint from \(L_a\) to \(L_x\).
(iii) \textbf{\hl{Respect actuator-directed coupling inside} \(Fb\).} If an actuator attaches between two non-ground members \(L_u,L_v\in\{L_b,L_c,L_d\}\) (e.g., tube parent at \(L_u\) and rod parent at \(L_v\)), add a generalized joint consistent with the actuator's direction (Steps~9--10).
Intuitively, (i) removes the cycle that enforces loop closure, while (ii)–(iii) reintroduce precisely those single-DoF transmission edges needed to propagate motion from the ground and to encode actuator-induced coupling among non-ground members.
Consequently, contraction does not create generalized joints between arbitrary pairs of non-ground four-bar members; a non-ground-to-non-ground generalized joint is added only when such a pair is directly coupled by an actuator before contraction as in Figure~\ref{fig:contraction}(b1,b2).

\begin{figure}[htbp]
  \centering
        \includegraphics[width=1.0\linewidth]{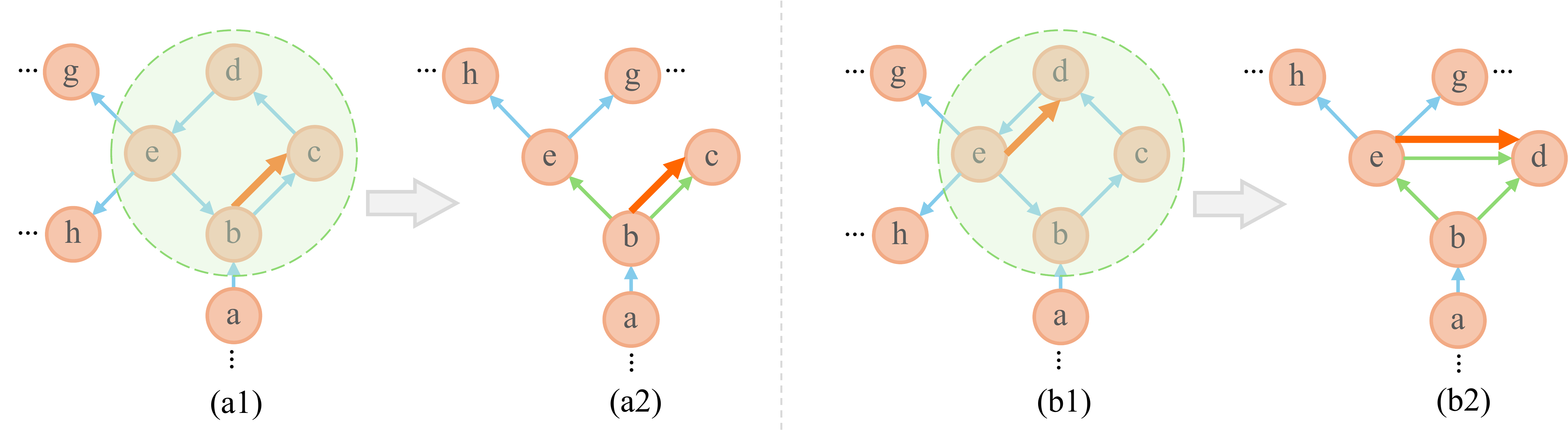}
    \caption{\hl{Contraction} 
 of four-bars. (\textbf{a1}) and (\textbf{b1}) show two representative four-bar configurations before contraction; 
(\textbf{a2}) and (\textbf{b2}) show the corresponding contracted topologies. 
A four-bar \(Fb=(L_b, L_c, L_d, L_e)\) (green) and its adjacent links; here the ground is \(L_b\) (\textbf{a1},\textbf{b1}). First, all joints inside \(Fb\) are removed. Then, if \(L_x\in\{L_c,L_d,L_e\}\) has external children (e.g., \(L_e\) in \textbf{a1},\textbf{b1}) or is an actuator parent (e.g., \(L_c\) in \textbf{a1}; \(L_d\) and \(L_e\) in \textbf{b1}), a generalized joint from \(L_b\) to \(L_x\) is added (green arrows). Finally, if an actuator attaches between \(L_u,L_v\in\{L_c,L_d,L_e\}\) (e.g., \(L_d\) and \(L_e\) in \textbf{b1}), a generalized joint between \(L_u\) and \(L_v\) is also added, consistent with the actuator’s direction.}
    \label{fig:contraction}
\end{figure}

Algorithm~\ref{alg:generalized-joints} implements these rules over all detected four-bars and updates \(\boldsymbol{J}\) in place.
With adjacency-based bookkeeping, four-bar contraction runs in linear time in the size of the robot graph and actuator set, i.e., $\mathcal{O}(n_L+n_J+n_A)$ time and $\mathcal{O}(n_L+n_J+n_A)$ space.
Figure~\ref{fig:contraction} illustrates typical cases using two examples.
After contracting all four-bars, the joint graph (actuators excluded) becomes a directed acyclic graph (DAG), laying the groundwork for Algorithm~\ref{alg:topo}.

\subsection{Independent Topologically Equivalent Paths (ITEPs)} \label{sec:topo-algo}

We observe that each actuator redundancy group corresponds to one independent DoF of the robot; with all other groups locked, the ordered chain of links and joints influenced by that group defines its \emph{\hl{Independent Topologically Equivalent Path}} (ITEP). 
Formally, let $\mathcal{A}$ be the set of actuators and $\sim$ the redundancy relation recorded in MRDF; each equivalence class $\mathcal{G}_i\subset\mathcal{A}$ contributes one DoF.
Choosing a representative $A_i\in\mathcal{G}_i$, we construct its ITEP $P_i$ while \emph{\hl{locking}} all actuators in $\mathcal{A}\setminus\mathcal{G}_i$. 
We categorize ITEPs into four types (A–D). Algorithm~\ref{alg:topo} then extracts, for each actuator, the associated ITEP from the contracted topology, providing the inputs to the per-type kinematics solvers in Section~\ref{sec:solver}.

Algorithm~\ref{alg:topo} takes as input the contracted joint graph $\boldsymbol{J}$, actuator incidence $\boldsymbol{At},\boldsymbol{Ar}$, and the redundancy matrix $\boldsymbol{Rd}$; it outputs the ITEP $P_i$ for each actuator $A_i$.
We maintain a marking vector $  \mathbf{M} \in \mathbb{R}^{n_A \times 1} $ to record whether an actuator has been processed ($\boldsymbol{M}_i=1$) or not ($\boldsymbol{M}_i=0$) (Step~1).
For each unprocessed $A_i$, we create a working copy $\boldsymbol{J}'$ of the joint matrix $\boldsymbol{J}$ (Step~5) and then \emph{lock} all actuators outside its redundancy group $\mathcal{G}_i$ (Steps~6--21).

For each other actuator $A_j$, we retrieve the tube link $L_t$, the rod link $L_r$, and their parent links $L_{p_t}$ and $L_{p_r}$ via the auxiliary function $\Call{getActuatorLinks}{}$ (Step~9).
We define a path $P$ in a directed graph as the ordered sequence of nodes traversed from one node to another, including both endpoints.
Let $|P|$ denote the node count of $P$, and let $[\,]$ denote the empty path.
For actuator $A_j$, extract the path from $L_{p_t}$ to $L_{p_r}$ (Step~10) or $L_{p_r}$ to $L_{p_t}$ (Step~17), whichever exists.
After Algorithm~\ref{alg:generalized-joints}, if configurations like Figure~\ref{fig:contraction}(b2) are absent, the robot’s topology forms a tree with unique paths between any two connected nodes.
Even if such cases exist, the mechanical design avoids actuators driving already actuated four-bar linkages, ensuring that no path passes through actuated generalized joints. 
Thus, the path between two connected nodes is unique and can be retrieved using a standard depth-first or breadth-first search through the auxiliary function $\Call{getTopoPath}{}$.

\begin{algorithm}[htbp]
\small
\caption{Topology construction algorithm for actuators.} 
\label{alg:topo}
\KwIn{Joint matrix $\boldsymbol{J}$; actuator incidence $\boldsymbol{At}, \boldsymbol{Ar}$; redundancy matrix $\boldsymbol{Rd}$}
\KwOut{ITEP $P_i$ for each actuator $A_i$}

$\boldsymbol{M} \gets \mathbf{0}^{n_A \times 1}$ \tcp{0=unprocessed, 1=processed}

\For{$i \gets 0$ \KwTo $n_A - 1$}{
    \If{$\boldsymbol{M}_i \neq 0$}{
        \textbf{continue}
    }

    $\boldsymbol{J'} \gets \boldsymbol{J}$

    \tcp{lock all other actuators}
    \For{$j \gets 0$ \KwTo $n_A - 1$}{ 
        \If{$\boldsymbol{Rd}_{i,j} \neq 0$}{
            \textbf{continue}
        }

        $t, r, p_t, p_r \gets \textsc{getActuatorLinks}(\boldsymbol{At}, \boldsymbol{Ar}, \boldsymbol{J}, j)$ \\
        $P \gets \textsc{getTopoPath}(\boldsymbol{J'}, p_t, p_r)$

        \If{$P \neq [\,]$}{
            \If{$|P| = 2$}{
                $\boldsymbol{J'}_{p_t, p_r} \gets 3$
            }
            \Else{
                $\boldsymbol{J'}_{r,t} \gets 3$; \quad $\boldsymbol{J'}_{t,p_t} \gets 1$; \quad $\boldsymbol{J'}_{p_t,t} \gets 0$
            }
        }
        \Else{
            $P \gets \textsc{getTopoPath}(\boldsymbol{J'}, p_r, p_t)$ \\
            \If{$|P| = 2$}{
                $\boldsymbol{J'}_{p_r, p_t} \gets 3$
            }
            \Else{
                $\boldsymbol{J'}_{t,r} \gets 3$; \quad $\boldsymbol{J'}_{r,p_r} \gets 1$; \quad $\boldsymbol{J'}_{p_r,r} \gets 0$
            }
        }
    }

    $p_t, p_r \gets \textsc{getActuatorLinks}(\boldsymbol{At}, \boldsymbol{Ar}, \boldsymbol{J}, i)$ \\
    $P_a \gets \textsc{getTopoPath}(\boldsymbol{J'}, p_t, p_r)$ \\
    $P_b \gets \textsc{getTopoPath}(\boldsymbol{J'}, p_r, p_t)$ \\
    $P_i \gets \textsc{mergePaths}(P_a, P_b)$ 

    \For{$j \gets 0$ \KwTo $n_A - 1$}{
        \If{$\boldsymbol{Rd}_{i,j} \neq 0$}{
            $P_j \gets P_i$; $\boldsymbol{M}_j \gets 1$ 
        }
    }
}
\end{algorithm}

The locking operation for every other actuator $A_j$ is as follows.
If a path $P$ exists from $L_{p_t}$ to $L_{p_r}$ (Step~11), we distinguish two cases.
\textbf{\hl{(i) Direct connection.}}
If $L_{p_t}$ and $L_{p_r}$ are connected by a single-DoF joint $J_{p_t p_r}$ ($R/P/G$), as in Figure~\ref{fig:topo-build-other-actuators}(a1).
Fixing the rod–tube joint $J_{rt}$ is kinematically equivalent to locking $J_{p_t p_r}$; therefore we directly set $J_{p_t p_r}$ to fixed (Step~13; Figure~\ref{fig:topo-build-other-actuators}(a3)).
\textbf{\hl{(ii) Indirect connection.}} 
Otherwise, the tube/rod parents are connected by a directed path $P$ with $|P|>2$.
Locking $A_j$ closes a local loop that is topologically equivalent to merging the tube and rod into a single link connected to $L_{p_t}$ and $L_{p_r}$ by two revolute joints (Figure~\ref{fig:topo-build-other-actuators}(b1--b3)).
Under our scope assumptions, such a loop is necessarily 1-DoF and unique; moreover, a mobility-counting argument implies $|P|=3$, yielding a four-link generalized four-bar (detailed in Appendix~\ref{appendix:A}).
In our mining-robot models, this indirect pattern appears as a contracted transmission ($G$) adjacent to a revolute joint ($R$).
We fix $J_{rt}$ and reverse $J_{t p_t}$ to close a local loop (Step~15; Figure~\ref{fig:topo-build-other-actuators}(b2)), which is topologically equivalent to merging $L_r$ and $L_t$ into a single node (Figure~\ref{fig:topo-build-other-actuators}(b3)). 
By the Chebychev–Gr{\"u}bler–Kutzbach criterion, the resulting loop has DoF~1.
Conversely, if instead the path exists in the reverse direction, i.e., from $L_{p_r}$ to $L_{p_t}$, we apply the symmetric operations with the joint directions swapped (Steps~16--21).

\begin{figure}[htbp]
  \centering
    \includegraphics[width=0.5\linewidth]{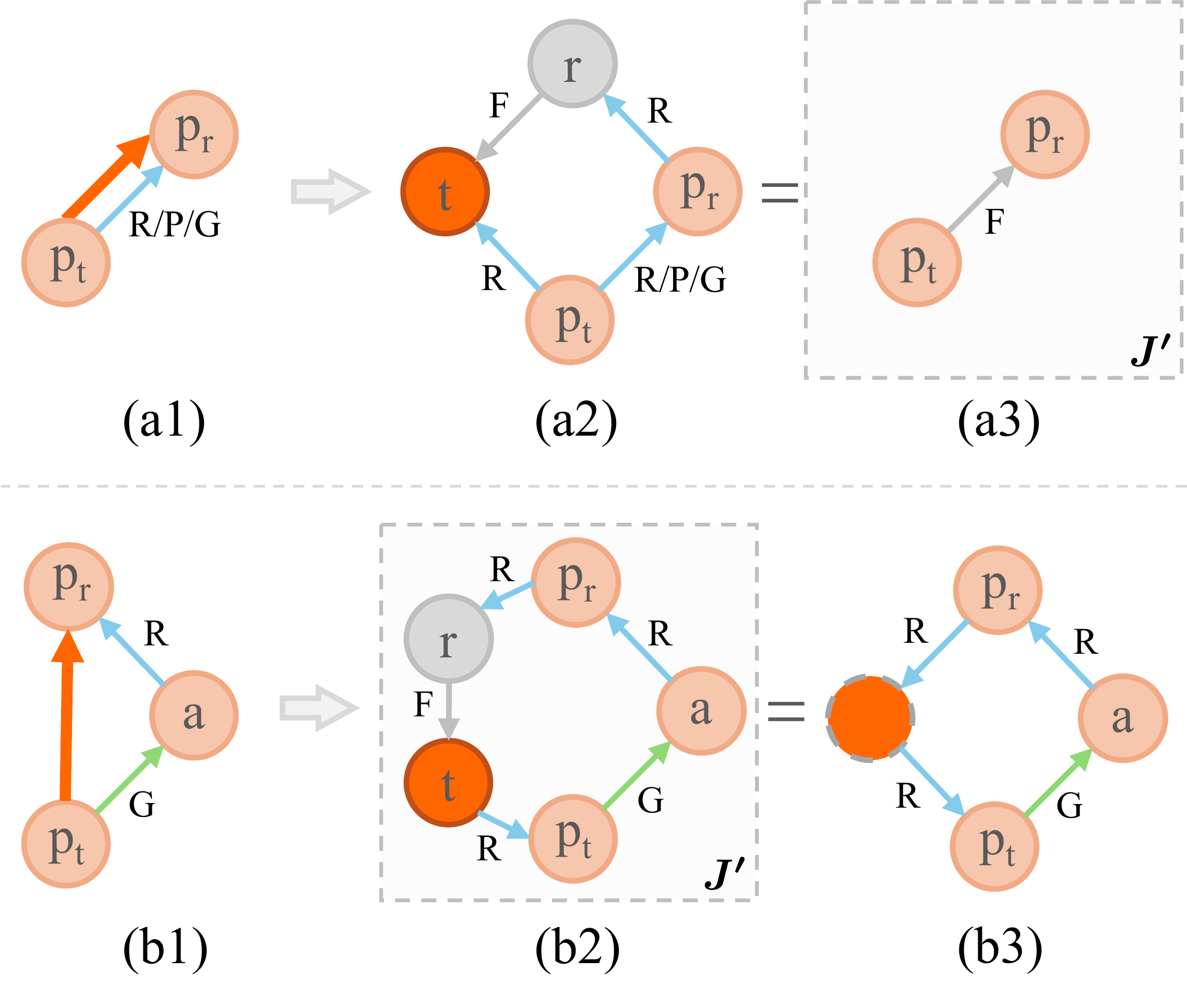}
  \caption{\hl{Locking} 
 other actuators (schematics). 
  (\textbf{a1}--\textbf{a3}) If a direct edge between the tube/rod parents exists, set it to \(F\).
  (\textbf{b1}--\textbf{b3}) Otherwise, fix \(J_{tr}\) and reverse the adjacency to form and collapse a 1–DoF loop, which is equivalent to merging the tube and rod.}
  \label{fig:topo-build-other-actuators}
\end{figure}

After locking all other actuators, we extract the links and joints influenced by $A_i$, namely its ITEP. 
Let $L_{p_t}$ and $L_{p_r}$ denote the parent links of the tube and rod of $A_i$, respectively (Step~22). 
The paths from $L_{p_t}$ to $L_{p_r}$ and from $L_{p_r}$ to $L_{p_t}$ are obtained and are denoted as $P_a$ and $P_b$, respectively (Steps~23--24).
If $\boldsymbol{J}'$ contains no local closed-loop configuration like Figure~\ref{fig:topo-build-other-actuators}(b2), exactly one of $P_a$ or $P_b$ is nonempty, yielding the unique simple path between $L_{p_t}$ and $L_{p_r}$. 
If such loops exist, each has one DoF, and since all other actuators are locked, the robot’s overall DoF is also 1. 
Therefore, there can be only one such loop, and $L_{p_t}$ and $L_{p_r}$ must lie on it.
The ITEP $P_i$ of actuator $A_i$ is obtained by merging $P_a$ and $P_b$ into a single path via the auxiliary function $\Call{mergePaths}{}$ (Step~25).
Given $P_a = [i, \dots, j]$ and $P_b = [j, \dots, i]$, the merged path is $P = [i, \dots, j, \dots, i]$.
Finally, because mutually redundant actuators share the same topology, $P_i$ is assigned to all $A_j\in\mathcal{G}_i$ (Steps~26--27).
With an adjacency-list representation of the robot topology, the algorithm runs in $\mathcal{O}(n_G n_A (n_L+n_J))$ time and $\mathcal{O}(n_L+n_J)$ space, where $n_G$ denotes the number of redundant actuator groups (i.e., the robot DoF).

We classify the ITEP \(P_i\) of actuator \(A_i\) into four topological equivalence classes (A--D) according to the single-DoF relation between its tube and rod parents \(L_{p_t}\) and \(L_{p_r}\).
ITEP topological graphs are shown in Figure~\ref{fig:topo-types}, and the corresponding mechanism examples are illustrated in Figure~\ref{fig:topo-types-demos}.

\begin{enumerate}[label=\textbf{\textit{Type~\Alph*}},align=parleft,leftmargin=*,labelsep=2mm]

  \item \emph{\hl{(Prismatic-equivalent).} 
}
  $L_{p_t}$ and $L_{p_r}$ are connected by a prismatic joint whose telescopic motion is driven by the actuator (e.g., Figure~\ref{fig:topo-types-demos}a).

  \item \emph{\hl{(Revolute-equivalent).}}
  $L_{p_t}$ and $L_{p_r}$ are connected by a revolute joint whose rotational motion is driven by the actuator (e.g., Figure~\ref{fig:topo-types-demos}b).

  \item \emph{\hl{(Four-bar-equivalent).}}
  $L_{p_t}$ and $L_{p_r}$ lie on the same contracted four-bar, i.e., they are connected through a generalized joint that mediates a single-DoF four-bar transmission (e.g., Figure~\ref{fig:topo-types-demos}(c1,c2)).

  \item \emph{\hl{(Generalized-four-bar-equivalent).}}
  A 1-DoF closed loop formed by four links connected in sequence by one generalized joint and three revolute joints; $L_{p_t}$ and $L_{p_r}$ are arbitrary links on this loop (e.g., Figure~\ref{fig:topo-types-demos}(d1,d2)).

\end{enumerate}

Types~A/B are two-link, single-joint ($P/R$) cases;
Type~C is a four-bar;
Type~D is a more complex configuration obtained by nesting one four-bar within another, which we refer to as a \emph{\hl{generalized four-bar}
}.
In this view, Types~A/B can be regarded as \emph{\hl{trivial}} four-bars.

From an engineering perspective, these four types correspond to recurring actuator-transmission patterns in the studied mining robots.
Types~A and B appear in direct hydraulic-cylinder transmissions where actuator stroke is mapped to a sliding or rotating joint motion, such as the cutter-head telescopic motion of the roadheader between \(L_2\) and \(L_3\) in Section~\ref{sec:experiments}
, and the cutter-drum pitching motion of the shearer between \(L_0\) and \(L_1\) in Section~\ref{sec:experiments}.
Type~C appears when a linear actuator drives a planar four-bar substructure, such as the opening/closing transmission of the ventilation door in Section~\ref{sec:experiments}.
Type~D appears in more coupled hydraulic-support mechanisms, where an actuator-induced path includes a contracted four-bar and forms a generalized four-bar transmission after other actuator groups are locked as illustrated in Section~\ref{sec:experiments}.

\begin{figure}[htbp]
  \centering
    \includegraphics[width=0.35\linewidth]{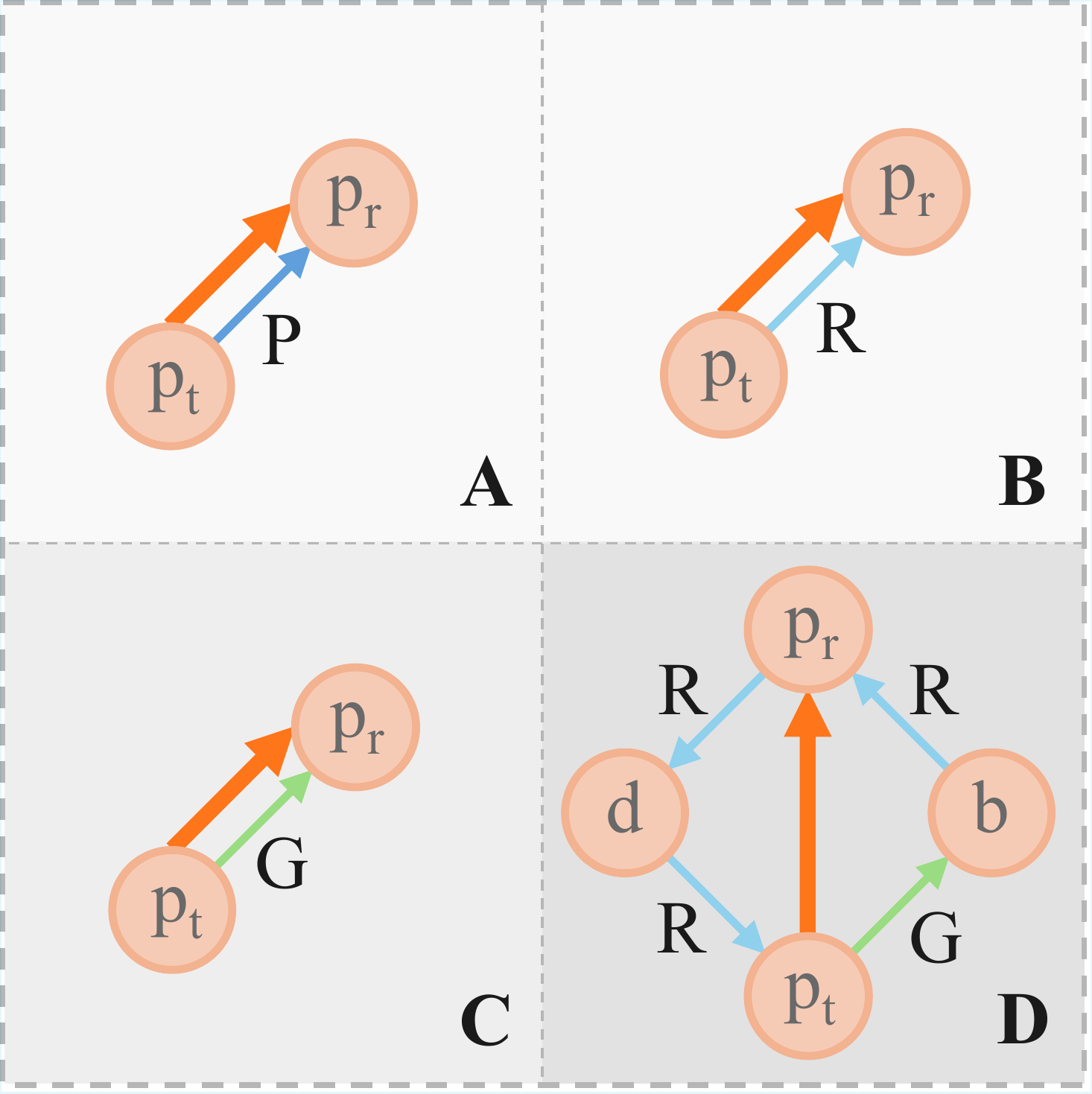}
    \caption{\hl{Topological} 
    graphs of four ITEP types.
    (\textbf{a}) Type A: prismatic-equivalent path; 
    (\textbf{b}) Type B: revolute-equivalent path; 
    (\textbf{c}) Type C: four-bar-equivalent path; 
    (\textbf{d}) Type D: generalized-four-bar-equivalent path. Darker backgrounds indicate higher complexity.}
    \label{fig:topo-types}
\end{figure}

The four ITEP types are complete only under the mechanism scope stated in Section~\ref{sec:modeling}.
More precisely, we assume that:
(i) actuation is provided by linear actuators and each actuator redundancy group contributes one DoF;
(ii) the remaining joint types after fixed-joint merging are revolute, prismatic, or generalized joints obtained from planar four-bar contraction;
(iii) after four-bar contraction, the actuator-free topology is tree-like, or more generally satisfies the independent-path condition that the endpoint links of each actuator are connected by a unique simple path unless locking another actuator induces one local 1-DoF loop; and
(iv) actuator paths do not pass through already actuated generalized joints.

Under these assumptions, locking all actuator redundancy groups except the active group leaves a 1-DoF actuator-induced structure.
If no local loop is induced, the active actuator endpoints are connected through exactly one 1-DoF transmission, which can only be prismatic, revolute, or generalized in the contracted graph; these cases correspond to Types~A, B, and C, respectively.
If locking another actuator induces a local loop, the scope assumptions and mobility counting restrict it to a single 1-DoF four-link loop, yielding the generalized-four-bar case, i.e., Type~D.
Thus every supported actuator-induced structure falls into one of Types~A--D.

\begin{figure}[htbp]    
  \centering
    \includegraphics[width=0.5\linewidth]{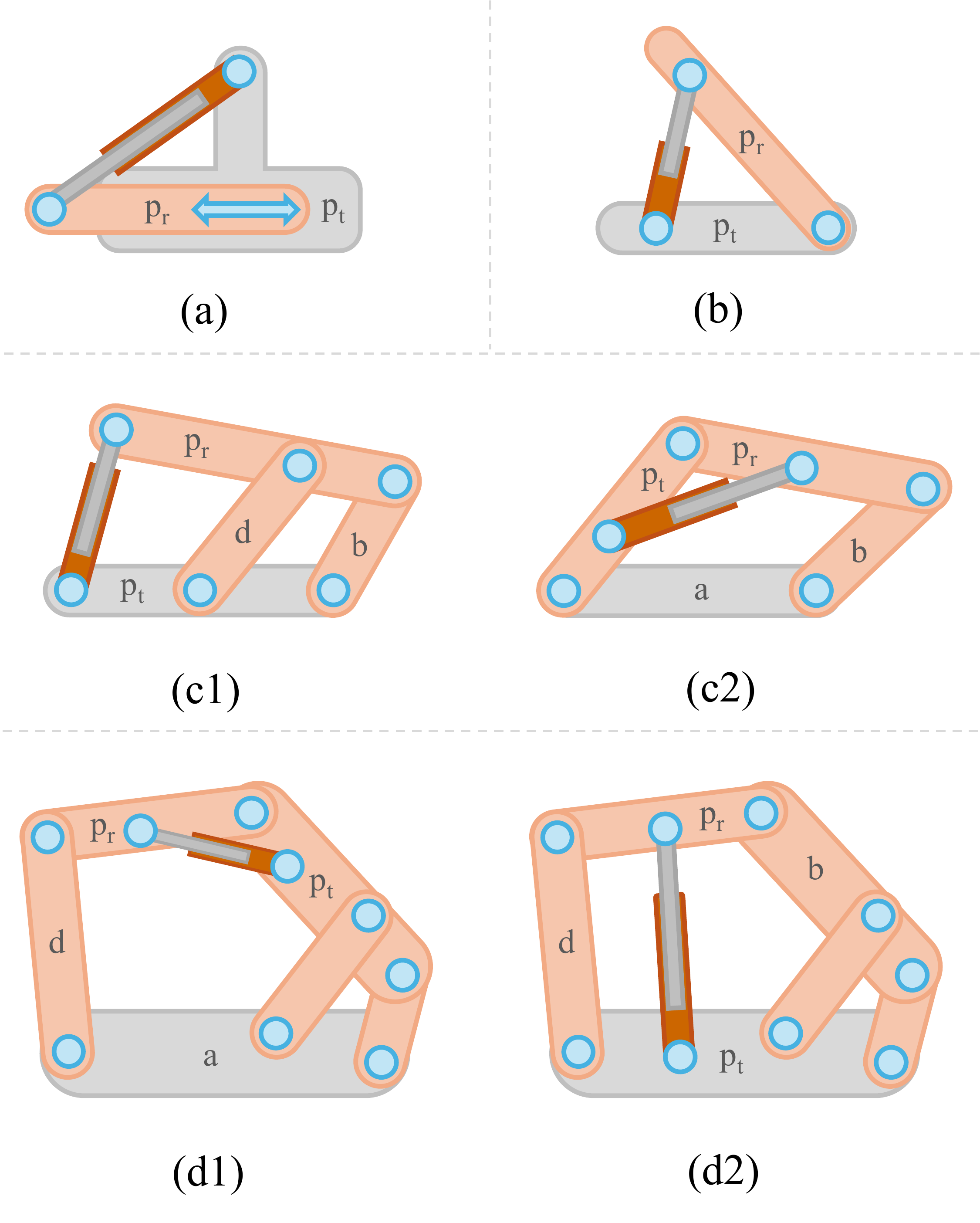}
    \caption{\hl{Illustrative} 
  mechanism examples for the four ITEP types. 
  (\textbf{a}) and (\textbf{b}) correspond to Types A and B, respectively; 
  (\textbf{c1},\textbf{c2}) and (\textbf{d1},\textbf{d2}) show two examples of Types C and D, respectively.}
    \label{fig:topo-types-demos}
\end{figure}

This completeness statement does not apply to arbitrary spatial multi-loop mechanisms, unsupported joint types, topologies that remain multi-cyclic after four-bar contraction, mechanisms requiring actuator paths through already actuated generalized joints, or branch-switching/singular configurations.
Appendix~\ref{appendix:A} provides the detailed proof and a fuller limitations discussion.

The independence of ITEPs refers to the topology-level decomposition obtained by locking all actuator groups except the active redundancy group.
The topological equivalence of ITEPs is characterized by three allowances: (i) the actuator may connect any two links on the ITEP, in either direction; (ii) the attachment points may lie arbitrarily along those links; and (iii) any link on the ITEP may itself be a rigid aggregate of links joined by fixed joints. 
Let \(S\) denote the set of all actuator structures in the robot, and let \(\sim\) be the topological-equivalence relation. The four ITEP types then form the quotient set \(S/\sim = \{\text{Type A},\,\text{Type B},\,\text{Type C},\,\text{Type D}\}\).
The kinematics independence and topological equivalence of ITEPs enable the general kinematics solvers developed in Section~\ref{sec:solver}.

\section{Kinematics Solvers} \label{sec:solver}

In this section, we build kinematics solvers on the ITEP structure obtained in Section~\ref{sec:topo}.
We first show that forward kinematics for each ITEP type reduces to a one-dimensional root-finding problem.
Leveraging ITEP independence, we compose these per-ITEP solvers into an intuitive, sequential forward kinematics pipeline for the full robot.
On top of this FK, we formulate inverse kinematics as a bound-constrained optimization problem and solve it via a Gauss--Seidel-style iterative scheme.
The resulting methods are stable, easy to implement, and efficient, making them suitable for repeated actuator-level kinematic queries in planning, training, and digital-twin systems.

\subsection{Kinematics Solvers for Four Types of ITEPs}
\label{sec:itep-solvers}

We cast the forward kinematics of each ITEP as a unified one-dimensional root-finding problem. 
Given a target actuator length \(l^*\), we determine the joint transformations along the ITEP that realize this length. 
Let \(\varphi\) denote the per-ITEP FK solver, and let \(\mathcal{T}=\{\mathbf{T}\}\) be the set of all joint transformations (the robot configuration). 
Starting from the current configuration \(\mathcal{T}'\), we set
\begin{equation}\label{eq:varphi}
\mathcal{T}=\varphi(l^*,\,\mathcal{T}'),
\end{equation}
which updates the transformations on the ITEP so that the actuator attains \(l^*\).

For any actuator \(A_i\), let \(\theta\) denote the single driving variable along its ITEP (a prismatic displacement, a revolute angle, or a designated input angle inside a four-bar/generalized four-bar). 
Let \(f : \mathbb{R} \xrightarrow{} \mathbb{R}\) denote the function that maps \(\theta\) to the actuator length \(l\).
Then all four ITEP types reduce to the scalar equation
\[
f(\theta) - l^* = 0.
\]

\hl{We solve} 
 this scalar equation using Newton--Raphson iterations with finite-difference derivatives on the physical assembly mode specified by the MRDF reference configuration and subsequent warm starts.
We do not assume that the residual function is globally monotone or convex over all mathematical assembly modes.
Closed-chain mechanisms may in principle admit multiple assembly modes, and branch switching can occur through singular or near-singular configurations.
In the mining robots considered here, such configurations are avoided in normal operation by mechanical design, actuator stroke limits, joint limits, link interference constraints, and safety margins.
For the hydraulic-support benchmarks, this operating-range assumption is also consistent with previous studies reporting one-to-one relationships between support configuration and actuator lengths within valid stroke ranges~\cite{pang2025intelligent,mu2024research}.
Accordingly, the MRDF actuator bounds encode mechanically admissible operating ranges rather than arbitrary numerical intervals, and the target actuator length satisfies \(l^*\in[\ell_i,u_i]\) within this supported range.

Numerically, the implementation uses a safeguarded Newton--Raphson procedure.
A Newton step is accepted only when the finite-difference derivative is sufficiently well conditioned and the proposed update remains within the admissible interval; otherwise, the solver falls back to bisection on the current valid bracket.
The residual tolerance is \(10^{-6}\), and the maximum number of iterations is 50.
If the tolerance cannot be reached within this budget, the corresponding ITEP solve is reported as failed.
We use Newton--Raphson because of its implementation simplicity and observed efficiency in our tests; Appendix~\ref{appendix:fk-rootfinding} shows that several representative one-dimensional solvers also converge on the Type C and Type D cases, indicating that alternative solvers can be adopted depending on the desired speed--robustness trade-off.

For clarity, we introduce the following notation.
The operator \(\mathbf{t}(\mathbf{T})\) extracts the translation vector from a homogeneous transformation \(\mathbf{T}\).
The notation \(\lVert\cdot\rVert\) denotes the Euclidean (\(\ell_2\)) norm.
Consistent with the earlier notation, the actuator’s tube and rod links are denoted by \(L_t\) and \(L_r\), and their parent links by \(L_{p_t}\) and \(L_{p_r}\), respectively.

\subsubsection{Type A}

In Type~A, $L_{p_t}$ and $L_{p_r}$ are connected by a prismatic joint $J_{p_t p_r}$ with displacement $\theta$.
We solve for $\theta$ so that the actuator length equals $l^*$. 
Transformations from $L_{p_t}$ to $L_t$ (orange) and from $L_{p_t}$ to $L_r$ (blue) are illustrated in Figure~\ref{fig:type-a}.
Based on the relative positions of $L_{p_t}$, $J_{p_r r}$, and $J_{p_t t}$ (illustrated by the green triangle in Figure~\ref{fig:type-a}(right)), the actuator length $l$ can be computed as
\begin{equation} \label{eq:type-a-l}
l = \left\lVert 
\mathbf{t}( \mathbf{T}_{p_t t} ) -
\mathbf{t}( \mathbf{T}_{p_t r} )
\right\rVert
\end{equation}
where $\mathbf{T}_{p_t r} = \mathbf{T}_{p_t p_r}(\theta) \mathbf{T}_{p_r} \mathbf{T}_{p_r r}$, according to (\ref{eq:tpc}) and (\ref{eq:Tpq}).
Let $f^0_A : \mathbb{R} \xrightarrow{} \mathbb{R}$ denote the function that maps $\theta$ to the actuator length $l$.
For Type A, the FK solver $\varphi$ reduces to solving a root-finding problem: given the target actuator length $l^*$, find $\theta$ such that
\begin{equation} \label{eq:type-a}
f_A^0(\theta) - l^* = 0.
\end{equation}

\begin{figure}[htbp]
  \centering
        \includegraphics[width=0.7\linewidth]{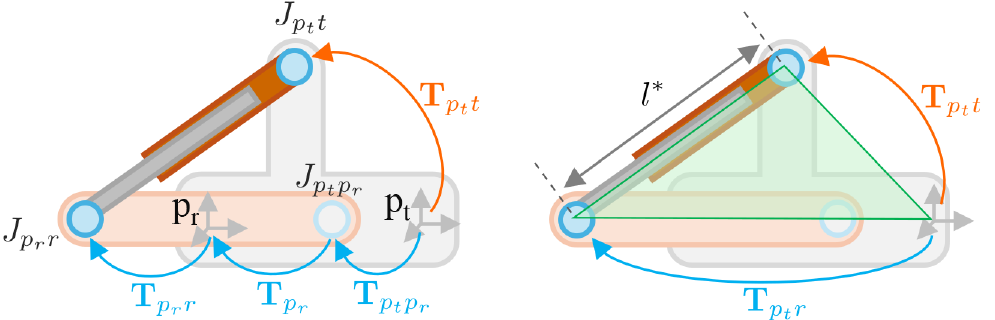}
    \caption{\hl{Type A.} 
     The tube- and rod-parent links $L_{p_t}$ and $L_{p_r}$ are connected by a prismatic joint $J_{p_t p_r}$ with displacement $\theta$. 
    Orange/blue paths show transforms from $L_{p_t}$ to $L_t$/$L_r$.}
    \label{fig:type-a}
\end{figure}

\subsubsection{Type B}

In Type~B, $L_{p_t}$ and $L_{p_r}$ are connected by a revolute joint $J_{p_t p_r}$ with rotation $\theta$. We solve for $\theta$ so that the actuator length equals $l^*$. Transformations from $L_{p_t}$ to $L_t$ (orange) and from $L_{p_t}$ to $L_r$ (blue) are illustrated in Figure~\ref{fig:type-b}. Based on the relative positions of $L_{p_t}$, $J_{p_r r}$, and $J_{p_t t}$ (illustrated by the green triangle in Figure~\ref{fig:type-b}(right)), the actuator length $l$ can be computed as
\begin{equation}\label{eq:type-b-l}
l = \left\lVert 
\mathbf{t}\!\left(\mathbf{T}_{p_t t}\right) -
\mathbf{t}\!\left(\mathbf{T}_{p_t r}\right)
\right\rVert .
\end{equation}

\hl{According} to (\ref{eq:tpc}) and (\ref{eq:Tpq}), the path transformation is
$\mathbf{T}_{p_t r} = \mathbf{T}_{p_t p_r}(\theta)\,\mathbf{T}_{p_r}\,\mathbf{T}_{p_r r}.$
Let $f_B^0:\mathbb{R}\to\mathbb{R}$ denote the function that maps $\theta$ to the actuator length $l$. For Type~B, the FK solver $\varphi$ reduces to solving the scalar equation
\begin{equation}\label{eq:type-b}
f_B^0(\theta) - l^* = 0.
\end{equation}

\vspace{-3pt}
\begin{figure}[htbp]
  \centering
        \includegraphics[width=0.6\linewidth]{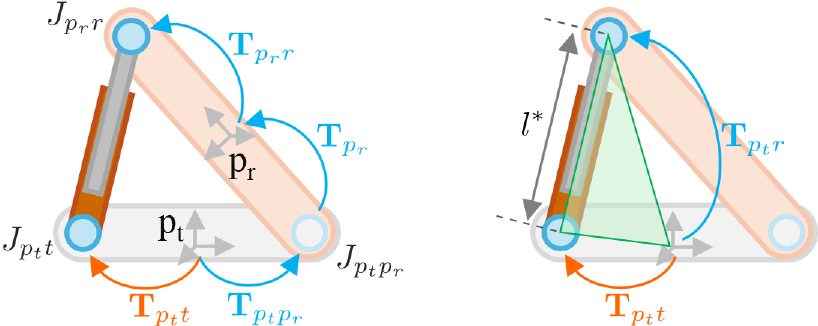}
    \caption{\hl{Type B}. 
    The tube- and rod-parent links $L_{p_t}$ and $L_{p_r}$ are connected by a revolute joint $J_{p_t p_r}$ with rotation $\theta$. 
    Orange/blue paths show transformations from $L_{p_t}$ to $L_t$/$L_r$.}
    \label{fig:type-b}
\end{figure}

\subsubsection{Type C}

In Type~C, the actuator lies on a four-bar \(Fb=[L_a,L_b,L_c,L_d]\) (Figure~\ref{fig:type-c}). We solve for the designated input angle \(\theta\) at \(J_{ab}\) so that the actuator length equals \(l^*\). Let \(\alpha\) and \(\beta\) be the angles at \(J_{bc}\) and \(J_{cd}\), respectively. The four-bar loop-closure is enforced by
\begin{equation}\label{eq:type-c-loop}
\mathbf{t}\!\left(\overline{\mathbf{T}}_{ad}\right)-\mathbf{t}\!\left(\mathbf{T}_{ad}\right)=\mathbf{0},
\end{equation}
where $\overline{\mathbf{T}}_{ad}=\mathbf{T}_{ab}(\theta)\,\mathbf{T}_{b}\,\mathbf{T}_{bc}(\alpha)\,\mathbf{T}_{c}\,\mathbf{T}_{cd}(\beta)\,\mathbf{T}_{d}\,\mathbf{T}_{da}$ according to (\ref{eq:tpc}) and (\ref{eq:Tpq}).
Given $\theta$, the values of $\alpha$ and $\beta$ can be efficiently obtained using a standard Newton--Raphson procedure.
Transformations from \(L_a\) to \(L_t\) (orange) and from \(L_a\) to \(L_r\) (blue) are illustrated in Figure~\ref{fig:type-c}. Based on the relative positions of \(L_a\), \(J_{p_r r}\), and \(J_{p_t t}\) (the green triangle in Figure~\ref{fig:type-c}(right)), the actuator length is
\begin{equation}\label{eq:type-c-l}
l=\left\lVert
\mathbf{t}\!\left(\mathbf{T}_{a t}\right)-
\mathbf{t}\!\left(\mathbf{T}_{a r}\right)
\right\rVert ,
\end{equation}
where \(\mathbf{T}_{a t}\) and \(\mathbf{T}_{a r}\) are the path transformations from \(L_a\) to \(L_t\) and \(L_r\), respectively.
For example, 
$\mathbf{T}_{at} = \mathbf{T}_{p_tt} $,
$\mathbf{T}_{ar} = \mathbf{T}_{ab}\mathbf{T}_{b}\mathbf{T}_{bc}\mathbf{T}_{c}\mathbf{T}_{p_rr} $ 
in Figure~\ref{fig:topo-types-demos}(c1), and 
$\mathbf{T}_{at} = 
\mathbf{T}_{ab}\mathbf{T}_{b}\mathbf{T}_{bc}\mathbf{T}_{c}\mathbf{T}_{cd}\mathbf{T}_{d}\mathbf{T}_{p_tt}$,
$\mathbf{T}_{ar} = \mathbf{T}_{ab}\mathbf{T}_{b}\mathbf{T}_{bc}\mathbf{T}_{c}\mathbf{T}_{p_rr} $
in Figure~\ref{fig:topo-types-demos}(c2).

Let \(f:\mathbb{R}\to\mathbb{R}\) map \(\theta\) to the actuator length \(l\).
Unlike Types A and B, where the actuator length $l$ can be computed directly from $\theta$, the function $f$ here implicitly involves first solving the loop-closure constraint in (\ref{eq:type-c-loop}) to obtain $\alpha$ and $\beta$, and then computing $l$ accordingly.
For Type~C, the FK solver \(\varphi\) reduces to the scalar equation
\begin{equation}\label{eq:type-c}
f(\theta)-l^*=0.
\end{equation}

\vspace{-9pt}
\begin{figure}[htbp]
  \centering
        \includegraphics[width=0.75\linewidth]{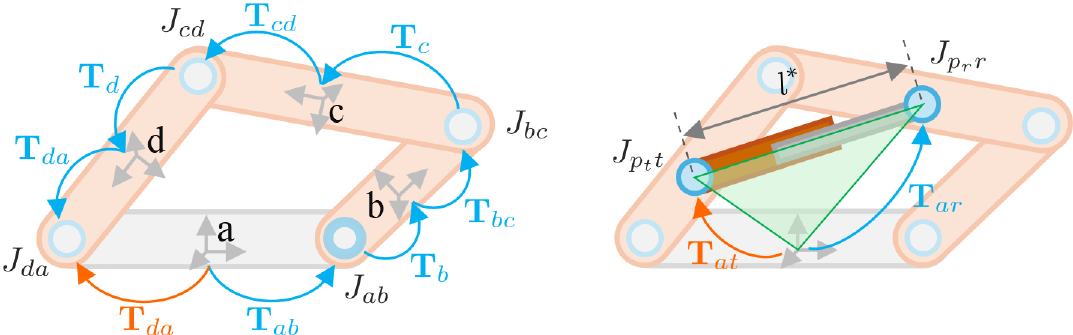}
    \caption{\hl{Type C}. 
    The actuator lies on a four-bar \(Fb=[L_a,L_b,L_c,L_d]\) with designated input \(J_{ab}\) (angle \(\theta\)) and loop angles \(\alpha,\beta\) at \(J_{bc},J_{cd}\). 
    Orange/blue paths show transformations from \(L_a\) to \(L_t\)/\(L_r\).}
    \label{fig:type-c}
\end{figure}

\subsubsection{Type D}

In Type~D, the actuator lies on two nested four-bars—an inner \({Fb}_1=[L_a,L_{b_1},L_{c_1},L_{d_1}]\) and an outer \({Fb}_2=[L_a,L_{b_2},L_{c_2},L_{d_2}]\)—with the shared link \(L_{c_1}\equiv L_{b_2}\) (Figure~\ref{fig:type-d}). We solve for the designated input angle \(\theta\) at \(J_{ab_1}\) so that the actuator length equals \(l^*\). Let \(\alpha_1,\beta_1\) be the angles at \(J_{b_1c_1}\) and \(J_{c_1d_1}\), and \(\alpha_2,\beta_2\) the angles at \(J_{b_2c_2}\) and \(J_{c_2d_2}\), respectively.

The inner-loop closure is enforced by
\begin{equation}\label{eq:type-d-loop-1}
\mathbf{t}\!\left(\overline{\mathbf{T}}_{a d_1}\right)-\mathbf{t}\!\left(\mathbf{T}_{a d_1}\right)=\mathbf{0},
\end{equation}
where
$\overline{\mathbf{T}}_{a d_1}
=\mathbf{T}_{a b_1}(\theta)\,\mathbf{T}_{b_1}\,
\mathbf{T}_{b_1 c_1}(\alpha_1)\,\mathbf{T}_{c_1}\,
\mathbf{T}_{c_1 d_1}(\beta_1)\,\mathbf{T}_{d_1}\,\mathbf{T}_{d_1 a}.$
Given \(\theta\), we obtain \(\alpha_1(\theta),\beta_1(\theta)\) by a standard Newton–Raphson procedure and then enforce the outer-loop closure
\begin{equation}\label{eq:type-d-loop-2}
\mathbf{t}\!\left(\overline{\mathbf{T}}_{a d_2}\right)-\mathbf{t}\!\left(\mathbf{T}_{a d_2}\right)=\mathbf{0},
\end{equation}
with
$\overline{\mathbf{T}}_{a d_2} = \mathbf{T}_{ab_1}(\theta)\, \mathbf{T}_{b_1}\, \mathbf{T}_{b_1 c_1}(\alpha_1)\, \mathbf{T}_{b_2}\, \mathbf{T}_{b_2 c_2}(\alpha_2)\, \mathbf{T}_{c_2}\, \mathbf{T}_{c_2 d_2}(\beta_2)\,\mathbf{T}_{d_2 a}$
which yields \(\alpha_2(\theta),\beta_2(\theta)\).

Transformations from \(L_a\) to \(L_t\) (orange) and from \(L_a\) to \(L_r\) (blue) are illustrated in Figure~\ref{fig:type-d}. Based on the relative positions of \(L_a\), \(J_{p_r r}\), and \(J_{p_t t}\) (the green triangle in \mbox{Figure~\ref{fig:type-d}(right)}), the actuator length is
\begin{equation}\label{eq:type-d-l}
l=\left\lVert
\mathbf{t}\!\left(\mathbf{T}_{a t}\right)-
\mathbf{t}\!\left(\mathbf{T}_{a r}\right)
\right\rVert ,
\end{equation}
where \(\mathbf{T}_{a t}\) and \(\mathbf{T}_{a r}\) are the path transformations from \(L_a\) to \(L_t\) and \(L_r\), respectively.
Let \(f^{2}:\mathbb{R}\to\mathbb{R}\) map \(\theta\) to the actuator length \(l\) (the superscript indicates the two nested closures). 
Here, \(f^{2}\) is defined implicitly: solve the inner closure \eqref{eq:type-d-loop-1} for \(\alpha_{1},\beta_{1}\), then the outer closure \eqref{eq:type-d-loop-2} for \(\alpha_{2},\beta_{2}\); compute \(l\) thereafter.
For Type~D, the FK solver \(\varphi\) reduces to the scalar equation
\begin{equation}\label{eq:type-d}
f^{2}(\theta)-l^*=0.
\end{equation}

\hl{The superscript} in \(f^{2}\) signifies two nested loop closures (a generalized four-bar). By contrast, \(f\) in~\eqref{eq:type-c} corresponds to a single four-bar, whereas \(f_A^{0}\) and \(f_B^{0}\) in~\eqref{eq:type-a} and~\eqref{eq:type-b} denote the prismatic- and revolute-equivalent (trivial) cases.

\begin{figure}[htbp]
  \centering
        \includegraphics[width=0.9\linewidth]{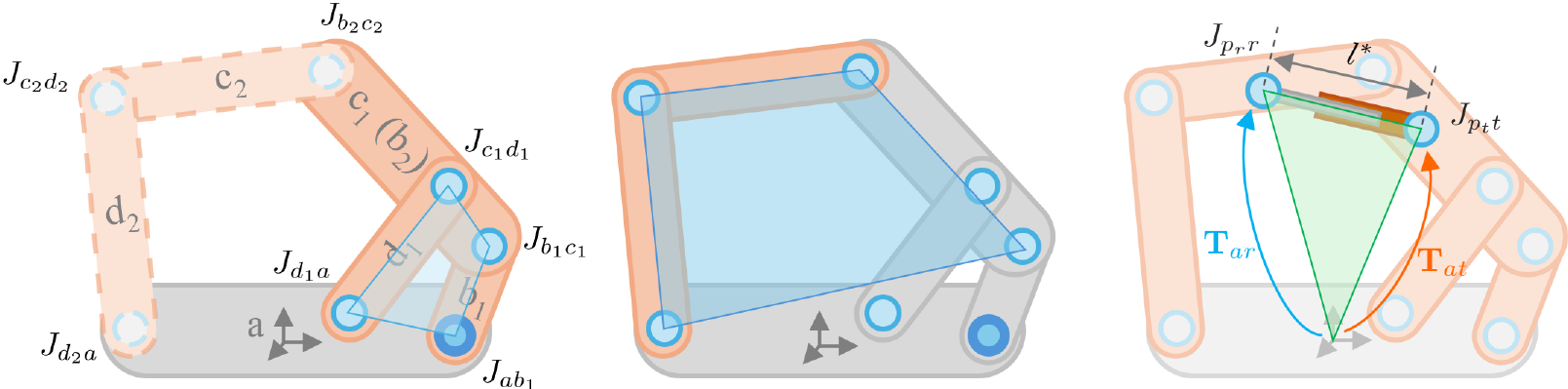}
    \caption{\hl{Type D}.
    The actuator lies on two nested four-bars: the inner \({Fb}_1=[L_a,L_{b_1},L_{c_1},L_{d_1}]\) (\textbf{left}, light blue) and the outer \({Fb}_2=[L_a,L_{b_2},L_{c_2},L_{d_2}]\) (\textbf{middle}, dark blue), which share the link \(L_{c_1}\equiv L_{b_2}\).
    The designated input is \(J_{ab_1}\) (angle \(\theta\)); solving the inner closure for \(\alpha_1,\beta_1\) and then the outer closure for \(\alpha_2,\beta_2\) determines the configuration.
    Orange/blue paths show transformations from \(L_a\) to \(L_t\)/\(L_r\)~(\textbf{right}). }
    \label{fig:type-d}
\end{figure}

\subsection{Forward Kinematics}

With the per-ITEP forward solvers $\varphi$ established, the joint transformations along the ITEP of actuator $A_i$ can be computed from its target length $l_i^*$. 
Leveraging ITEP independence, the robot-level forward kinematics reduces to sequentially solving the FK of each actuator.
Formally, the forward kinematics of the mining robot can be expressed as
\begin{equation} \label{eq:Phi}
     \mathcal{T} = \Phi\left(\mathcal{L}^*, \mathcal{T}' \right)
\end{equation}
where $\mathcal{L}^* = \{ l_i^* \}$ denotes the set of target lengths for all actuators.
More specifically, the computation is carried out sequentially by applying the kinematics solver $\varphi$ to each actuator in turn:
\begin{equation} \label{eq:Phi-detail}
    \Phi\left(\mathcal{L}^*, \mathcal{T}' \right)=
       \varphi\left( l^*_{n_A-1}, \varphi\left( l_{n_A-2}^*, \cdots \varphi\left( l^*_0, \mathcal{T}' \right) \cdots \right) \right)
\end{equation}
which represents the iterative update of joint transformations along the ITEPs, one actuator at a time.

The complete procedure for solving the forward kinematics of a mining robot is presented in Algorithm~\ref{algo:3}.
The input to the algorithm is the set of target actuator lengths $\mathcal{L}^*$ and the current configuration $\mathcal{T}'$ of the robot. The output is an updated configuration $\mathcal{T}$ in which all actuators satisfy their respective target lengths.
For the first FK call, when no previous configuration is available, \(\mathcal{T}'\) is initialized from the MRDF reference configuration, which is constructed from the origin translation and orientation fields of all links and joints.
For subsequent calls, \(\mathcal{T}'\) is warm-started from the configuration returned by the previous FK/IK query.
For each actuator $A_i$, if the target length $l_i^*$ is equal to its current length $l_i'$ (that is, the actuator is \textit{\hl{inactive}}), the algorithm proceeds directly to the next actuator (Steps~2--3). 
Otherwise, the actuator is considered \textit{\hl{active}}, and the corresponding root-finding problem is solved according to its ITEP type $P_i$.
(i) For Type A, the solver computes $\theta$ by solving the root-finding problem in (\ref{eq:type-a}), and updates the joint transformation $\mathbf{T}_{p_t p_r}(\theta)$ accordingly (Steps~4--6).
(ii) For Type B, $\theta$ is obtained by solving (\ref{eq:type-b}), and the transformation $\mathbf{T}_{p_t p_r}(\theta)$ is subsequently updated (Steps~7--9).
(iii) For Type C, the solver returns $\theta$, $\alpha$, and $\beta$ by solving (\ref{eq:type-c}), followed by updating the joint transformations $\mathbf{T}_{ab}(\theta)$, $\mathbf{T}_{bc}(\alpha)$, and $\mathbf{T}_{cd}(\beta)$ (Steps~\mbox{10--12}).
(iv) For Type D, the root-finding problem in (\ref{eq:type-d}) yields $\theta$, $\alpha_1$, $\beta_1$, $\alpha_2$, and $\beta_2$, and the corresponding transformations $\mathbf{T}_{ab_1}(\theta)$, $\mathbf{T}_{b_1c_1}(\alpha_1)$, $\mathbf{T}_{c_1d_1}(\beta_1)$, $\mathbf{T}_{b_2c_2}(\alpha_2)$, and $\mathbf{T}_{c_2d_2}(\beta_2)$ are updated accordingly (Steps~13--15).
Detailed solution strategies for each case are provided in Section~\ref{sec:itep-solvers}.

\vspace{12pt}
\begin{algorithm}[htbp]
\small
\caption{Forward kinematics solver.} \label{algo:3}
\KwIn{Target actuator lengths $\mathcal{L}^*$ and current configuration $\mathcal{T}'$ of a mining robot $\mathcal{R}$}
\KwOut{Updated configuration $\mathcal{T}$ satisfying $\mathcal{L}^*$}

\For{$A_i \in \mathcal{A}$}{
    \If{$l_i^* = l_i'$}{
        \textbf{continue}
    }

    \If{type of $P_i = A$}{
        $\theta \gets$ Solve (\ref{eq:type-a}) \\
        Update $\mathbf{T}_{p_t p_r} \in \mathcal{T}$
    }
    \ElseIf{type of $P_i = B$}{
        $\theta \gets$ Solve (\ref{eq:type-b}) \\
        Update $\mathbf{T}_{p_t p_r} \in \mathcal{T}$
    }
    \ElseIf{type of $P_i = C$}{
        $\theta, \alpha, \beta \gets$ Solve (\ref{eq:type-c}) \\
        Update $\{ \mathbf{T}_{ab}, \mathbf{T}_{bc}, \mathbf{T}_{cd} \} \subset \mathcal{T}$
    }
    \Else{
        $\theta, \alpha_1, \beta_1, \alpha_2, \beta_2 \gets$ Solve (\ref{eq:type-d}) \\
        Update $\{ \mathbf{T}_{ab_1}, \mathbf{T}_{b_1c_1}, \mathbf{T}_{c_1d_1}, \mathbf{T}_{b_2c_2}, \mathbf{T}_{c_2d_2} \} \subset \mathcal{T}$
    }
}

\For{$A_i \in \mathcal{A}$}{
    $t, r, p_t, p_r \gets \textsc{getActuatorLinks}(\boldsymbol{At}, \boldsymbol{Ar}, \boldsymbol{J}, i)$ \\
    $\theta_t, \theta_r \gets$ Solve (\ref{eq:lookat}) \\
    Update $\{ \mathbf{T}_{p_t t}, \mathbf{T}_{p_r r} \} \in \mathcal{T}$
}
\end{algorithm}
\vspace{12pt}

After solving the actuator ITEPs, all actuators satisfy their target lengths, although the tube and rod body frames may not yet be aligned with the line connecting the two mounting points. We therefore apply a lightweight endpoint-alignment refinement to maintain the geometric consistency of actuator-body orientations (Steps~16--17). 
This step is consistent with endpoint-based hydraulic-actuator representations in multibody models~\cite{khadim2023experimental}; it is applied only as a post-processing operation and does not affect subsequent FK/IK solver iterations. Its detailed formulation is provided in Appendix~\ref{appendix:B1}.

The runtime of Algorithm~\ref{algo:3} is dominated by solving the ITEP subproblems for active actuators. 
The final endpoint-alignment refinement involves only constant-cost vector operations and is negligible compared to the ITEP solving process.
Let \(m\) be the number of active actuators and \(a,b,c,d\) the counts of Types~A–D (\(m=a{~+~}b{~+~}c{~+~}d\)). With per-type constant costs \(t_A,t_B,t_C,t_D\), the total time is
$t \;\approx\; a\,t_A \;+\; b\,t_B \;+\; c\,t_C \;+\; d\,t_D,$ and hence the overall complexity is \(O(m)\). Quantitative details are reported in Section~\ref{sec:experiments}.
In summary, the forward kinematics computes the full joint configuration $\mathcal{T}$ from $\mathcal{L}^*$ by exploiting the preconstructed ITEPs. 
Leveraging ITEP independence, it avoids the manual construction and complex solving of the large nonlinear systems required by conventional approaches, and supports independent updates for arbitrary actuator subsets, satisfying flexibility and performance requirements in planning, training, and digital-twin workflows.

\subsection{Inverse Kinematics}

The inverse kinematics of the mining robot aims to determine the set of actuator target lengths $\mathcal{L}^*$ that minimizes the discrepancy between the global transformation $\mathbf{T}_{0e}$ of a designated end-effector link $L_e$ and a given target transformation $\mathbf{T}^*$.
See Section~\ref{sec:experiments} 
 for an example in which the roadheader cutter head and the hydraulic support canopy are solved by the IK solver to meet the given target transformations.
Given a set of actuator lengths $\mathcal{L}$ and an initial configuration $\mathcal{T}'$, the corresponding robot configuration can be obtained via the forward kinematics algorithm as $\Phi\left( \mathcal{L}, \mathcal{T}' \right)$. Based on this configuration, the global transformation of the end-effector link $L_e$ is computed according to (\ref{eq:Tpq}), and is denoted as $\mathbf{T}_{0e}\left( \Phi\left( \mathcal{L}, \mathcal{T}' \right) \right)$.

We measure the pose difference between the end-effector transformation extracted from the FK result \(\Phi(\mathcal{L},\mathcal{T}')\) and the target transformation \(\mathbf{T}^*\) using the Lie algebra representation of their relative transformation~\cite{lynch2017modern}.
Specifically, let
\begin{equation}
    \log\left(
        \mathbf{T}_{0e}^{-1}
        \left( \Phi \left(\mathcal{L}, \mathcal{T}'\right)\right)
        \mathbf{T}^*
    \right)^\vee
    =
    \begin{bmatrix}
    \boldsymbol{\rho} \\
    \boldsymbol{\theta}
    \end{bmatrix},
\end{equation}
where \(\mathbf{T}_{0e}\left(\Phi(\mathcal{L},\mathcal{T}')\right)\) is the end-effector transformation extracted from the FK result, \(\boldsymbol{\rho}\in\mathbb{R}^3\) is the translational component of the logarithmic pose error, and \(\boldsymbol{\theta}\in\mathbb{R}^3\) is the rotational component represented as an axis-angle vector.
Because \(\boldsymbol{\rho}\) has units of length whereas \(\boldsymbol{\theta}\) is measured in radians, we normalize the translational component by a characteristic length \(L_c\).
In this work, \(L_c\) is defined as the distance from the base frame to the selected end-effector in the MRDF reference configuration:
$
    L_c =
    \left\|
    \mathbf{t}\!\left(
    \mathbf{T}_{0e}\left(\mathcal{T}^{\mathrm{ref}}\right)
    \right)
    \right\|_2 
$,
where \(\mathbf{t}(\cdot)\) extracts the translation vector and \(\mathcal{T}^{\mathrm{ref}}\) is the MRDF reference configuration.
The IK objective is defined as the dimensionless weighted pose residual
\begin{equation}
    \Psi(\mathcal{L}, \mathbf{T}^*) =
    \left\|
    \begin{bmatrix}
    \boldsymbol{\rho}/L_c \\
    \boldsymbol{\theta}
    \end{bmatrix}
    \right\|_2
    =
    \sqrt{
    \left\|\boldsymbol{\rho}/L_c\right\|_2^2
    +
    \left\|\boldsymbol{\theta}\right\|_2^2
    } .
\end{equation}

This implies that the evaluation of the objective function $\Psi$ relies on solving the forward kinematics $\Phi$.
Inverse kinematics is thus formulated as the following optimization~problem:
\begin{align} \label{eq:ik-opt}
\mathcal{L}^* = & \mathop{\arg\min}\limits_{\mathcal{L}} ~ \Psi(\mathcal{L}, \mathbf{T}^*) \notag \\
&\mathop{s.t.} \quad \forall l_i \in \mathcal{L}, \, \ell_i \leq l_i \leq u_i
\end{align}
where $\ell_i$ and $u_i$ denote the lower and upper bounds of actuator $i$.

To solve the optimization problem (\ref{eq:ik-opt}), we propose a stable, efficient and easy-to-implement algorithm, as detailed in Algorithm~\ref{alg:ik}.
The algorithm iteratively optimizes actuator lengths and terminates when the objective change between two successive outer iterations falls below a tolerance (set to \(10^{-6}\)) or when the maximum iteration budget is reached.
This tolerance is used to detect numerical convergence of the optimization process; final pose accuracy is evaluated separately using the residual \(\Psi^*\).
First, a marking vector \( \mathbf{M} \in \mathbb{R}^{n_A \times 1} \) is used to record whether an actuator needs to be optimized \( M_i = 0 \) or not \( M_i = 1 \) (Step~1).
For each actuator \( i \), it and all of its redundant actuators are directly marked with \( 1 \) (Steps~5--7).
Next, we check whether the parent links of the tube and rod of \( A_i \), denoted \( L_{p_t} \) and \( L_{p_r} \), are connected to the end-effector link \( L_e \).  
If neither is connected, the motion of \(A_i\) has no effect on \(L_e\), and we refer to such an actuator as an \textit{\hl{irrelevant actuator}}.
For example, in Figure~\ref{fig:sp1}, the motion of \( A_{10} \) does not affect the canopy \( L_4 \).
Otherwise, the actuator is \textit{\hl{relevant}} (i.e., needs to be optimized) (Steps~8--12).
This preprocessing procedure eliminates all irrelevant actuators, ensuring that only one actuator per redundancy group needs to be optimized.

\begin{algorithm}[htbp]
\small
\caption{Inverse kinematics solver.}\label{alg:ik}
\KwIn{Target transformation $\mathbf{T}^*$ of the end-effector; current configuration $\mathcal{T}'$ of mining robot $\mathcal{R}$}
\KwOut{Optimal actuator lengths $\mathcal{L}^*$}

$\boldsymbol{M} \gets \mathbf{0}^{n_A \times 1}$ \\
\For{$i \gets 0$ \KwTo $n_A - 1$}{
    \If{$\boldsymbol{M}_i \neq 0$}{
        \textbf{continue}
    }

    \For{$j \gets 0$ \KwTo $n_A - 1$}{
        \If{$\boldsymbol{Rd}_{i,j} \neq 0$}{
            $\boldsymbol{M}_j \gets 1$
        }
    }
    
    $p_t, p_r \gets \textsc{getActuatorLinks}(\boldsymbol{At}, \boldsymbol{Ar}, \boldsymbol{J}, i)$ \\
    $P_a \gets \textsc{getTopoPath}(\boldsymbol{J'}, p_t, e)$ \\
    $P_b \gets \textsc{getTopoPath}(\boldsymbol{J'}, p_r, e)$ \\
    \If{$P_a \neq [\,] \lor P_b \neq [\,]$}{
        $\boldsymbol{M}_i \gets 0$
    }
}

\For{$k \gets 0$ \KwTo $maxIter$}{
    \For{$i \gets 0$ \KwTo $n_A - 1$}{
        \If{$ \boldsymbol{M}_i = 0 $}{
            $l_i^k \gets \textsc{GoldenSectionSearch}(i)$
        }
    }
    $\Psi^k \gets \Psi(\mathcal{L}^k, \mathbf{T}^*)$ \\
    \If{$|\Psi^k - \Psi^{k-1}| < \text{tol}$}{
        $\mathcal{L}^* \gets \mathcal{L}^k$ \\
        \textbf{break}
    }
}
\end{algorithm}
\vspace{3pt}

In each iteration step $k$, the algorithm sequentially optimizes the length of each actuator using Golden Section Search (GSS)~\cite{kiefer1953sequential,conn2009introduction} over its bound $[\ell_i, u_i]$ (Steps~14--16).
After the search over all marked actuators is completed, the current objective value $\Psi^{k}$ is evaluated (Step~17).
If the absolute difference from the previous iterate is below the tolerance ($10^{-6}$), the IK optimization is regarded as numerically converged (Step~18).
Each objective evaluation invokes the FK routine to update all actuator states; consequently, upon termination the current lengths are returned as the optimal set $\mathcal{L}^{*}$.
The detailed GSS procedure is provided in Appendix~\ref{appendix:B22}.

Ignoring preprocessing and treating each objective evaluation $\Psi(\cdot)$ as a constant-cost FK call ($C_{\mathrm{FK}}$), with $S$ optimized actuators and $K$ outer iterations, GSS needs $\mathcal{O}(\ln((u_i-\ell_i)/\text{tol}))$ evaluations per variable.
One sweep thus uses $\mathcal{O}(S\,\ln(W_{\max}/\text{tol}))$ evaluations, and the overall time complexity is
$\mathcal{O}\big(K\,S\,C_{\mathrm{FK}}\,\ln(W_{\max}/\text{tol})\big)$,
where $W_{\max}=\max_i(u_i-\ell_i)$.

Due to the independence of ITEPs, the multidimensional optimization problem in (\ref{eq:ik-opt}) is decomposed at each iteration into a sequence of bounded one-dimensional subproblems.
This decomposition, together with the Gauss--Seidel-style outer update strategy, is the primary source of the efficiency gain of our IK solver.
A detailed numerical comparison against classical gradient-based and global optimization methods is provided in Section~\ref{sec:experiments}.

In the implementation reported in this paper, each one-dimensional subproblem is solved by GSS.
We use GSS as the default scalar solver because the actuator update is naturally bounded by $[\ell_i,u_i]$, the objective value is obtained through FK evaluations rather than an analytic derivative, and interval-reduction search avoids extra derivative estimation and boundary-handling steps.
This makes GSS simple and stable for repeated actuator-level updates.
The use of GSS also has limitations: it is a local bounded one-dimensional search applied after fixing the other actuator lengths, and therefore it does not by itself guarantee a globally optimal solution for the original coupled, nonconvex IK problem.
For strongly coupled mechanisms, several Gauss--Seidel-style sweeps may be required before the actuator updates stabilize, and the final solution can still depend on the target pose, bounds, and initialization.
Thus, GSS is used here as a robust bounded scalar-search component within the ITEP-based decomposition, rather than as a general global optimizer.
Other bounded one-dimensional solvers can also be used; the ablation in Appendix~\ref{appendix:B22} shows that the dominant efficiency gain comes from the ITEP-based decomposition itself, while the specific choice of the 1D solver plays a secondary role.

\section{Experiments and Results} \label{sec:experiments}

In this section, we present experimental evaluations and analysis of the four algorithms proposed in this paper. 
The effectiveness of topology construction Algorithms~\ref{alg:generalized-joints} and \ref{alg:topo} in processing robot topologies is first qualitatively demonstrated using several representative mining robot examples.
This is followed by a quantitative evaluation of the kinematics solving Algorithms~\ref{algo:3} and \ref{alg:ik} through numerical experiments.

\subsection{Topology Construction Results} 
We begin with a detailed demonstration of Algorithms~\ref{alg:generalized-joints} and \ref{alg:topo} on a complex mining robot, the Top Coal Caving Hydraulic Support Robot. 
Several other representative mining robot examples are then briefly presented in topological graphs.

\begin{figure}[htbp] 
  \centering
        \includegraphics[width=1.0\linewidth]{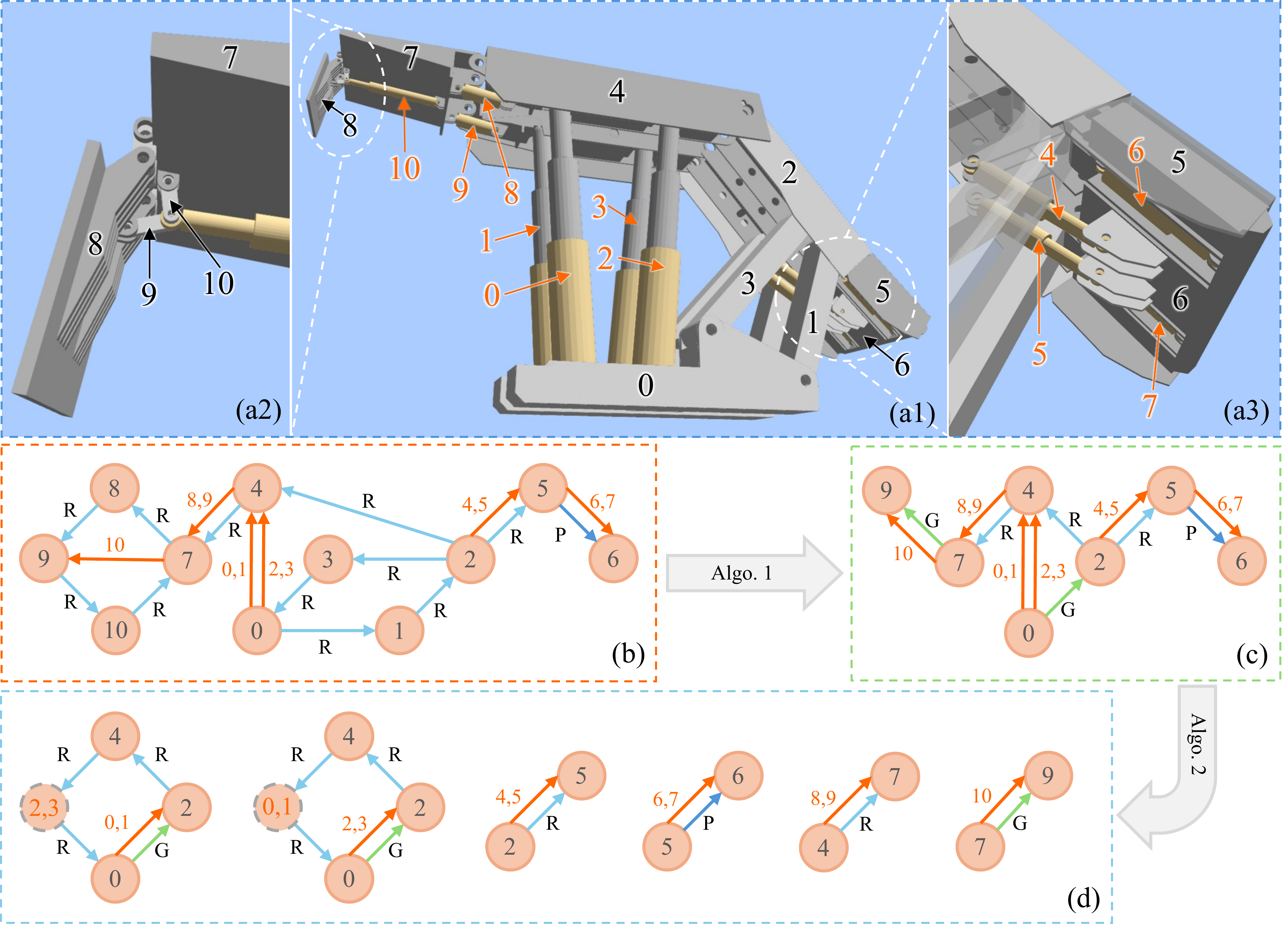}
    \caption{\hl{Topology} 
 construction of Top Coal Caving Hydraulic Support Robot.
  construction of Top Coal Caving Hydraulic Support Robot. 
(\textbf{a1})--(\textbf{a3}) Link and actuator annotations of the robot; 
(\textbf{b}) initial topological graph; 
(\textbf{c}) topology after four-bar contraction; 
(\textbf{d}) extracted actuator ITEPs.} 
    \label{fig:sp1}
\end{figure}

Figure~\ref{fig:sp1} shows the model and topology processing of the Top Coal Caving Hydraulic Support Robot.
Figure~\ref{fig:sp1} (a1),(a2),(a3) show the rendered model, where the indices of links and actuators are marked in black and orange, respectively.
For brevity, we omit the link indices of all actuator tubes and rods.
This robot has a complex structure, featuring two pairs of supporting hydraulic actuators, front ($A_0$, $A_1$) and rear ($A_2$, $A_3$), as well as a multilevel front beam (see Figure~\ref{fig:sp1}(a2)) and a telescopic tail beam (see Figure~\ref{fig:sp1}(a3)).
In the topological graph shown in Figure~\ref{fig:sp1}b, the links, joints, and actuators are represented as nodes, blue directed edges, and orange directed edges, respectively. Redundant actuators are visualized using a shared arrow.
This topological information is encoded in the MRDF and is parsed into matrices $\boldsymbol{J}$, $\boldsymbol{Rd}$, $\boldsymbol{At}$, and $\boldsymbol{Ar}$ (see Equation~\eqref{eq:sp1} for $\boldsymbol{J}$ and $\boldsymbol{Rd}$):
\begin{equation} \label{eq:sp1}
\begin{aligned}
\boldsymbol{J} &=
\scalebox{0.4}{$
\begin{blockarray}{rccccccccccccccc}
         & 0 & 1 & 2 & 3 & 4 & 5 & 6 & 7 & 8 & 9 & 10 & 11 & \cdots & 32\\  
\begin{block}{r(ccccccccccc|cccc)}
0        & 0 & 1 & 0 & 0 & 0 & 0 & 0 & 0 & 0 & 0 & 0  & 0  & \cdots & 0 \\
1        & 0 & 0 & 1 & 0 & 0 & 0 & 0 & 0 & 0 & 0 & 0  & 0  & \cdots & 0 \\
2        & 0 & 0 & 0 & 1 & 1 & 1 & 0 & 0 & 0 & 0 & 0  & 0  & \cdots & 0 \\
3        & 1 & 0 & 0 & 0 & 0 & 0 & 0 & 0 & 0 & 0 & 0  & 0  & \cdots & 0 \\
4        & 0 & 0 & 0 & 0 & 0 & 0 & 0 & 1 & 0 & 0 & 0  & 0  & \cdots & 0 \\
5        & 0 & 0 & 0 & 0 & 0 & 0 & 2 & 0 & 0 & 0 & 0  & 0  & \cdots & 0 \\
6        & 0 & 0 & 0 & 0 & 0 & 0 & 0 & 0 & 0 & 0 & 0  & 0  & \cdots & 0 \\
7        & 0 & 0 & 0 & 0 & 0 & 0 & 0 & 0 & 1 & 0 & 0  & 0  & \cdots & 0 \\
8        & 0 & 0 & 0 & 0 & 0 & 0 & 0 & 0 & 0 & 1 & 0  & 0  & \cdots & 0 \\
9        & 0 & 0 & 0 & 0 & 0 & 0 & 0 & 0 & 0 & 0 & 1  & 0  & \cdots & 0 \\
10       & 0 & 0 & 0 & 0 & 0 & 0 & 0 & 1 & 0 & 0 & 0  & 0  & \cdots & 0 \\
\cline{2-15}
11       & 0 & 0 & 0 & 0 & 0 & 0 & 0 & 0 & 0 & 0 & 0  & 0  & \cdots & 0 \\
\vdots   & \vdots & \vdots & \vdots & \vdots & \vdots & \vdots & \vdots & \vdots & \vdots & \vdots & \vdots & \vdots & \ddots & \vdots \\
32       & 0 & 0 & 0 & 0 & 0 & 0 & 0 & 0 & 0 & 0 & 0  & 0  & \cdots & 0 \\
\end{block}
\end{blockarray}
$}
, \quad
\boldsymbol{Rd} =
\scalebox{0.5}{$
\begin{blockarray}{rcccccccccccc}
		 & 0 & 1 & 2 & 3 & 4 & 5 & 6 & 7 & 8 & 9 & 10 \\  
\begin{block}{r(cccccccccccc)}
0        & 1 & 1 & 0 & 0 & 0 & 0 & 0 & 0 & 0 & 0 & 0 \\
1        & 1 & 1 & 0 & 0 & 0 & 0 & 0 & 0 & 0 & 0 & 0 \\
2        & 0 & 0 & 1 & 1 & 0 & 0 & 0 & 0 & 0 & 0 & 0 \\
3        & 0 & 0 & 1 & 1 & 0 & 0 & 0 & 0 & 0 & 0 & 0 \\
4        & 0 & 0 & 0 & 0 & 1 & 1 & 0 & 0 & 0 & 0 & 0 \\
5        & 0 & 0 & 0 & 0 & 1 & 1 & 0 & 0 & 0 & 0 & 0 \\
6        & 0 & 0 & 0 & 0 & 0 & 0 & 1 & 1 & 0 & 0 & 0 \\
7        & 0 & 0 & 0 & 0 & 0 & 0 & 1 & 1 & 0 & 0 & 0 \\
8        & 0 & 0 & 0 & 0 & 0 & 0 & 0 & 0 & 1 & 1 & 0 \\
9        & 0 & 0 & 0 & 0 & 0 & 0 & 0 & 0 & 1 & 1 & 0 \\
10       & 0 & 0 & 0 & 0 & 0 & 0 & 0 & 0 & 0 & 0 & 1 \\
\end{block}
\end{blockarray}
$}
\end{aligned}
\end{equation}

After Algorithm~\ref{alg:generalized-joints}, all four-bars ($[L_{0},L_{1},L_{2},L_{3}]$ and $[L_{7},L_{8},L_{9},L_{10}]$) have been contracted into generalized joints as indicated by the green directed arrows in Figure~\ref{fig:sp1}c.
The ITEPs of all actuators are then obtained by Algorithm~\ref{alg:topo} (see Figure~\ref{fig:sp1}d).
In addition, the topology processing results of several other representative mining robots are shown in Figure~\ref{fig:cases}. These include two important types of hydraulic support robots: the Shield-Type Hydraulic Support Robot (Figure~\ref{fig:cases}a) and the Support-Shield Hydraulic Support Robot (Figure~\ref{fig:cases}b), as well as the Roadheader Robot (Figure~\ref{fig:cases}c). Several other robots with relatively simple topologies—namely, the Load-Haul-Dump Robot, Ventilation Door Robot, Drilling Jumbo Robot, and Shearer Robot—are illustrated in Figure~\ref{fig:cases}d.
The numbers of links, joints, and actuators for all mining robots, after merging fixed-joint chains during preprocessing, are summarized in Table~\ref{tab:robot_statistics}.

\begin{figure}[htbp]
\centering 
    \includegraphics[width=1.0\linewidth]{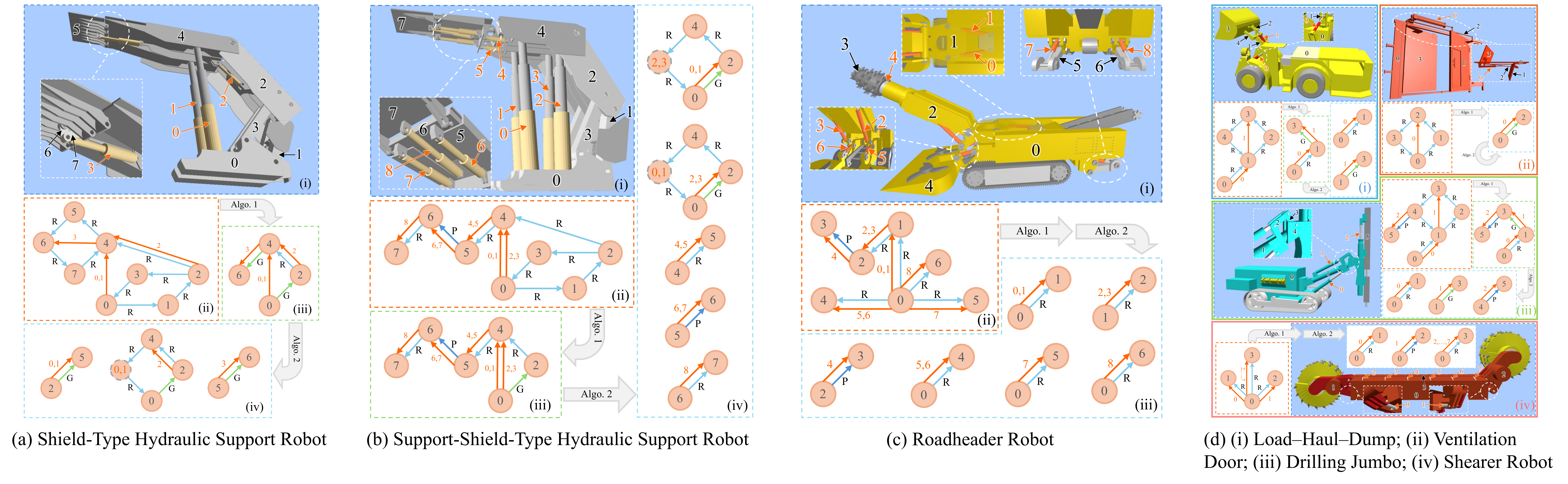}
    \caption{\hl{Topology} 
  construction of various mining robots. 
  Panels (\textbf{a})--(\textbf{c}) show the Shield-Type Hydraulic Support Robot, Support-Shield-Type Hydraulic Support Robot, and Roadheader Robot, respectively. 
  Panel (\textbf{d}) summarizes four auxiliary mining robots: 
  (\textbf{i}) Load--Haul--Dump, 
  (\textbf{ii}) Ventilation Door, 
  (\textbf{iii}) Drilling Jumbo, and 
  (\textbf{iv}) Shearer Robot.
 }
    \label{fig:cases}
\end{figure}

\vspace{-6pt}

\begin{table}[htbp]
  \caption{Mining robots and their structural information.}
  \label{tab:robot_statistics}
  \centering 
  \resizebox{\textwidth}{!}{
    \begin{tabular}{lcccl}
      \toprule
      \textbf{Mining Robot} & \textbf{Links} & \textbf{Joints} & \textbf{Actuators} & \textbf{End-Effector} \\
      \midrule
      Top Coal Caving Hydraulic Support Robot (TCCHS) (Figure~\ref{fig:sp1}) & 11 & 34 & 11 & Canopy ($L_4$) \\
      Shield-Type Hydraulic Support Robot (STHS) (Figure~\ref{fig:cases}a) & 8 & 16 & 4 & Canopy ($L_4$) \\
      Support-Shield Hydraulic Support Robot (SSHS) (Figure~\ref{fig:cases}b) & 8 & 22 & 7 & Canopy ($L_4$) \\
      Roadheader Robot (RH) (Figure~\ref{fig:cases}c) & 7 & 25 & 9 & Cutter Head ($L_3$)\\
      Load-Haul-Dump Robot (LHD) (Figure~\ref{fig:cases}d(i)) & 5 & 9 & 2 & Bucket ($L_3$) \\
      Ventilation Door Robot (VD) (Figure~\ref{fig:cases}d(ii)) & 4 & 6 & 1 & Door Leaf ($L_1$ or $L_3$) \\
      Drilling Jumbo Robot (DJ) (Figure~\ref{fig:cases}d(iii)) & 6 & 12 & 3 & Rock Drill ($L_5$) \\
      Shearer Robot (SH) (Figure~\ref{fig:cases}d(iv)) & 4 & 20 & 8 & Cutter Drum ($L_1$ or $L_2$) \\
      \bottomrule
    \end{tabular}
  }
\end{table}

These examples demonstrate the topology construction process for mining robots. Algorithm~\ref{alg:generalized-joints} contracts all four-bar linkages into generalized joints, after which Algorithm~\ref{alg:topo} constructs the ITEP for each actuator.
As the preprocessing stage of the MineRobot framework, these algorithms automatically construct ITEPs from the topological information recorded in the MRDF, serving as the foundation for forward and inverse kinematics and avoiding the tedious manual kinematics analysis required by traditional methods.

\subsection{Numerical Results of Kinematics Solvers}

The following presents the quantitative evaluation and analysis of the proposed kinematics solvers.
We first evaluate the forward kinematics computation time on various mining robots and analyze the scalability of the algorithm.
Then, we present the convergence behavior and performance of the inverse kinematics algorithm and compare it with several alternative methods.
All timing results were obtained on a laptop equipped with an AMD Ryzen~7 5800H CPU, 32~GB of RAM, and an NVIDIA GeForce RTX~3060 Laptop GPU. 
The solver was implemented in C++ using Eigen~3.4.0~\cite{eigenweb}.
For reproducibility, all randomized experiments use a fixed random seed of 42 and 100 trials per robot or per target set unless otherwise stated.
Actuator lengths are sampled uniformly within the MRDF-defined bounds.
The solver parameters, optimizer settings, reference-workflow settings, and data-availability details are summarized in Appendix~\ref{appendix:experiment-details}.

\subsubsection{Forward Kinematics Solver}

To evaluate the performance of the forward kinematics solver (Algorithm~\ref{algo:3}), we design the following experiment.
For a given robot \( \mathcal{R} \), each actuator may be active or inactive.  
Let \( a \), \( b \), \( c \), and \( d \) denote the number of active actuators of each ITEP type, respectively.  
Each valid combination of \( (a, b, c, d) \) corresponds to one experimental case.
For each case, we randomly generate a set of target actuator lengths \( \mathcal{L}^* \),  
where the target length \( l_i^* \) of each active actuator \( A_i \) is sampled uniformly from the interval \( [\ell_i, u_i] \).  
The initial configuration \( \mathcal{T}' \) of the robot is also randomly initialized.  
The solver is then executed to compute the resulting configuration \( \mathcal{T} \).  
An FK trial is considered successful if all active actuator-length residuals and all loop-closure residuals are below the solver tolerance \(10^{-6}\).
Each case is repeated 100 times with randomized inputs, and the average execution time is recorded.
Table~\ref{tab:fk_time} presents a subset of the results, specifically the average runtime when all actuators are active.

\begin{table}[htbp]
  \caption{Forward kinematics computation time.}
  \label{tab:fk_time}
  \centering
  \resizebox{0.7\textwidth}{!}{
    \begin{tabular}{lcccccccc}
      \toprule
      & \textbf{TCCHS} & \textbf{STHS} & \textbf{SSHS} & \textbf{RH} & \textbf{LHD} & \textbf{VD} & \textbf{DJ} & \textbf{SH} \\
      \midrule
      $a$ & 2 & 0 & 2 & 1 & 1 & 0 & 1 & 0 \\
      $b$ & 4 & 0 & 3 & 8 & 0 & 0 & 1 & 8 \\
      $c$ & 1 & 3 & 0 & 0 & 1 & 1 & 1 & 0 \\
      $d$ & 4 & 1 & 4 & 0 & 0 & 0 & 0 & 0 \\
      Time (ms) & 1.0557 & 0.5623 & 0.8693 & 0.0091 & 0.1762 & 0.1764 & 0.1691 & 0.0072 \\
      \bottomrule
    \end{tabular}
  }
\end{table}

The experimental results demonstrate that our forward kinematics algorithm is highly efficient. 
For the most complex robot, TCCHS, the computation time is approximately 1~ms when all actuators are active, which meets the requirements for real-time applications.
Conventional methods formulate one closed-loop constraint equation per actuator, yielding a coupled nonlinear system that scales with the total number of actuators. This system must be solved in full, regardless of actuator activity.
By exploiting the independence of ITEPs, our approach decouples the problem and achieves linear scalability---the computational cost scales linearly with the number of active actuators.
According to the procedure of Algorithm~\ref{algo:3}, the total execution time \( t \) can be approximated as
\begin{equation}
\label{eq:fk-time-model}
    t \approx a \cdot t_A + b \cdot t_B + c \cdot t_C + d \cdot t_D,
\end{equation}
where \( a \), \( b \), \( c \), and \( d \) denote the number of active actuators of each ITEP type,  
and \( t_A \), \( t_B \), \( t_C \), and \( t_D \) represent the average time required to solve each type,  
including both numerical root-finding and auxiliary overhead.
We perform a linear regression using the values of \( a \), \( b \), \( c \), and \( d \) from each experimental case to fit the execution time \( t \).  
The resulting coefficients are \( t_A = 0.0025 \)~ms, \( t_B = 0.0022 \)~ms, \( t_C = 0.1850 \)~ms, and \( t_D = 0.2396 \)~ms.  
Using these coefficients, the predicted execution times closely match the measured values, achieving a coefficient of determination \(R^2 = 0.9936\) and a root mean squared error (RMSE) of 0.02948~ms.
Figure~\ref{fig:fk-runtime-regression} compares the measured FK runtimes with the values predicted by the linear model; 
the close alignment with the 45-degree reference line supports the scalability analysis of the FK solver.

\begin{figure}[htbp]
  \centering
    \includegraphics[width=0.4\linewidth]{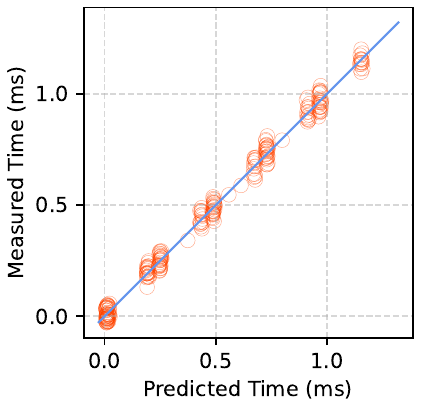}
    \caption{Measured versus predicted FK runtimes for the scalability analysis.
    Each orange point corresponds to one randomized FK timing case.
    The horizontal axis shows the runtime predicted by the linear model in Equation~\eqref{eq:fk-time-model}, and the vertical axis shows the measured runtime.
    The blue line indicates the ideal 45-degree agreement line.}
    \label{fig:fk-runtime-regression}
\end{figure}

\subsubsection{Inverse Kinematics Solver}
\label{sec:ik-experiments}

To evaluate the convergence behavior and performance of the inverse kinematics solver (Algorithm~\ref{alg:ik}), we conduct comparative experiments on all robots in Table~\ref{tab:robot_statistics}.
For each trial, the initial configuration is generated by sampling actuator lengths within their bounds and applying the FK solver, ensuring a feasible closed-chain state.
We first compare convergence speed against commonly used optimizers using randomized target poses around the robot workspace.
We then separately evaluate reachable and unreachable target sets to distinguish pose accuracy from numerical convergence.
Figure~\ref{fig:ik-imgs} shows two representative examples (Roadheader and TCCHS) before and after IK solving.

For the optimizer comparison, we run 100 randomized IK trials per robot.
In each trial, we record the objective value \(\Psi\) at every outer iteration of Algorithm~\ref{alg:ik}.
This experiment evaluates numerical convergence speed under a common objective-change stopping criterion, rather than serving as the pose-accuracy test.
We compare against commonly used optimizers under the same objective function, actuator bounds, initial configurations, target poses, random seed, tolerance, maximum iteration budget, and time budget, with detailed settings provided in Appendix~\ref{appendix:experiment-details}.
The baselines include Broyden--Fletcher--Goldfarb--Shanno (BFGS), conjugate gradient (CG), basic gradient descent, covariance matrix adaptation evolution strategy (CMA-ES), and (1+1)-ES~\cite{li2023difffr, geilinger2020add}.

The experimental results are summarized numerically in Table~\ref{tab:simulation_results}, while Figure~\ref{fig:ik_results} visualizes the corresponding convergence curves.
Table~\ref{tab:simulation_results} reports the median iteration count and average runtime over 100 randomized trials for each method, providing the numerical values associated with the trends shown in Figure~\ref{fig:ik_results}.
Numerical convergence is declared using the same objective-change tolerance as in Algorithm~\ref{alg:ik} (\(10^{-6}\)).
For (1+1)-ES, which converges slowly, we report the first 200 iterations.
Figure~\ref{fig:ik_results} illustrates the convergence behavior on three hydraulic support robots and the Roadheader; the remaining robots exhibit similar trends to the Roadheader and are omitted for brevity.
As gradient descent converges substantially more slowly, it is included only as a visual baseline in Figure~\ref{fig:ik_results} and omitted from Table~\ref{tab:simulation_results} due to non-competitive performance.

\begin{figure}[htbp]
\centering
    \subcaptionbox{Roadheader in initial configuration}[0.4\linewidth][c]{%
        \includegraphics[width=\linewidth]{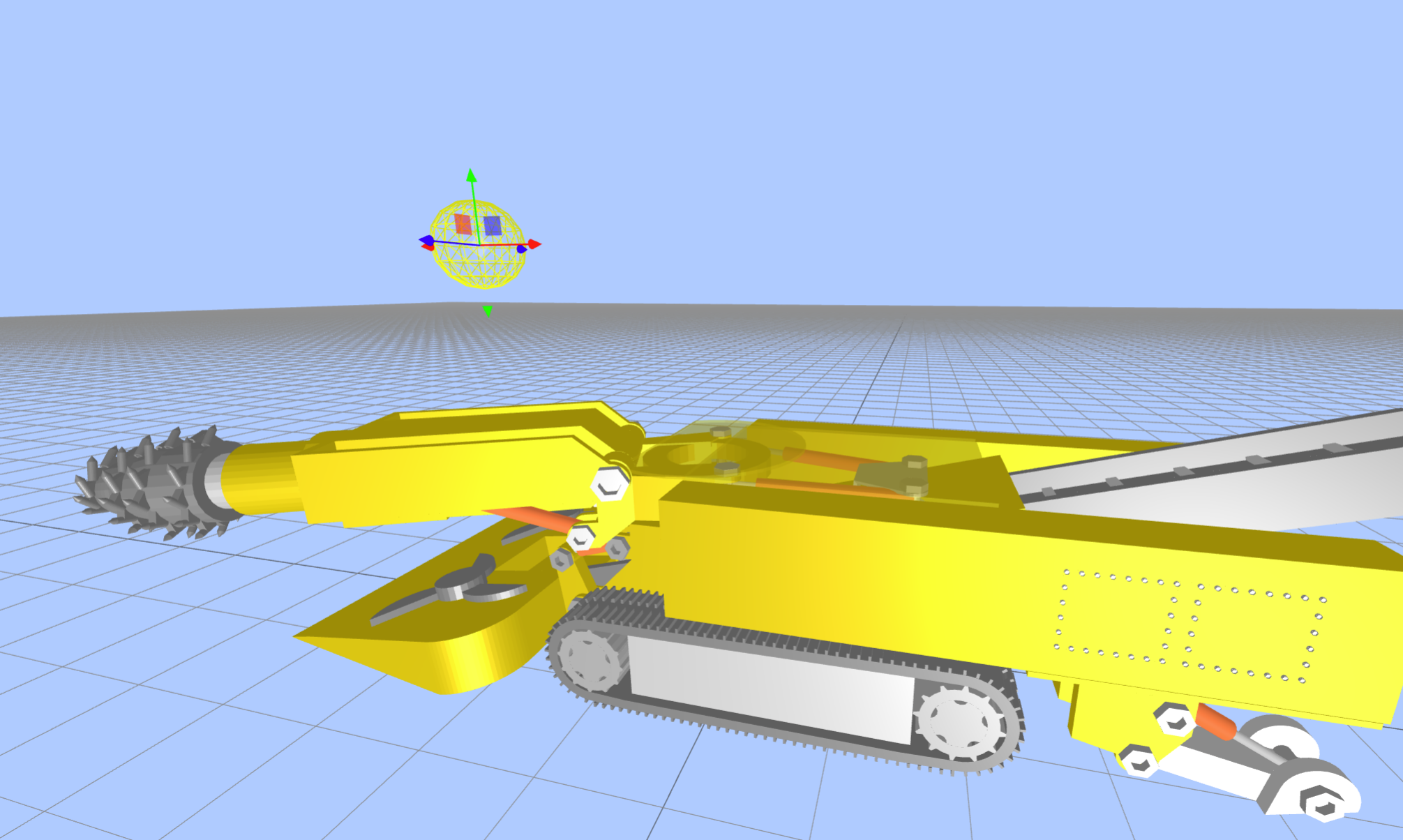}}
    \subcaptionbox{Roadheader after IK solving}[0.4\linewidth][c]{%
        \includegraphics[width=\linewidth]{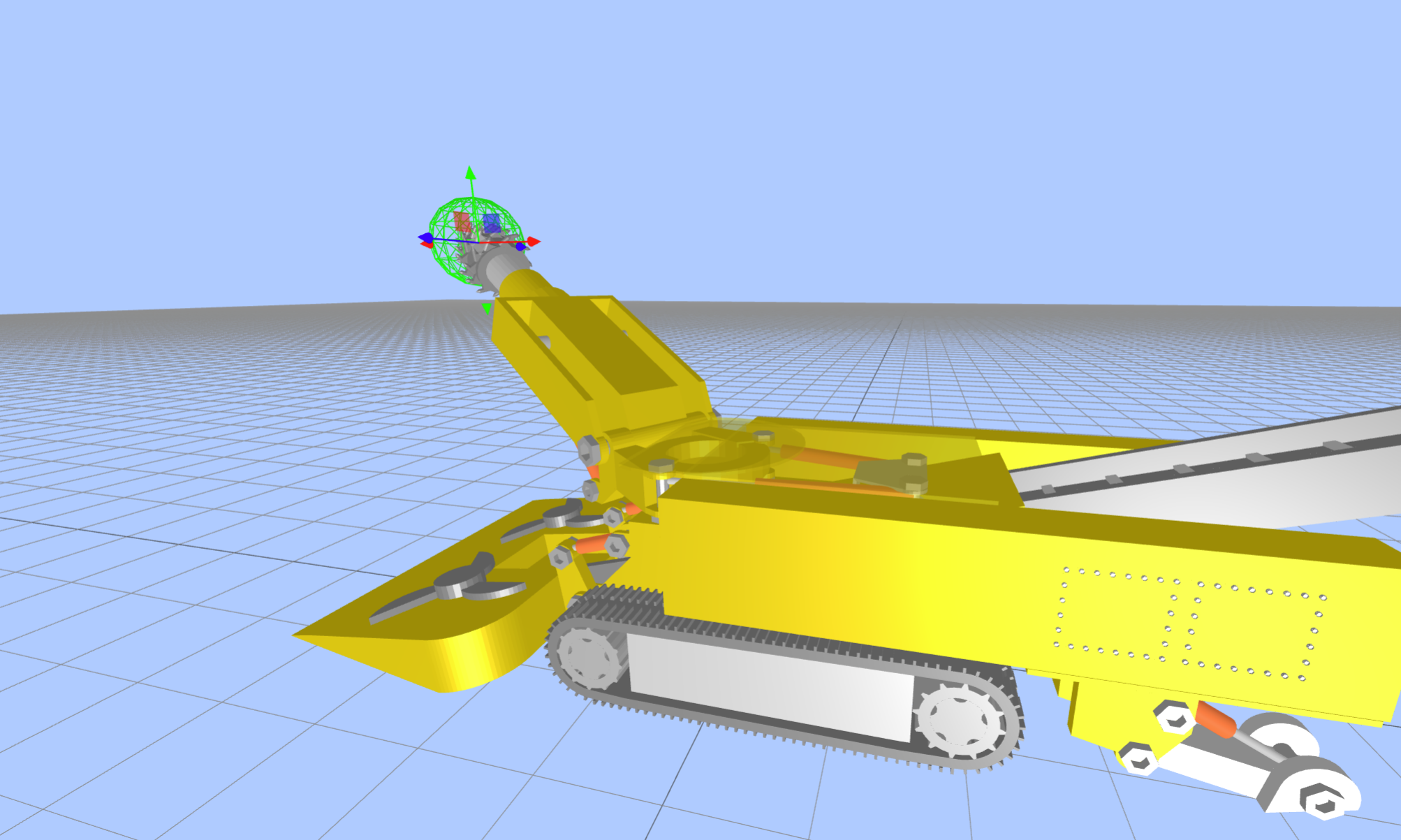}}

    \subcaptionbox{TCCHS in initial configuration}[0.4\linewidth][c]{%
        \includegraphics[width=\linewidth]{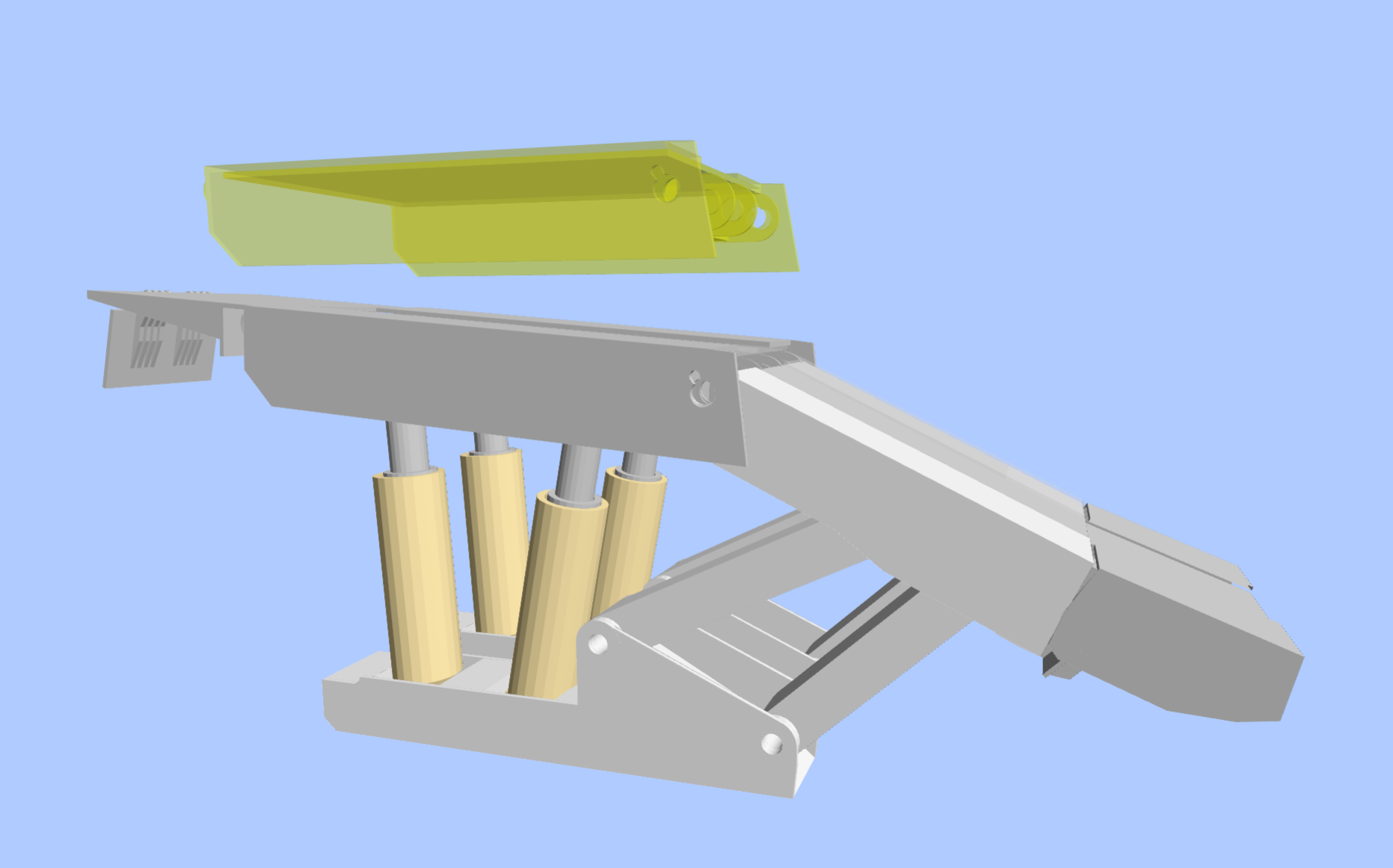}}
    \subcaptionbox{TCCHS after IK solving}[0.4\linewidth][c]{%
        \includegraphics[width=\linewidth]{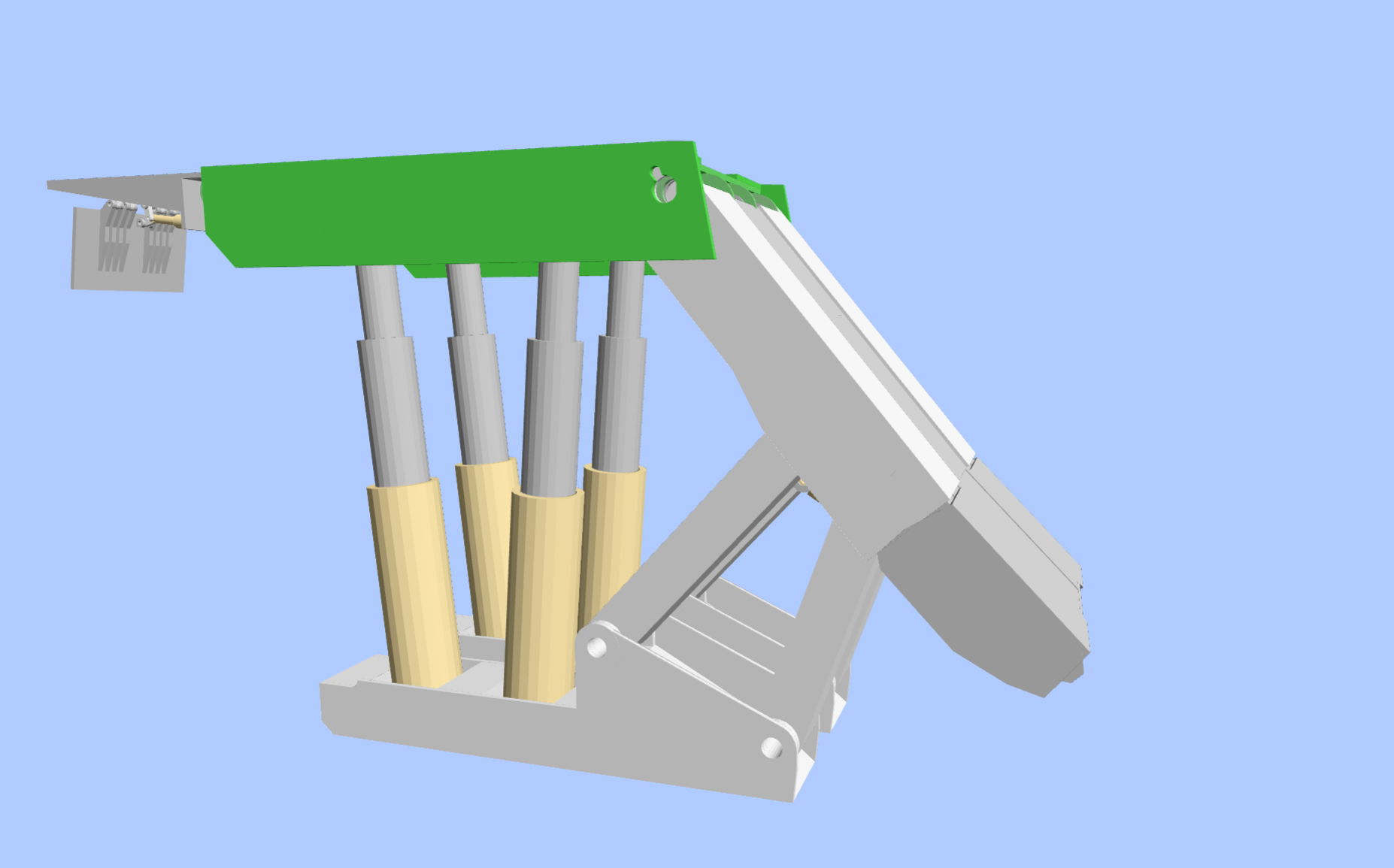}}

    \caption{\hl{Inverse} kinematics results for two mining robots. 
    (\textbf{a}) Roadheader in a random initial configuration; the target end-effector position is indicated by the yellow sphere. 
    (\textbf{b}) Roadheader after IK solving. 
    (\textbf{c}) TCCHS in a random initial configuration; the target canopy position is shown in yellow. 
    (\textbf{d}) TCCHS after IK solving. }
    \label{fig:ik-imgs}
\end{figure}

\begin{table}[htbp]
  \caption{\textbf{Inverse kinematics performance comparison} over 100 randomized trials per robot. All methods use the same stopping criterion and tolerance. Gradient descent was also tested but is omitted from the table due to non-competitive performance. For the (1+1) evolution strategy ((1+1)-ES), only the first 200 iterations are reported.}
  \label{tab:simulation_results}
  \centering 
  \resizebox{\textwidth}{!}{
    \begin{tabular}{lrrrrrrrr}
      \toprule
      & \multicolumn{2}{c}{\textbf{TCCHS}}
      & \multicolumn{2}{c}{\textbf{STHS}}
      & \multicolumn{2}{c}{\textbf{SSHS}}
      & \multicolumn{2}{c}{\textbf{RH}} \\
      \cmidrule{2-3} \cmidrule{4-5} \cmidrule{6-7} \cmidrule{8-9}
      & \textbf{Iters.} & \textbf{Time (ms)}
      & \textbf{Iters.} & \textbf{Time (ms)}
      & \textbf{Iters.} & \textbf{Time (ms)}
      & \textbf{Iters.} & \textbf{Time (ms)} \\
      \midrule
      \textbf{Ours} & \textbf{11} & \textbf{78.63} & \textbf{14} & \textbf{96.27} & \textbf{10} & \textbf{73.91} & \textbf{1} & \textbf{0.21} \\
      Quasi-Newton (BFGS) & 40 & 263.04 & 19 & 137.34 & 18 & 123.18 & 6  & 0.48 \\
      Conjugate Gradient  & 33 & 174.76 & 29 & 122.03 & 28 & 144.18 & 6  & 0.34 \\
      CMA-ES              & 34 & 143.56 & 58 & 253.07 & 27 & 129.31 & 34 & 1.39 \\
      (1+1)-ES            & 200 & 278.81 & 200 & 283.70 & 200 & 302.13 & 200 & 9.35 \\
      \midrule
      & \multicolumn{2}{c}{\textbf{LHD}}
      & \multicolumn{2}{c}{\textbf{VD}}
      & \multicolumn{2}{c}{\textbf{DJ}}
      & \multicolumn{2}{c}{\textbf{SH}} \\
      \cmidrule{2-3} \cmidrule{4-5} \cmidrule{6-7} \cmidrule{8-9}
      & Iters. & Time (ms)
      & Iters. & Time (ms)
      & Iters. & Time (ms)
      & Iters. & Time (ms) \\
      \midrule
      \textbf{Ours} & \textbf{1} & \textbf{4.21} & \textbf{1} & \textbf{4.06} & \textbf{1} & \textbf{4.45} & \textbf{1} & \textbf{0.15} \\
      Quasi-Newton (BFGS) & 7  & 26.02 & 9  & 35.31 & 6  & 28.18 & 7  & 0.63 \\
      Conjugate Gradient  & 8  & 21.33 & 7  & 19.20 & 7  & 20.41 & 6  & 0.47 \\
      CMA-ES              & 37 & 127.45 & 29 & 113.21 & 45 & 192.34 & 23 & 4.28 \\
      (1+1)-ES            & 200 & 177.08 & 200 & 153.83 & 200 & 147.61 & 200 & 12.15 \\
      \bottomrule
    \end{tabular}
  }
\end{table}

The IK objective $\Psi(\mathcal{L},\mathbf{T}^{*})$ is generally nonconvex because it composes FK with closed-chain constraints and measures pose discrepancy through the Lie-algebra mapping on $\mathrm{SE}(3)$.
Accordingly, as with other local optimizers such as BFGS and CG, our Gauss--Seidel-style solver (i.e., a coordinate-descent method) is not guaranteed to find a global optimum.
In repeated application workflows, our goal is not a global-optimality guarantee but robust and efficient convergence to a high-quality feasible solution under a consistent stopping~criterion.

As shown in Figure~\ref{fig:ik_results} and quantified in Table~\ref{tab:simulation_results}, our method reduces $\Psi$ to a low value within a small number of iterations on all tested robots.
BFGS and conjugate gradient show similar early-stage progress but typically require more iterations near the optimum.
In practice, BFGS and conjugate gradient operate directly on the original high-dimensional coupled variable space, where multiple actuator variables are optimized jointly.
By contrast, the main advantage of our method lies in the ITEP-based decomposition of the original bound-constrained IK problem into a sequence of low-dimensional subproblems within a Gauss--Seidel-style outer framework.
This decomposition substantially reduces the difficulty of each update and leads to faster convergence in our benchmark.
Moreover, gradient-based methods typically require gradient information, line search, and additional handling for bound constraints, whereas our method avoids these requirements by solving bounded one-dimensional subproblems directly.
The global methods are gradient free, but they require substantially more iterations and runtime on these problems, making them less suitable for repeated actuator-level kinematic queries in the tested planning, training, and digital-twin workflows as shown in Table~\ref{tab:simulation_results}.

Robots whose end-effector is jointly affected by multiple actuators, such as the canopies of TCCHS and SSHS driven by the front legs ($A_0,A_1$) and rear legs ($A_2,A_3$), and the canopy of STHS driven by the legs ($A_0,A_1$) and a balance jack ($A_2$), typically require more iterations due to coupling.
This behavior is expected because updating one actuator changes the best response of others, so multiple sweeps are needed before the iterates stabilize.
This distinction is important: stronger coupling mainly affects convergence speed, but it does not necessarily imply a higher risk of converging to poor local minima.
To assess this empirically, we report success rates and final residual statistics over 100 randomized trials per robot in Appendix~\ref{appendix:B21} (Table~\ref{tab:ik_robustness_compact}).
We also include a near-limit stress test, denoted Ours-NL, where actuator lengths are sampled within the nearest \(5\%\) of their MRDF-defined stroke bounds.
Across all robots, including the coupled hydraulic supports, Ours, Ours-NL, BFGS, and CG converge reliably under the same tolerance, and the reported statistics show no systematic degradation of the final solutions on coupled or near-limit feasible~configurations.

To separate optimization convergence from pose accuracy, we further construct two target sets.
Reachable targets are generated by sampling actuator lengths within MRDF bounds and applying FK; therefore, each target has a feasible actuator-length solution.
For these targets, a trial is counted as successful only if the final weighted residual satisfies \(\Psi^*<10^{-6}\).
Unreachable targets are generated by translating reachable FK poses outward by a controlled random offset and are used to evaluate numerical convergence when exact pose matching is not expected.
For these targets, we report runtime and convergence rate according to the objective-change criterion; 
final residuals are not used as pose-accuracy metrics because their magnitudes depend on the random distance from the feasible workspace.

\begin{figure}[htbp]
  \centering
    \includegraphics[width=1.0\linewidth]{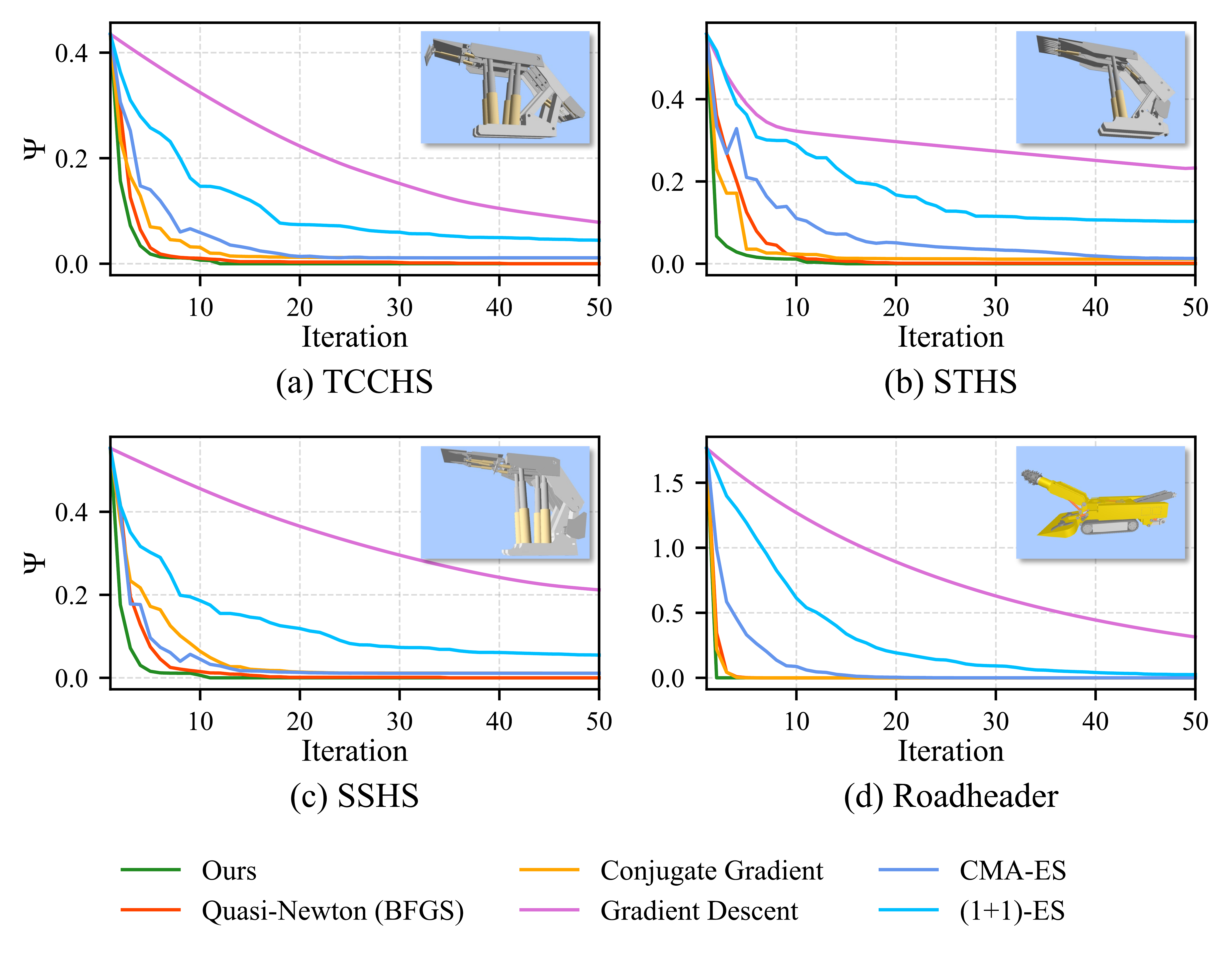}
    \caption{IK convergence curves on four representative mining robots:
    (\textbf{a}) TCCHS, (\textbf{b}) STHS, (\textbf{c})~SSHS, and (\textbf{d}) Roadheader.
    Each curve shows how the IK objective \(\Psi\) decreases over outer iterations for one optimization method in the randomized IK benchmark.
    The corresponding numerical iteration counts and runtime statistics are reported in Table~\ref{tab:simulation_results}.}
    \label{fig:ik_results}  
\end{figure}

Table~\ref{tab:ik_reachable_unreachable} separates pose accuracy from numerical convergence.
For the reachable target set, all robots achieve a \(100\%\) pose-success rate, and both the median and P95 final residuals remain below \(10^{-6}\), indicating that the numerical stopping criterion leads to accurate final poses in these feasible trials.
For the unreachable target set, all trials reach numerical convergence under the same objective-change criterion.
In this case, the returned actuator lengths should be interpreted as best-effort bound-constrained solutions rather than exact pose-matching results.

\begin{table}[htbp]
  \caption{\textbf{IK performance on reachable and unreachable targets} (100 trials per target set). Reachable targets are generated by sampling actuator lengths within MRDF bounds and applying FK; Succ.\ denotes pose success with \(\Psi^*<10^{-6}\). Unreachable targets are generated by translating reachable FK poses outward by a controlled random offset; Conv.\ denotes numerical convergence according to the objective-change criterion within the iteration budget. Final residual statistics are reported only for reachable targets because unreachable targets do not have feasible actuator-length solutions, and their residual magnitudes are dominated by the sampling distance from the feasible workspace.}
  \label{tab:ik_reachable_unreachable}
  \centering 
  \resizebox{\textwidth}{!}{
    \begin{tabular}{lrrrrrr}
      \toprule
      \multirow{2}{*}{ }
      & \multicolumn{4}{c}{\textbf{Reachable Targets}}
      & \multicolumn{2}{c}{\textbf{Unreachable Targets}} \\
      \cmidrule{2-5}\cmidrule{6-7}
      & \textbf{Avg. Time (ms)} & \textbf{Succ. (\%)} & \textbf{Med.} \boldmath{\(\Psi^*\)} \textbf{(}\boldmath{$10^{-6}$}\textbf{)} & \textbf{P95} \boldmath{\(\Psi^*\)} \textbf{(}\boldmath{$10^{-6}$}\textbf{)}
      & \textbf{Avg. Time (ms)} & \textbf{Conv. (\%)} \\
      \midrule
      TCCHS & 69.45 & 100\% & 0.23 & 0.68 &  85.16 & 100\% \\
      STHS  & 94.21 & 100\% & 0.27 & 0.64 & 103.82 & 100\% \\
      SSHS  & 75.43 & 100\% & 0.26 & 0.69 &  78.94 & 100\% \\
      RH    &  0.20 & 100\% & 0.06 & 0.15 &   0.23 & 100\% \\
      LHD   &  4.41 & 100\% & 0.09 & 0.20 &   4.56 & 100\% \\
      VD    &  4.23 & 100\% & 0.05 & 0.23 &   4.52 & 100\% \\
      DJ    &  4.28 & 100\% & 0.07 & 0.19 &   4.66 & 100\% \\
      SH    &  0.17 & 100\% & 0.06 & 0.16 &   0.16 & 100\% \\
      \bottomrule
    \end{tabular}
  }
\end{table}

In summary, our IK solver exploits ITEP independence to decompose the bound-constrained problem into a sequence of bounded one-dimensional subproblems per iteration.
It converges efficiently on reachable targets with small final residuals and remains numerically stable on the tested unreachable targets, where the returned actuator lengths are interpreted as best-effort bound-constrained solutions.

\subsection{Comparison with General-Purpose Robotics Formats and Frameworks}
To better position MineRobot against widely used robotics formats and frameworks, we compare them from two perspectives.
First, we provide a qualitative comparison of representation-level modeling expressiveness and framework-level functionality and support for the kinematics of mining robots.
Second, we report quantitative runtime evaluations on the same set of mining robots and tasks.

\subsubsection{Modeling Expressiveness and Framework Support} \label{sec:comparison-expressiveness}

To better position MineRobot with respect to widely used robotics descriptions and frameworks, we compare prior work from two complementary perspectives.
Table~\ref{tab:format_comparison} focuses on \emph{\hl{representation-level semantics} 
}, i.e., whether a robot description format itself can natively encode the structures that are central to underground mining robots.
Table~\ref{tab:tool_comparison}, in contrast, focuses on \emph{\hl{framework-level functionality and support}}, i.e., whether a robotics framework natively provides solver interfaces or dedicated handling pipelines for such mechanisms.
This distinction is important: a format may be expressive at the representation level, while a framework may still lack a native pipeline to directly exploit that information.

\begin{table}[htbp]
  \caption{Comparison of robot description formats in terms of kinematics-modeling expressiveness for mining robots.}
  \label{tab:format_comparison}
  \centering
  \resizebox{0.8\textwidth}{!}{
    \begin{tabular}{lcccc}
      \toprule
      \textbf{Format} & \shortstack[c]{\textbf{Closed Loop} \\ \textbf{Support}} & \shortstack[c]{\textbf{Linear Actuator} \\ \textbf{Semantics}} & \shortstack[c]{\textbf{Four-Bar} \\ \textbf{Semantics}} & \shortstack[c]{\textbf{Redundant Actuator} \\ \textbf{Semantics}} \\
      \midrule
      \textbf{MRDF} & \cmark & \cmark & \cmark & \cmark \\
      URDF~\cite{tola2023understanding} & \xmark & \pmark & \xmark & \pmark \\
      SDFormat~\cite{sdformat} & \cmark & \pmark & \pmark & \pmark \\
      MJCF~\cite{todorov2012mujoco} & \cmark & \cmark & \pmark & \pmark \\
      USD~\cite{usd_pixar} & \pmark & \pmark & \pmark & \pmark \\
      \bottomrule
    \end{tabular}
  }
  
  \vspace{2mm}
  \noindent{\footnotesize{\textit{Note:} 
    This table compares representation-level semantics.
    \cmark~Native (directly supported as a first-class semantic concept); 
    \pmark~Conditional (representable, but not as a first-class semantic concept); 
    \xmark~No (not supported).
  }}
\end{table}

Table~\ref{tab:format_comparison} shows that existing general-purpose robot description formats only partially cover the kinematics patterns prevalent in mining robots.
URDF is effective for tree-structured robots, but it does not natively support closed loops, and its support for linear actuators and redundant actuation is only indirect.
SDFormat and MJCF improve the situation by allowing closed-loop mechanisms to be encoded in the file, and MJCF further provides an explicit cylinder actuator primitive.
However, for these formats, planar four-bars and redundant actuator groups are still not treated as first-class semantic concepts.
USD is highly flexible as a general scene and physics representation, but from the viewpoint of mining-robot kinematics it remains largely generic and indirect.
In contrast, MRDF natively represents all four structures considered here: closed-loop mechanisms, linear actuators, four-bars, and redundant actuators.
This comparison clarifies the role of MRDF in our framework: it is not introduced merely as another robot file format but as a domain-specific representation whose semantics are aligned with the dominant kinematic patterns of underground mining robots.

This distinction is particularly important for linear actuators and redundant actuation.
In general-purpose formats, linear actuation is usually represented indirectly through prismatic joints, drives, generic actuators, or controllers, while redundancy is typically approximated through mimic relations, equality constraints, or external synchronization logic.
Such encodings may be sufficient for simulation or control in specific software stacks, but they do not explicitly expose the actuator-level topology required by our topology construction and solver scheduling.
By contrast, MRDF treats the actuator itself as a first-class modeling entity, together with its mounting relations, bounds, and redundancy-group membership.
This actuator-centered representation is the key enabler of the subsequent four-bar contraction, ITEP extraction, and sequential FK/IK pipeline developed in MineRobot.

\begin{table}[htbp]
  \caption{Comparison of general-purpose robotics frameworks in terms of kinematics support for mining robots.}
  \label{tab:tool_comparison}
  \centering 
  \resizebox{\textwidth}{!}{
    \begin{tabular}{llcccccc}
      \toprule
      \textbf{Framework} & \textbf{Formats} & \shortstack[c]{\textbf{FK} \\ \textbf{Solver}} & \shortstack[c]{\textbf{IK} \\ \textbf{Solver}} & \shortstack[c]{\textbf{Closed Loop} \\ \textbf{Pipeline}} & \shortstack[c]{\textbf{Actuator-Centered} \\ \textbf{Kinematics}} & \shortstack[c]{\textbf{Four-Bar} \\ \textbf{Handling}} & \shortstack[c]{\textbf{Redundant Actuator} \\ \textbf{Handling}} \\
      \midrule
      \textbf{MineRobot} & MRDF                 & \cmark & \cmark & \cmark & \cmark & \cmark & \cmark \\
      Drake~\cite{drake}               & URDF, SDFormat, etc.       & \cmark & \cmark & \pmark & \xmark & \xmark & \xmark \\
      MuJoCo~\cite{todorov2012mujoco}             & MJCF, URDF           & \cmark & \xmark & \cmark & \xmark & \xmark & \pmark \\
      Gazebo~\cite{koenig2004gazebo}             & SDFormat, etc.       & \pmark & \xmark & \pmark & \xmark & \pmark & \pmark \\
      PyBullet~\cite{coumans2021pybullet}           & URDF, SDFormat, etc.            & \cmark & \cmark & \pmark & \xmark & \xmark & \xmark \\
      Isaac Sim~\cite{nvidia_isaac_sim}          & USD, etc.                  & \cmark & \cmark & \pmark & \xmark & \xmark & \xmark \\
      DART~\cite{lee2018dart}               & URDF, SDFormat, etc.       & \cmark & \cmark & \pmark & \xmark & \xmark & \xmark \\
      \bottomrule
    \end{tabular}
  }
   
  \vspace{2mm}
  \noindent{\footnotesize{ \textit{Note:}  
    This table compares framework-level functionality and support.  
    \cmark~Native (natively provided through solver interfaces, import capability, or dedicated handling pipelines);  
    \pmark~Conditional (partially supported under additional conditions);  
    \xmark~No (not supported).
  }}
\end{table}

Table~\ref{tab:tool_comparison} complements the format-level analysis by comparing general-purpose robotics frameworks in terms of mining-robot kinematics support. Here, \emph{\hl{FK solver}} and \emph{\hl{IK solver}} indicate native solver interfaces; \emph{\hl{closed loop pipeline}} indicates whether closed-loop information can be imported from supported standard formats without manual constraint coding; and \emph{\hl{actuator-centered kinematics}}, \emph{\hl{four-bar handling}}, and \emph{\hl{redundant actuator handling}} indicate whether these mining-robot-specific structures are natively incorporated into the kinematics pipeline. A fair reading of the table is not that existing frameworks are weak in general. On the contrary, several of them provide strong support for standard kinematics tasks and closed-loop modeling under various conditions. Among them, Drake and MuJoCo are particularly representative: Drake provides strong native support for general FK/IK and multibody modeling, while MuJoCo provides native support for closed-loop import through MJCF together with explicit actuator modeling. Therefore, MineRobot is not motivated by the absence of general FK/IK functionality in the literature.

Instead, the main gap lies in the combination of capabilities required by mining robots, and more importantly, in how these capabilities are organized into the kinematics pipeline. As Table~\ref{tab:tool_comparison} shows, none of the compared general-purpose frameworks natively provide actuator-centered kinematics, dedicated four-bar handling, and redundant-actuator-aware processing in a unified pipeline. Even when closed loops can be represented or imported, they are generally handled as generic constraints within a joint-centric formulation. As a result, the corresponding FK/IK problems remain coupled at the mechanism level, and the user typically has to rely on generic numerical solvers over the full constrained system.

MineRobot differs precisely in this respect. By contracting four-bars into generalized joints and extracting Independent Topologically Equivalent Paths (ITEPs), it reorganizes the original coupled closed-loop mechanism into actuator-centered subproblems associated with independent actuator groups. This decomposition turns the original kinematics problem into a sequence of small, topology-aware local solves, which improves not only computational efficiency but also numerical robustness and solver stability in interactive settings. Consequently, FK no longer needs to resolve the entire mechanism as one generic constrained system, and IK can alternate over actuator variables in a Gauss--Seidel-style manner aligned with the same decomposition. In other words, MineRobot does not merely add support for mining-robot-specific mechanisms; it uses their topology to decouple the kinematics computation itself.

Taken together, Tables~\ref{tab:format_comparison} and~\ref{tab:tool_comparison} highlight the intended scope of our contribution.
General-purpose formats and frameworks are broad and versatile, but they are not designed around the characteristic kinematics of underground mining robots.
MineRobot addresses this gap by coupling a domain-specific representation (MRDF) with a topology-aware kinematics pipeline.
The practical consequence is that the information needed by our FK/IK solvers is available natively in the representation and can be exploited directly, rather than being reconstructed indirectly from generic joints, constraints, or external control logic.
This stronger alignment between modeling semantics and solver design is the main reason why MineRobot provides both clearer authoring for mining robots and a more direct route to efficient closed-loop kinematics, which will be reflected in the runtime comparison of Section~\ref{comparison-runtime}.

\subsubsection{Runtime Comparison} \label{comparison-runtime}

We next compare runtime performance with representative general-purpose frameworks.
Based on the analysis in Section~\ref{sec:comparison-expressiveness}, we choose Drake and MuJoCo as the most informative baselines.
Drake provides strong native support for generic FK/IK and multibody modeling, whereas MuJoCo provides native support for closed-loop import through MJCF together with an explicit cylinder actuator primitive.
They therefore represent two strong yet complementary points of comparison for MineRobot.
This comparison is a workflow-level reference under matched robot geometry, actuator bounds, target poses, tolerances, and time budgets, rather than a claim about framework-wide performance limits.
In this comparison, the success rate is defined as the percentage of randomized trials that satisfy the prescribed stopping criterion within the common \(1\,\mathrm{s}\) time budget.
Detailed reference-workflow settings are provided in Appendix~\ref{appendix:experiment-details}.
For FK, success requires all active actuator-length residuals and loop-closure residuals to be below \(10^{-6}\).
For IK, success requires the objective-change convergence criterion in Algorithm~\ref{alg:ik} to be reached while all actuator lengths remain within their bounds.

Before comparing runtime, we first validate FK consistency against an independent Drake reference workflow.
For the FK validation, each robot is evaluated on 100 actuator-length samples drawn within the MRDF-defined admissible bounds.
The same actuator-length inputs are solved by MineRobot and a Drake-based closed-chain reference model on the same reference assembly branch.
We define \(e_{\mathrm{FK}}\) as the weighted SE(3) pose residual between the MineRobot and Drake end-effector transformations, using the same translational normalization length \(L_c\) as in the IK objective.

Table~\ref{tab:fk_drake_validation} shows that MineRobot produces FK poses consistent with the independent Drake reference workflow across all tested robots.
This validation is complementary to the runtime comparison below: Table~\ref{tab:fk_drake_validation} checks kinematic consistency under the same actuator-length inputs, whereas Table~\ref{tab:engine_runtime_comparison} compares computational efficiency and observed convergence rates.

\begin{table}[htbp]
  \caption{\textbf{FK validation against an independent Drake reference workflow}. The table reports median and P95 weighted SE(3) pose differences \(e_{\mathrm{FK}}\) over 100 actuator-length samples per robot. All Drake reference solves converged for the sampled cases.}
  \label{tab:fk_drake_validation}
  \centering
  \resizebox{0.6\textwidth}{!}{
    \begin{tabular}{lrrrrrrrr}
      \toprule
      & \textbf{TCCHS} & \textbf{STHS} & \textbf{SSHS} & \textbf{RH} & \textbf{LHD} & \textbf{VD} & \textbf{DJ} & \textbf{SH} \\
      \midrule
      Med. \(e_{\mathrm{FK}}\) (\(10^{-6}\)) & 0.39 & 0.35 & 0.28 & 0.13 & 0.11 & 0.25 & 0.19 & 0.20 \\
      P95 \(e_{\mathrm{FK}}\) (\(10^{-6}\)) & 0.75 & 0.69 & 0.57 & 0.33 & 0.35 & 0.42 & 0.45 & 0.47 \\
      \bottomrule
    \end{tabular}
  }
\end{table}

\begin{table}[htbp]
  \caption{\textbf{Runtime comparison with MineRobot, Drake, and MuJoCo} over 100 randomized trials per robot. Under each robot, FK and IK entries report the average runtime in milliseconds, with the success rate in parentheses. All methods use the same task setup, tolerance, and time budget. Iteration counts are omitted because the compared frameworks use substantially different internal update mechanisms, making runtime and success rate more meaningful for comparison. \(^{*}\) indicates entries obtained from an external hand-implemented reference workflow outside the native engine interface; these entries are reported only as workflow-level references.}
  \label{tab:engine_runtime_comparison}
  \centering 
  \resizebox{\textwidth}{!}{
    \begin{tabular}{lcccccccc}
      \toprule
      & \multicolumn{2}{c}{\textbf{TCCHS}}
      & \multicolumn{2}{c}{\textbf{STHS}}
      & \multicolumn{2}{c}{\textbf{SSHS}}
      & \multicolumn{2}{c}{\textbf{RH}} \\
      \cmidrule{2-9}
      & \multicolumn{1}{c}{\textbf{FK}} & \multicolumn{1}{c}{\textbf{\textbf{IK}}}
      & \multicolumn{1}{c}{\textbf{FK}} & \multicolumn{1}{c}{\textbf{IK}}
      & \multicolumn{1}{c}{\textbf{FK}} & \multicolumn{1}{c}{\textbf{IK}}
      & \multicolumn{1}{c}{\textbf{FK}} & \multicolumn{1}{c}{\textbf{IK}} \\
      \midrule
      \textbf{MineRobot}
        & \textbf{1.06 (100\%)} & \textbf{78.63 (100\%)}
        & \textbf{0.56 (100\%)} & \textbf{96.27 (100\%)}
        & \textbf{0.87 (100\%)} & \textbf{73.91 (100\%)}
        & \textbf{0.01 (100\%)} & \textbf{0.21 (100\%)} \\
      Drake
        & 25.35 (100\%) & 252.62 (86\%)
        & 11.84 (100\%) & 141.30 (92\%)
        & 18.72 (100\%) & 128.45 (90\%)
        & 0.26 (100\%)  & 0.92 (100\%) \\
      MuJoCo
        & 6.23 (100\%) & 189.40 (82\%) \rlap{$^{*}$}
        & 3.15 (100\%) & 151.72 (89\%) \rlap{$^{*}$}
        & 4.78 (100\%) & 137.55 (87\%) \rlap{$^{*}$}
        & 0.07 (100\%) & 0.78 (100\%) \rlap{$^{*}$} \\
      \midrule
      & \multicolumn{2}{c}{\textbf{LHD}}
      & \multicolumn{2}{c}{\textbf{VD}}
      & \multicolumn{2}{c}{\textbf{DJ}}
      & \multicolumn{2}{c}{\textbf{SH}} \\
      \cmidrule{2-9}
      & \multicolumn{1}{c}{FK} & \multicolumn{1}{c}{IK}
      & \multicolumn{1}{c}{FK} & \multicolumn{1}{c}{IK}
      & \multicolumn{1}{c}{FK} & \multicolumn{1}{c}{IK}
      & \multicolumn{1}{c}{FK} & \multicolumn{1}{c}{IK} \\
      \midrule
      \textbf{MineRobot}
        & \textbf{0.18 (100\%)} & \textbf{4.21 (100\%)}
        & \textbf{0.18 (100\%)} & \textbf{4.06 (100\%)}
        & \textbf{0.17 (100\%)} & \textbf{4.45 (100\%)}
        & \textbf{0.01 (100\%)} & \textbf{0.15 (100\%)} \\
      Drake
        & 1.94 (100\%)  & 31.85 (98\%)
        & 1.73 (100\%)  & 39.60 (97\%)
        & 1.88 (100\%)  & 33.12 (98\%)
        & 0.22 (100\%)  & 0.96 (100\%) \\
      MuJoCo
        & 0.62 (100\%) & 26.44 (97\%) \rlap{$^{*}$}
        & 0.58 (100\%) & 32.18 (96\%) \rlap{$^{*}$}
        & 0.61 (100\%) & 27.06 (97\%) \rlap{$^{*}$}
        & 0.06 (100\%) & 0.84 (100\%) \rlap{$^{*}$} \\
      \bottomrule
    \end{tabular}
  }
\end{table}

Table~\ref{tab:engine_runtime_comparison} reports the average runtime and success rate over 100 randomized trials for each robot.
All methods use the same robot instances, task setup, tolerance, and time budget, and success rates are computed using the definitions stated above.
We emphasize runtime and success rate rather than iteration counts because the compared frameworks rely on substantially different internal update mechanisms, making direct iteration-to-iteration comparison less meaningful.
For MuJoCo, the IK results marked with $^{*}$ are reported only for reference, since MuJoCo does not natively provide the corresponding IK workflow and the measurements are obtained from a hand-implemented C++/Eigen~3.4.0 version.
More generally, the reported implementations are research-oriented and not specially optimized; therefore, the absolute runtimes should be interpreted as reference values, whereas the relative trends and success rates are the main points of comparison.

Several consistent trends can be observed.
First, MineRobot achieves the best overall FK efficiency across all tested robots while maintaining a 100\% success rate.
This advantage is especially clear on the more complex hydraulic-support robots, but it also remains consistent on simpler mechanisms.
Second, the advantage becomes more pronounced in IK, where MineRobot again attains 100\% success on all robots and consistently outperforms the compared baselines in runtime.
By contrast, Drake and the MuJoCo-based reference implementation show lower success rates on several of the more strongly coupled robots, indicating that the practical difference is not only computational speed but also solution~robustness.

The main reason for this advantage is the topology-aware decomposition introduced by MineRobot.
As discussed in Section~\ref{sec:comparison-expressiveness}, general-purpose frameworks typically handle these robots through joint-centric formulations and generic closed-loop constraints, so FK/IK remains a coupled problem at the mechanism level.
In contrast, MineRobot contracts four-bars, extracts ITEPs, and solves the resulting actuator-centered subproblems sequentially.
This decoupling substantially reduces the size of each local solve in FK and enables the Gauss--Seidel-style IK solver to operate directly on actuator variables rather than repeatedly resolving the full constrained system.
The runtime comparison therefore provides quantitative evidence that the ITEP-based decomposition improves not only efficiency but also stability under repeated randomized tests.

At the same time, these results should be interpreted within scope.
Drake and MuJoCo are general-purpose frameworks designed for a much broader class of robots and simulation tasks, whereas MineRobot is specialized for the dominant kinematic patterns of underground mining robots.
Our claim is therefore not that general-purpose frameworks are ineffective, but that their broad formulations do not directly exploit the structural regularities of mining robots.
These computational properties also have implications for deployment.
Because FK/IK are decomposed into bounded low-dimensional subproblems with fixed iteration budgets, the CPU-based C++ implementation suggests potential for deployment on low-power edge-computing or embedded platforms.
Nevertheless, platform-specific profiling, implementation optimization, and controller-stack integration are still required before claiming embedded real-time performance.

\section{Applications} \label{sec:apps}

We briefly demonstrate how MineRobot supports several representative underground-mining-robot application workflows.
As summarized in Figure~\ref{fig:apps}, the framework can be integrated into workspace analysis, trajectory generation, operator-training interfaces, and digital-twin workflows.
These applications are enabled by the same actuator-centered kinematics layer: FK maps actuator states to robot configurations for visualization and interaction, while IK maps task-space targets to actuator-length references for trajectory generation and command planning.

\begin{figure}[htbp]
\centering 
    \includegraphics[width=1.0\linewidth]{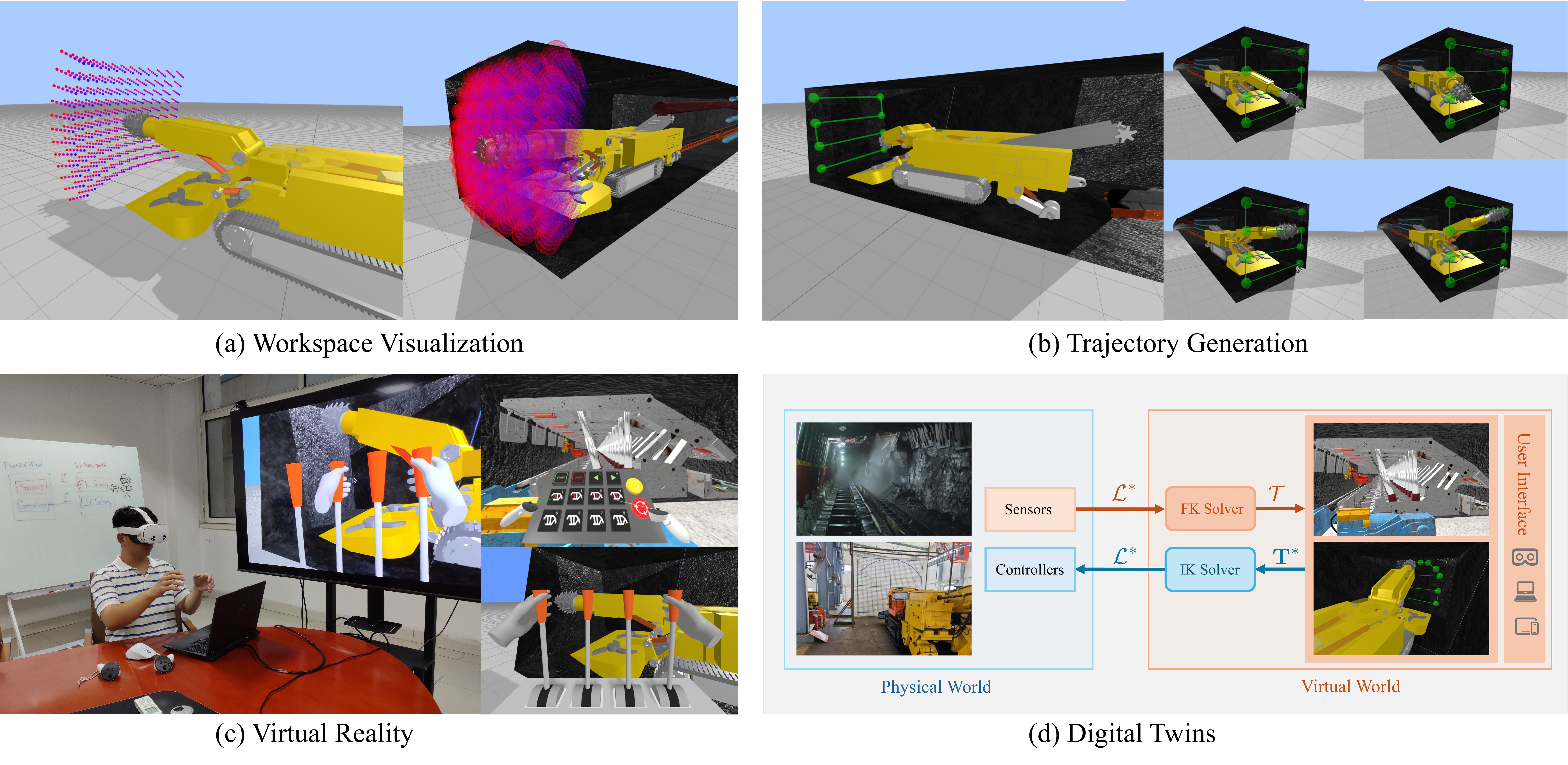}
    \caption{\hl{Representative} 
 application workflows enabled by MineRobot.
    (\textbf{a}) Workspace visualization using forward-kinematics (FK)-based sampling of reachable end-effector poses.
    (\textbf{b}) Trajectory generation by solving inverse kinematics (IK) along a target path.
    (\textbf{c}) Virtual reality interaction with real-time robot updates driven by FK.
    (\textbf{d}) Kinematics-layer synchronization workflow between physical robots and virtual counterparts through FK/IK.}
    \label{fig:apps}
\end{figure}

Specifically, workspace visualization (Figure~\ref{fig:apps}a) is realized by sampling the relevant actuator space and evaluating FK to obtain the corresponding end-effector poses.
Trajectory generation (Figure~\ref{fig:apps}b) is realized by discretizing a target path in task space and solving IK at each sample.
For virtual reality (VR) interaction (Figure~\ref{fig:apps}c), virtual controls are mapped to actuator targets, and the robot is updated online through FK at interactive rates.
For digital twins (Figure~\ref{fig:apps}d), measured actuator lengths are sent to FK for physical-to-virtual synchronization, while desired end-effector targets are sent to IK for virtual-to-physical command generation.
These examples illustrate how MineRobot serves as an actuator-centered kinematics layer for repeated FK/IK queries in representative underground-mining-robot workflows.
Additional implementation details are provided in Appendix~\ref{appendix:C}.

\section{Conclusions and Future Work} 
\label{sec:conclusion}

We presented MineRobot, an actuator-centered framework for modeling and solving the kinematics of a representative class of underground mining robots.
We introduced the Mining Robot Description Format (MRDF), a domain-specific language for parametric kinematic modeling that explicitly supports linear actuators and four-bar linkages as the framework's data backbone.
On top of MRDF, we developed a topology pipeline that contracts four-bar linkages into generalized joints and automatically extracts per-actuator Independent Topologically Equivalent Paths (ITEPs) in four types.
Leveraging ITEP independence, MineRobot composes per-type solvers into a sequential forward-kinematics pipeline and formulates inverse kinematics as an FK-aligned, bound-constrained optimization problem solved by a Gauss--Seidel-style actuator-length update scheme.
Experiments on representative robots show real-time FK performance and robust IK convergence within the tested operating ranges, supporting the use of MineRobot as an actuator-centered kinematics layer for planning, virtual training, and digital-twin workflows.

The present formulation is intentionally limited to a specific mechanism class.
It assumes linear-actuator-driven mechanisms composed mainly of revolute, prismatic, and fixed joints, with planar four-bar linkages that can be contracted into generalized joints and actuator-free topologies that are acyclic or satisfy the independent-path conditions required for ITEP extraction.
Accordingly, MineRobot is not intended as a general-purpose solver for arbitrary closed-chain robots.
Mechanisms with arbitrary spatial multi-loop closures, unsupported joint types, deliberate assembly-mode switching through singular branch transitions, or actuator-free topologies that cannot be decomposed into independent actuator paths are outside the current scope.
The framework is also limited to internal actuator-driven mechanism kinematics: for each FK/IK query, the base or chassis frame is treated as a given reference frame, and MineRobot solves the closed-chain motion of the robot mechanism relative to that frame.
The experiments in this study are numerical kinematic evaluations rather than physical machine tests; therefore, the reported results should be interpreted as validating the actuator-centered kinematic modeling and FK/IK computation layer under idealized model parameters and prescribed actuator bounds.
In actual underground applications, additional discrepancies may arise from geometric calibration errors, sensor noise and delay in actuator-length measurements, joint clearance and wear, structural deformation under heavy loads, hydraulic compliance and friction, base-frame uncertainty, contact with the roof, floor, or coal wall, and low-level controller dynamics.
These effects can influence the agreement between the computed kinematic configuration and physical machine motion, and should be handled by calibration, state estimation, feedback control, hydraulic-actuator models, collision/contact checking, safety supervision, or higher-fidelity multibody simulation modules that use MineRobot's FK/IK routines as an actuator-centered kinematic sublayer.

Although this study focuses on underground mining robots, the actuator-centered topology-processing idea may be transferable to other heavy-duty robots and construction equipment that use linear actuators, planar four-bar linkages, and similar closed-chain transmission structures.
Future work will focus on physical validation with real mining equipment, calibration and sensing integration, mobile-base kinematics and trajectory planning, interfaces with general robotics frameworks and multibody simulation tools, deployment evaluation on low-power edge-computing and embedded platforms, physics-based actuator/contact models, and MRDF extensions for robot groups and operating environments.

\vspace{6pt}

\section*{Funding}
This research was funded by the National Natural Science Foundation of China, grant number 52574195; the Natural Science Foundation of Shandong Province, grant number ZR2023ME032; and the Central Funds Guiding the Local Science and Technology Development of Shandong Province, grant number YDZX2025101, under the project ``Research and Application of the Surveying and Localization Robotic System for Underground Engineering''.




\appendix

\section{Details on ITEP Completeness and Limitations}
\label{appendix:A}

This appendix provides the detailed argument behind the completeness claim of the four ITEP types (A--D) used in Section~\ref{sec:topo}. Specifically, we show that, under the robot scope assumptions in Section~\ref{sec:modeling}, the set of links and joints influenced by a given actuator redundancy group when all other groups are locked must fall into one of the four ITEP types. We also summarize the main limitations and out-of-scope cases where these assumptions may be violated and the taxonomy may no longer apply.

\subsection{Completeness of ITEPs}

\textbf{\hl{Preliminaries and notation.}
}
We restrict attention to robots within the scope of \mbox{Section~\ref{sec:modeling}}. 
After preprocessing (merging fixed-joint chains) and planar four-bar contraction, all remaining joints on the actuator-free graph are 1-DoF joints of type $R$, $P$, or generalized $G$ (a contracted four-bar transmission), and the graph is acyclic by Assumption~\ref{asp:3}.
As discussed in Section~\ref{sec:topo-algo}, the ITEP extraction is applied within the target mining-robot topologies in which the relevant endpoint links are connected by a unique simple topological path: tree-like cases satisfy this directly, and cases involving generalized joints are restricted by the design condition that actuator paths do not pass through already actuated generalized joints.
For an actuator $A_j$, let $L_{p_t}$ and $L_{p_r}$ be the tube/rod parent links, and let $P$ be this unique simple path between them in the contracted topology, with $|P|$ the number of links on $P$.

\begin{lemma}\label{lem:lock-loop}
Locking an actuator $A_j$ is kinematically equivalent to merging its tube and rod into a single rigid link $L_{rt}$ attached to $L_{p_t}$ and $L_{p_r}$ by two revolute joints (Figure~\ref{fig:topo-build-other-actuators}). If $|P|=2$, the locked actuator is equivalent, after constraint reduction, to fixing the direct joint between $L_{p_t}$ and $L_{p_r}$, so no independent local loop is retained in the reduced topology; if $|P|>2$, locking creates a local loop consisting of $P$ and the two revolute attachments via $L_{rt}$.
\end{lemma}
\begin{proof}
Fixing the rod--tube relative motion removes the actuator's internal DoF, so the tube and rod behave as one rigid body. They remain connected to the mechanism through their mounting revolute joints, yielding the merged link $L_{rt}$. 
If $L_{p_t}$ and $L_{p_r}$ are adjacent ($|P|=2$), the locked actuator and the direct parent--parent joint form only a zero-mobility constraint that is reduced to a fixed joint in Algorithm~\ref{alg:topo}; therefore no independent local loop remains to be solved. Otherwise, the existing path $P$ together with the new two-edge connection through $L_{rt}$ forms a cycle.
\end{proof}

\begin{lemma}\label{lem:path-length-3}
Assume all actuator redundancy groups except one are locked. If locking creates a local loop as in Lemma~\ref{lem:lock-loop}, then the unique path $P$ between the endpoints satisfies $|P|=3$.
\end{lemma}
\begin{proof}
By Assumption~\ref{asp:1}, after locking all other groups the mechanism has total DoF equal to 1. Let $n$ and $j$ denote the numbers of links and 1-DoF joints in the induced loop, respectively. The loop contains the $|P|$ links on $P$ plus the merged rigid link $L_{rt}$, hence $n=|P|+1$. It contains $|P|-1$ joints along $P$ and two additional revolute joints incident to $L_{rt}$, and hence $j=|P|+1$. Under our scope, this loop is planar, so the planar Chebychev--Gr\"ubler--Kutzbach criterion gives
\[
\mathrm{DoF}=3(n-1)-2j=3|P|-2(|P|+1)=|P|-2.
\]

Since $\mathrm{DoF}=1$, we obtain $|P|=3$.
\end{proof}

\begin{theorem}\label{thm:itep-complete}
Under the scope assumptions of Section~\ref{sec:modeling}, every actuator-induced structure obtained by locking all other actuator redundancy groups falls into one of the four ITEP types (A--D).
\end{theorem}

\begin{proof}
Lock all actuator groups except the current one $\mathcal{G}$; by Assumption~\ref{asp:1} the remaining DoF is 1. There are two cases.
If no local loop is created during locking, the contracted graph remains acyclic. Since the active mechanism has only 1 DoF, the influence of the current actuator must reduce to a unique 1-DoF transmission between its endpoint parents $(L_{p_t},L_{p_r})$. In the contracted representation, such a transmission can only be prismatic, revolute, or generalized, i.e., $P$, $R$, or $G$, which yields Type A/B/C.
Otherwise, by Lemma~\ref{lem:lock-loop} a local loop is created, and by Lemma~\ref{lem:path-length-3} it must be a four-link loop, i.e., a generalized four-bar in our taxonomy, which yields Type D.
\end{proof}

\subsection{Limitations and Out-of-Scope Cases}

The completeness result above holds only within the target scope and has the following~limitations.

\paragraph{\hl{(L1)} 
 Spatial or multi-DoF closed chains}
MineRobot assumes planar single-DoF four-bars as the dominant loop primitive. Spatial closed chains and multi-DoF loops are out of scope since contraction into a 1-DoF generalized joint is no longer valid and the corresponding ITEP may not reduce to a 1D~subproblem.

\paragraph{\hl{(L2)} Topologies that remain multi-cyclic after contraction}
If the actuator-free topology remains multi-cyclic after contraction, the locked actuator-induced structure may not be a simple path or a single 1-DoF loop, and the A--D taxonomy may fail or become non-unique.

\paragraph{\hl{(L3)} Theoretical variants of the indirect-lock loop}
Lemma~\ref{lem:path-length-3} implies $|P|=3$ for the indirect-lock loop, but the two joints along $P$ can in principle be any 1-DoF types in $\{R,P,G\}$. In our benchmark mining robots and common industrial designs, we only observed the $G$--$R$ and $R$--$G$ variants, corresponding to hydraulic-support base mechanisms that shape end-effector trajectories and improve stability and load capacity. 
Other combinations are theoretically possible but were not observed in our target domain and are therefore outside the present scope.

\paragraph{\hl{(L4)} Non-length-driven transmissions and additional constraints}
The framework is actuator length driven and kinematics centric. Mechanisms dominated by other transmissions (e.g., gear trains, cams, or cable systems), or scenarios requiring dynamics, contact, collision, or environment constraints, are outside the present~scope.

MineRobot does not claim universality beyond the defined scope. Within the major underground mining-robot classes considered in this paper, all actuators in our benchmarks were successfully classified into Types A--D, with no unclassified cases observed.

\section{Additional Details of the Kinematics Algorithms}
\label{appendix:B}

This appendix provides additional details of the kinematics algorithms described in Section~\ref{sec:solver}.
We first examine the numerical robustness of the per-ITEP FK root-finding procedure by comparing representative one-dimensional solvers on Type C and Type D cases, which involve four-bar and generalized-four-bar closure constraints.
We then present the FK post-processing step used to maintain endpoint-aligned actuator-body orientations.
Finally, we provide supplementary details for the IK solver, including additional robustness statistics, the detailed GSS procedure, and the ablation of different one-dimensional solvers within the same decomposed IK framework.

\subsection{Per-ITEP FK Root-Finding Robustness}
\label{appendix:fk-rootfinding}

Type C and Type D ITEPs are the most relevant cases for assessing the numerical behavior of the per-ITEP FK root finder because they involve planar four-bar or generalized-four-bar closure constraints.
Type A and Type B ITEPs are simpler scalar cases and converge rapidly in our tests; therefore, we focus the supplementary analysis on Types C and~D.

For each type, we generated 100 randomized trials by sampling initial configurations and target actuator lengths within the MRDF-defined admissible bounds.
For each trial, we recorded the residual \(|\varphi|\), where \(\varphi\) denotes the scalar actuator-length constraint residual of the corresponding per-ITEP FK subproblem.
All methods used the same residual tolerance of \(10^{-6}\), and the implementation limit for the per-ITEP FK solve was 50 iterations.
Figure~\ref{fig:fk-rootfinding} shows the first 20 iterations for several representative one-dimensional solvers.

For Type C, Newton--Raphson reduced the residual below \(10^{-6}\) within approximately 3 iterations on average, while Secant required about 5 iterations, Brent about 11 iterations, and Bisection about 16 iterations.
Evolution-strategy-based methods required substantially more iterations and are less suitable as default per-ITEP FK solvers.
Type D showed the same trend, with slightly higher iteration counts due to the more complex generalized-four-bar closure.
These results indicate that several representative one-dimensional solvers can be used for the Type C and Type D FK subproblems, while Newton--Raphson provides a favorable implementation simplicity and speed trade-off in our tested mechanisms.
The experiment is intended as empirical robustness evidence within the MRDF-defined operating ranges, rather than a global convergence proof for arbitrary closed-chain mechanisms.

\begin{figure}[htbp]
  \centering
    \includegraphics[width=1.0\linewidth]{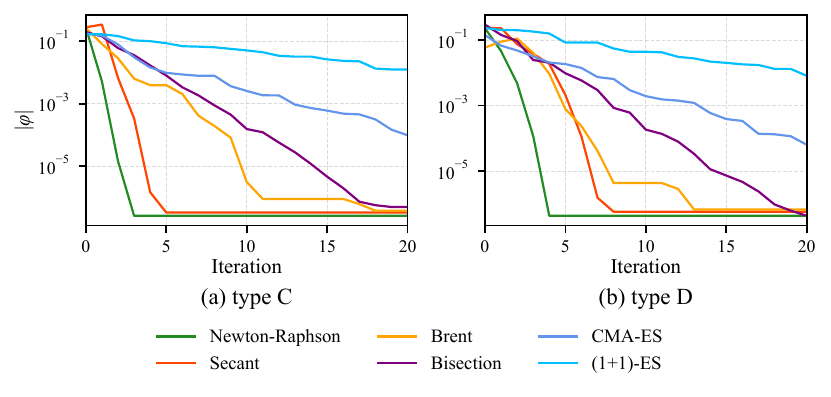}
    \caption{Residual convergence of one-dimensional FK root finding for Type C and Type D ITEPs. Each setting uses 100 randomized trials within the MRDF-defined actuator bounds. The residual \(|\varphi|\) denotes the actuator-length constraint error of the corresponding per-ITEP FK subproblem.}
    \label{fig:fk-rootfinding}
\end{figure}

\subsection{FK Post-Processing: Endpoint-Alignment Reorientation}
\label{appendix:B1}

After solving the forward kinematics, all actuators satisfy their target lengths; however, the tube and rod body frames may not yet be aligned with the line connecting the two mounting points, as illustrated in Figure~\ref{fig:lookat}.
In the final stage of the implementation, we revisit each actuator and align its tube and rod body frames with the endpoint-to-endpoint direction.
This endpoint-alignment step maintains the geometric consistency of actuator-body orientations and is consistent with endpoint-based hydraulic-actuator representations in multibody models~\cite{khadim2023experimental}.
It is applied after all actuator-length constraints have been satisfied, and it does not affect subsequent FK/IK solver iterations.

\begin{figure}[htbp]
  \centering
    \includegraphics[width=0.9\linewidth]{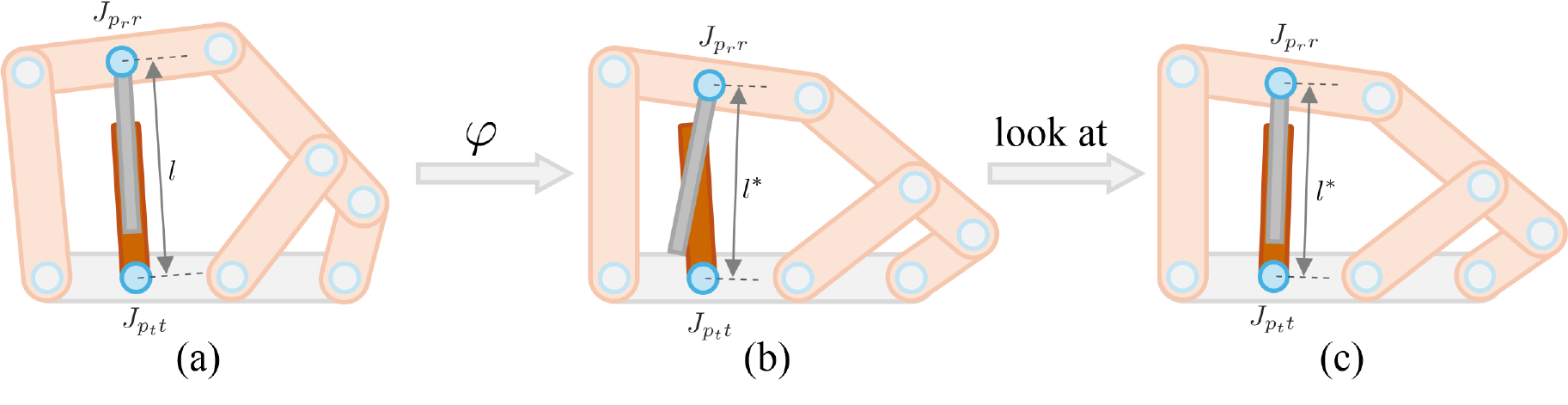}
    \caption{Endpoint-alignment reorientation. (\textbf{a},\textbf{b}) After length updates, tube/rod body frames may be inconsistent with the actuator endpoint direction; (\textbf{c}) we realign $J_{p_t t}$ and $J_{p_r r}$ with the line connecting the two mounting points while preserving length.}
    \label{fig:lookat}
\end{figure}

By convention, let the actuator tube and rod be denoted by \(L_t\) and \(L_r\), and their parent links by \(L_{p_t}\) and \(L_{p_r}\), respectively.
Let \( \mathbf{T}_{0 t} \) and \( \mathbf{T}_{0 r} \) denote the transformations from the world frame \(L_0\) to \(L_t\) and \(L_r\), and let \( \mathbf{R}_{0 t} \) and \( \mathbf{R}_{0 r} \) be their rotational components.
The rotation angles of joints \( J_{p_t t} \) and \( J_{p_r r} \) are then
\begin{equation}\label{eq:lookat}
\begin{aligned}
\theta_t &= \arccos\!\left( (\mathbf{R}_{0 t}^{-1}\hat{\mathbf d}_{tr})^\top \hat{\mathbf d}_0 \right), \quad
\theta_r &= \arccos\!\left( (\mathbf{R}_{0 r}^{-1}\hat{\mathbf d}_{rt})^\top \hat{\mathbf d}_0 \right).
\end{aligned}
\end{equation}
where
\[
\hat{\mathbf{d}}_{t r} =
\frac{ \mathbf{t}(\mathbf{T}_{0 r}) - \mathbf{t}(\mathbf{T}_{0 t}) }
{ \left\| \mathbf{t}(\mathbf{T}_{0 r}) - \mathbf{t}(\mathbf{T}_{0 t}) \right\| },
\quad
\hat{\mathbf{d}}_{r t} = -\hat{\mathbf{d}}_{t r},
\]
and \( \hat{\mathbf{d}}_0 \) is the initial orientation, by default aligned with the x-axis,
\[
\hat{\mathbf{d}}_0 = [1, 0, 0]^\top .
\]

\hl{The corresponding} transformations \( \mathbf{T}_{p_t t}(\theta_t) \) and \( \mathbf{T}_{p_r r}(\theta_r) \) are then updated accordingly.

\subsection{Additional Details of the IK Solver}
\label{appendix:B2}

This subsection complements the IK description in Section~\ref{sec:solver}.
We first provide additional robustness statistics under randomized targets and initialization, and then give the detailed GSS procedure together with an ablation of different one-dimensional solvers within the same decomposed IK framework.

\subsubsection{IK Robustness Statistics}
\label{appendix:B21}

This subsection provides additional statistics for the inverse kinematics experiments in Section~\ref{sec:ik-experiments}.
Our goal is to complement Table~\ref{tab:simulation_results} and Figure~\ref{fig:ik_results} with robustness evidence under randomized targets and randomized feasible initial states.
As in Section~\ref{sec:ik-experiments}, each initial configuration is generated by sampling actuator lengths within their bounds and applying FK, ensuring a feasible closed-chain state; the target pose is sampled from a bounding region around the robot workspace and may be unreachable.
Success denotes convergence according to the objective-change criterion of Algorithm~\ref{alg:ik} within the iteration budget, with all actuator lengths remaining within their bounds.
In particular, these statistics help assess whether stronger actuator coupling mainly slows convergence or also leads to degraded final solutions.

Table~\ref{tab:ik_robustness_compact} shows that Ours, Ours-NL, BFGS, and CG all converge reliably under the same tolerance across the tested robots.
For the more strongly coupled hydraulic supports (TCCHS, STHS, and SSHS), the median iteration counts and runtimes are indeed higher than for more weakly coupled robots such as RH, LHD, VD, DJ, and SH.
However, the final objective statistics remain stable: both the median and the P95 values of $\Psi^*$ stay at the same order of magnitude as those of BFGS and CG, without indicating systematic degradation on the coupled systems.
These results support the interpretation in Section~\ref{sec:ik-experiments}: in our benchmark, stronger coupling primarily affects convergence speed rather than increasing the observed risk of poor local minima.
The Ours-NL rows remain very close to the standard Ours rows in success rate, final residuals, iteration counts, and runtime.
This indicates that feasible near-boundary configurations within the MRDF-defined actuator ranges do not noticeably degrade IK robustness in the tested mechanisms.

\begin{table}[htbp]
  \caption{\textbf{IK robustness under randomized targets and initialization} (100 trials per robot). In the standard rows, the initial configuration is generated from randomized actuator lengths through FK, and the target $\mathbf{T}^*$ is randomized around the workspace. The \textbf{Ours-NL} row denotes a near-limit stress test using the same solver as Ours: each actuator length is sampled within the nearest \(5\%\) of either its lower or upper MRDF-defined stroke bound, and FK is used to obtain feasible closed-chain configurations. All methods use the same tolerance and stopping criterion as Section~\ref{sec:solver}. We report success rate, final objective statistics ($\Psi^*$, in units of \(10^{-6}\)), median iterations, and median runtime; Succ.\ denotes convergence within the iteration budget while respecting actuator bounds, and P95 denotes the 95th percentile.}
  \label{tab:ik_robustness_compact}
  \centering
  \resizebox{0.6\textwidth}{!}{
    \begin{tabular}{llrrrrr}
      \toprule
      \textbf{Robot} & \textbf{Method} & \textbf{Succ.\ (\%)} & \textbf{Med.}\ \boldmath{$\Psi^*$} \textbf{(}\boldmath{$10^{-6}$\textbf{)}} & \textbf{P95} \boldmath{$\Psi^*$} \textbf{(}\boldmath{$10^{-6}$}\textbf{)} & \textbf{Med.\ Iters} & \textbf{Med.\ Time (ms)} \\
      \midrule

      \multirow{4}{*}{\textbf{TCCHS}}
        & \textbf{Ours} & 100 & 0.25 & 0.60 & 10 & 71.48 \\
        & \textbf{Ours-NL} & 100 & 0.26 & 0.64 & 10 & 75.53 \\
        & BFGS          & 100 & 0.22 & 0.55 & 35 & 230.16 \\
        & CG            & 100 & 0.23 & 0.57 & 29 & 153.58 \\
      \midrule

      \multirow{4}{*}{\textbf{STHS}}
        & \textbf{Ours} & 100 & 0.28 & 0.65 & 12 & 82.52 \\
        & \textbf{Ours-NL} & 100 & 0.28 & 0.67 & 13 & 85.24 \\
        & BFGS          & 100 & 0.24 & 0.58 & 16 & 115.65 \\
        & CG            & 100 & 0.25 & 0.60 & 25 & 105.20 \\
      \midrule

      \multirow{4}{*}{\textbf{SSHS}}
        & \textbf{Ours} & 100 & 0.24 & 0.58 & 9  & 66.52 \\
        & \textbf{Ours-NL} & 100 & 0.23 & 0.62 & 9  & 68.67 \\
        & BFGS          & 100 & 0.23 & 0.56 & 15 & 102.65 \\
        & CG            & 100 & 0.24 & 0.58 & 24 & 123.58 \\
      \midrule

      \multirow{4}{*}{\textbf{RH}}
        & \textbf{Ours} & 100 & 0.05 & 0.12 & 1  & 0.20 \\
        & \textbf{Ours-NL} & 100 & 0.06 & 0.14 & 1  & 0.18 \\
        & BFGS          & 100 & 0.06 & 0.15 & 5  & 0.40 \\
        & CG            & 100 & 0.06 & 0.14 & 5  & 0.28 \\
      \midrule

      \multirow{4}{*}{\textbf{LHD}}
        & \textbf{Ours} & 100 & 0.08 & 0.18 & 1  & 4.12 \\
        & \textbf{Ours-NL} & 100 & 0.10 & 0.29 & 1  & 5.06 \\
        & BFGS          & 100 & 0.10 & 0.25 & 6  & 22.30 \\
        & CG            & 100 & 0.11 & 0.27 & 7  & 18.66 \\
      \midrule

      \multirow{4}{*}{\textbf{VD}}
        & \textbf{Ours} & 100 & 0.06 & 0.15 & 1  & 3.95 \\
        & \textbf{Ours-NL} & 100 & 0.07 & 0.21 & 1  & 4.33 \\
        & BFGS          & 100 & 0.09 & 0.22 & 8  & 31.39 \\
        & CG            & 100 & 0.09 & 0.24 & 6  & 16.46 \\
      \midrule

      \multirow{4}{*}{\textbf{DJ}}
        & \textbf{Ours} & 100 & 0.07 & 0.16 & 1  & 4.38 \\
        & \textbf{Ours-NL} & 100 & 0.11 & 0.18 & 1  & 4.74 \\
        & BFGS          & 100 & 0.08 & 0.20 & 5  & 23.48 \\
        & CG            & 100 & 0.08 & 0.22 & 6  & 17.50 \\
      \midrule

      \multirow{4}{*}{\textbf{SH}}
        & \textbf{Ours} & 100 & 0.04 & 0.10 & 1  & 0.14 \\
        & \textbf{Ours-NL} & 100 & 0.03 & 0.13 & 1  & 0.16 \\
        & BFGS          & 100 & 0.05 & 0.14 & 6  & 0.54 \\
        & CG            & 100 & 0.05 & 0.15 & 5  & 0.39 \\
      \bottomrule
    \end{tabular}
  }
\end{table}

\subsubsection{GSS Details and 1D-Solver Ablation}
\label{appendix:B22}

For a fixed actuator index \(i\), GSS searches over the admissible interval \([\ell_i,u_i]\) while keeping the other actuator lengths fixed at their current values.
Each candidate length is evaluated by running FK and computing the IK residual \(\Psi\).
The method maintains a bracketing interval \([a,b]\) and two interior points \(c\) and \(d\) separated according to the golden ratio.
At each step, the side with the larger residual is discarded, so the interval shrinks monotonically without requiring derivatives.
Algorithm~\ref{alg:gss} gives the detailed procedure used for each bounded one-dimensional IK subproblem in Section~\ref{sec:solver}.

Because only function values are used, the same procedure can be applied to all actuator types and to both reachable and unreachable target poses; however, it is still a local scalar search within the current Gauss--Seidel-style sweep.

To clarify whether the IK speed-up mainly comes from the ITEP-based decomposition or from the specific choice of the one-dimensional solver, we compare four 1D solvers within the same IK framework: GSS, Brent~\cite{brent1973algorithms}, Newton--Raphson, and Secant.
All four solvers are used under the same ITEP-based Gauss--Seidel-style outer iterations, the same bounds, and the same stopping criterion as in Section~\ref{sec:solver}.
All four solvers converged reliably to the same tolerance and produced comparable final objectives in our trials; therefore, Table~\ref{tab:ik_1d_ablation_compact} reports only iteration and runtime statistics.

\setcounter{algocf}{0}

\renewcommand\thealgocf{A\arabic{algocf}}

\setlength{\algotitleheightrule}{0.5pt}
\begin{algorithm}[htbp]
\small
\caption{Golden section search.}
\label{alg:gss}
\KwIn{Actuator index $i$; current actuator set $\mathcal{L}$; target transformation $\mathbf{T}^*$}
\KwOut{Updated actuator length $l_i^*$}

$a \gets \ell_i$; \quad $b \gets u_i$; \quad $\phi \gets \frac{\sqrt{5} - 1}{2}$ \\
$c \gets b - \phi (b - a)$; \quad $d \gets a + \phi (b - a)$ \\
$\Psi_c \gets \Psi(\mathcal{L} \mid l_i = c, \mathbf{T}^*)$ \\
$\Psi_d \gets \Psi(\mathcal{L} \mid l_i = d, \mathbf{T}^*)$ \\

\While{$|b - a| > \text{tol}$}{
    \If{$\Psi_c < \Psi_d$}{
        $b \gets d$; \quad $d \gets c$; \quad $\Psi_d \gets \Psi_c$ \\
        $c \gets b - \phi (b - a)$ \\
        $\Psi_c \gets \Psi(\mathcal{L} \mid l_i = c, \mathbf{T}^*)$
    }
    \Else{
        $a \gets c$; \quad $c \gets d$; \quad $\Psi_c \gets \Psi_d$ \\
        $d \gets a + \phi (b - a)$ \\
        $\Psi_d \gets \Psi(\mathcal{L} \mid l_i = d, \mathbf{T}^*)$
    }
}
\Return $l_i^* \gets \frac{a + b}{2}$
\end{algorithm}

\begin{table}[htbp]
  \caption{\textbf{Ablation of 1D solvers within the same IK decomposition framework} (100 trials per robot). All methods use the same ITEP-based Gauss--Seidel-style outer iterations, bounds, and stopping criterion; only the 1D subproblem solver is changed. All four solvers converged reliably to the same tolerance and produced comparable final objectives, so we report average and median outer iterations, and average and median runtime.}
  \label{tab:ik_1d_ablation_compact}
  \centering
  \resizebox{0.6\textwidth}{!}{
    \begin{tabular}{llrrrr}
      \toprule
      \textbf{Robot} & \textbf{1D Solver} & \textbf{Avg.\ Iter.} & \textbf{Med.\ Iter.} & \textbf{Avg.\ Time (ms)} & \textbf{Med.\ Time (ms)} \\
      \midrule

      \multirow{4}{*}{\textbf{TCCHS}}
        & \textbf{GSS}   & 11.43 & 11 & 78.63 & 71.48 \\
        & Brent          & 11.21 & 11 & 74.25 & 68.94 \\
        & Newton         & 11.35 & 11 & 84.90 & 77.35 \\
        & Secant         & 11.52 & 11 & 82.10 & 75.62 \\
      \midrule

      \multirow{4}{*}{\textbf{STHS}}
        & \textbf{GSS}   & 14.28 & 14 & 96.27 & 82.52 \\
        & Brent          & 14.05 & 14 & 91.40 & 78.36 \\
        & Newton         & 14.17 & 14 & 103.15 & 88.70 \\
        & Secant         & 14.33 & 14 & 100.82 & 86.94 \\
      \midrule

      \multirow{4}{*}{\textbf{SSHS}}
        & \textbf{GSS}   & 10.31 & 10 & 73.91 & 67.26 \\
        & Brent          & 10.14 & 10 & 69.47 & 63.22 \\
        & Newton         & 10.22 & 10 & 79.82 & 72.64 \\
        & Secant         & 10.45 & 10 & 77.61 & 70.63 \\
      \midrule

      \multirow{4}{*}{\textbf{RH}}
        & \textbf{GSS}   & 1.02 & 1 & 0.21 & 0.20 \\
        & Brent          & 1.01 & 1 & 0.18 & 0.17 \\
        & Newton         & 1.03 & 1 & 0.24 & 0.22 \\
        & Secant         & 1.04 & 1 & 0.23 & 0.21 \\
      \midrule

      \multirow{4}{*}{\textbf{LHD}}
        & \textbf{GSS}   & 1.08 & 1 & 4.21 & 4.12 \\
        & Brent          & 1.05 & 1 & 3.98 & 3.91 \\
        & Newton         & 1.12 & 1 & 4.86 & 4.73 \\
        & Secant         & 1.14 & 1 & 4.62 & 4.49 \\
      \midrule

      \multirow{4}{*}{\textbf{VD}}
        & \textbf{GSS}   & 1.06 & 1 & 4.06 & 3.95 \\
        & Brent          & 1.03 & 1 & 3.82 & 3.73 \\
        & Newton         & 1.09 & 1 & 4.40 & 4.27 \\
        & Secant         & 1.11 & 1 & 4.26 & 4.13 \\
      \midrule

      \multirow{4}{*}{\textbf{DJ}}
        & \textbf{GSS}   & 1.07 & 1 & 4.45 & 4.30 \\
        & Brent          & 1.04 & 1 & 4.18 & 4.05 \\
        & Newton         & 1.11 & 1 & 4.81 & 4.67 \\
        & Secant         & 1.13 & 1 & 4.67 & 4.53 \\
      \midrule

      \multirow{4}{*}{\textbf{SH}}
        & \textbf{GSS}   & 1.02 & 1 & 0.15 & 0.14 \\
        & Brent          & 1.01 & 1 & 0.13 & 0.12 \\
        & Newton         & 1.04 & 1 & 0.17 & 0.16 \\
        & Secant         & 1.03 & 1 & 0.16 & 0.15 \\
      \bottomrule
    \end{tabular}
  }
\end{table}

Table~\ref{tab:ik_1d_ablation_compact} shows that the average and median outer iteration counts are nearly identical across the four one-dimensional solvers.
This indicates that the dominant convergence behavior is governed by the ITEP-based decomposition and the Gauss--Seidel-style outer update framework, rather than by the specific one-dimensional solver.
The runtime differences are comparatively modest: Brent is slightly faster than GSS, whereas Newton--Raphson and Secant are slightly slower.
In this setting, GSS remains a practical and robust choice because it is derivative-free, simple to implement, and operates directly on the actuator bounds, thereby avoiding derivative computation and extra boundary handling.
Other classical one-dimensional solvers are also effective within the same decomposed IK framework.

\section{Application Details}
\label{appendix:C}

This appendix provides additional details for the applications briefly summarized in Section~\ref{sec:apps}.
These examples illustrate how the proposed MineRobot framework supports a spectrum of mining-robot application workflows.
Workspace visualization is implemented using the FK solver, while trajectory generation is driven by the IK solver.
We also developed virtual reality (VR) interfaces for robot operation and a digital-twin workflow prototype that leverages the kinematics solvers for bidirectional synchronization between physical and virtual systems.

\subsection{Workspace Visualization}
\label{appendix:C1}

Workspace visualization is essential for assessing reachability, task feasibility, and safe operation envelopes of mining robots.
Leveraging our FK solver, we realize this capability with a simple sampling pipeline.
A representative result is shown in Figure~\ref{fig:apps}a.

For a designated end-effector, let $\mathcal{A}_{\mathrm{rel}}$ denote the set of relevant actuators, and for each $A_i\in\mathcal{A}_{\mathrm{rel}}$ let the length interval be $b_i=[\ell_i,u_i]$.
We discretize each interval into $N_i$ samples and form the Cartesian grid
\[
\mathcal{B}=\prod_{A_i\in\mathcal{A}_{\mathrm{rel}}}\Bigl\{\ell_i+k\,(u_i-\ell_i)/(N_i-1)\;\big|\;k=0,\dots,N_i-1\Bigr\}.
\]

\hl{For every} sampled vector $\mathbf{l}_{\mathcal{A}_{\mathrm{rel}}}\in\mathcal{B}$, we assemble a complete target-length set $\mathcal{L}^*$ by assigning $\mathbf{l}_{\mathcal{A}_{\mathrm{rel}}}$ to the corresponding entries in $\mathcal{A}$, and evaluate FK as $\mathcal{T}=\Phi(\mathcal{L}^*,\mathcal{T}')$.
From $\mathcal{T}$, we then extract the end-effector transformation $\mathbf{T}_{0e}$.
Rendering the collection of these configurations in the virtual environment produces the workspace.
In our roadheader example in Figure~\ref{fig:apps}a, this process is used to visualize both sampled end-effector center positions and the corresponding end-effector poses within the roadway environment.

\subsection{Trajectory Generation}
\label{appendix:C2}

Trajectory generation is fundamental for motion planning and task execution in mining robotics.
In our framework, the inverse-kinematics (IK) solver serves as the computational backbone for trajectory synthesis.
A representative example is shown in Figure~\ref{fig:apps}b.

In a roadheader example, a set of end-effector via points is specified in the workspace and interpolated into a continuous trajectory curve $\mathbf{C}(t)$, $t\in[0,1]$, with collision constraints omitted for clarity.
The curve is discretized over $[0,1]$.
For each sample $t_k$, the corresponding target transformation is set to $\mathbf{T}^*_k=\mathbf{C}(t_k)$, and IK is solved to obtain the optimized actuator length set $\mathcal{L}^*_k$ and the resulting configuration.
By evaluating IK sequentially along the sampled path, the framework produces a feasible actuator-space trajectory together with the corresponding robot motion sequence.
This procedure also supports animation playback and interpolation in the virtual environment as illustrated by the trajectory-following sequence in Figure~\ref{fig:apps}b.

\subsection{Virtual Reality}
\label{appendix:C3}

Virtual reality (VR) can provide a safe and repeatable environment for operator-training and evaluation workflows in mining scenarios.
We developed an immersive VR system for interactive robot training as shown in Figure~\ref{fig:apps}c.

In our prototype, a user operates a virtual roadheader using a head-mounted display with hand tracking.
To maintain actuator-level correspondence with real hardware, we implement a virtual control panel for hydraulic support groups and a virtual console with levers for the roadheader.
All 3D user interfaces are procedurally generated from actuator parameters encoded in MRDF~\cite{hou2021vrmenudesigner}.
The virtual button states and the lever angles are mapped to actuator target lengths $\mathcal{L}^*$, and the FK solver computes the resulting configuration $\mathcal{T}$ at interactive rates.
The updated robot pose is then immediately rendered in the virtual scene.
This design enables intuitive operator interaction while preserving a direct correspondence between virtual controls and actuator-level robot motion as shown by the VR interface examples in Figure~\ref{fig:apps}c.

\subsection{Digital Twins}
\label{appendix:C4}

A digital twin is a bidirectionally coupled representation that synchronizes a physical system and its virtual model.
Such coupling is useful for safe monitoring, operator training, and decision-support workflows in mining robotics.
Our kinematics solvers provide the kinematics layer for this bidirectional synchronization.

We designed a digital-twin workflow prototype that maps kinematic states in both directions.
Physical robots, such as those used in longwall production and roadheading, are modeled in the virtual world via MRDF.
In the physical-to-virtual direction, sensor streams provide actuator lengths $\mathcal{L}^*$, which are fed to the FK solver to compute the configuration $\mathcal{T}$ and update the virtual robot.
In the virtual-to-physical direction, an interactively specified or automatically planned end-effector target $\mathbf{T}^*$ is passed to the IK solver to obtain target actuator lengths $\mathcal{L}^*$ for downstream command-planning modules.
The bidirectional synchronization workflow is illustrated in Figure~\ref{fig:apps}d.

Under a calibrated kinematic model, the actuator-centered formulation allows actuator-length measurements to estimate the robot configuration through FK, suggesting a practical kinematics-layer role in constrained mining environments.
Full digital-twin accuracy, however, depends on calibration quality, sensing accuracy, dynamics, contact interactions, hydraulic control, and system-level integration, which are outside the scope of the present kinematics-focused study.

\section{MRDF Schema and Conventions}
\label{appendix:mrdf-schema}

MRDF is a JSON-based research prototype format used as the input representation of MineRobot.
The version used in this paper is MRDF~v0.1.
An MRDF file contains three required top-level arrays: \texttt{links}, \texttt{joints}, and \texttt{actuators}.
MRDF~v0.1 is intended to support the actuator-centered kinematics pipeline studied in this paper, rather than to serve as a general-purpose robotics interchange standard.
Table~\ref{tab:mrdf_schema} summarizes the key fields used by the parser and the kinematics algorithms.

\begin{table}[htbp]
  \caption{Key MRDF fields for link, joint, and actuator entities.}
  \label{tab:mrdf_schema}
  \centering 
  \resizebox{\textwidth}{!}{
    \begin{tabular}{lllll}
      \toprule
      \textbf{Entity} & \textbf{Field} & \textbf{Type/Unit} & \textbf{Required} & \textbf{Description} \\
      \midrule

      \multirow{4}{*}{Link}
      & \texttt{name} & string & Yes & Unique link identifier. \\
      & \texttt{origin\_translation} & array[3]/m & Yes & Link-frame translation. \\
      & \texttt{origin\_orientation} & array[3]/rad & Yes & Link-frame RPY orientation. \\
      & \texttt{visual} & object & Optional & Primitive geometry and material, or an external model link. \\

      \midrule

      \multirow{7}{*}{Joint}
      & \texttt{name} & string & Yes & Unique joint identifier. \\
      & \texttt{parent} & string & Yes & Parent link name. \\
      & \texttt{child} & string & Yes & Child link name. \\
      & \texttt{type} & enum & Yes & Joint type: \texttt{Revolute}, \texttt{Prismatic}, or \texttt{Fixed}. \\
      & \texttt{origin\_translation} & array[3]/m & Yes & Joint-frame translation. \\
      & \texttt{origin\_orientation} & array[3]/rad & Yes & Joint-frame RPY orientation. \\
      & \texttt{axis} & array[3] & For R/P & Revolute or prismatic motion axis. \\

      \midrule

      \multirow{8}{*}{Actuator}
      & \texttt{name} & string & Yes & Unique actuator identifier. \\
      & \texttt{tube\_parent} & string & Yes & Tube-side mounting parent link. \\
      & \texttt{rod\_parent} & string & Yes & Rod-side mounting parent link. \\
      & \texttt{tube\_offset} & array[3]/m & Yes & Tube-side mounting offset. \\
      & \texttt{rod\_offset} & array[3]/m & Yes & Rod-side mounting offset. \\
      & \texttt{limit} & object/m & Yes & Stroke bounds \texttt{lower} and \texttt{upper}. \\
      & \texttt{redundants} & array[string] & Optional & Redundant actuator names. \\
      & \texttt{tube\_visual}, \texttt{rod\_visual} & object & Optional & Tube/rod visual geometry and material. \\
      \bottomrule
    \end{tabular}
  }
\end{table}

All names referenced by \texttt{parent}, \texttt{child}, \texttt{tube\_parent}, \texttt{rod\_parent}, and \texttt{redundants} must refer to entities defined in the same MRDF file.
Lengths are expressed in meters and angles in radians.
Frames are right-handed, and origin translations/orientations are expressed relative to the corresponding parent frame.
RPY angles follow the URDF convention, i.e., fixed-axis/extrinsic XYZ rotations.
Visual fields are optional and do not affect kinematic solving.
They may define primitive geometry and material parameters, or provide a link to an external mesh model.
Actuator redundancy is encoded by the \texttt{redundants} array.
Actuators listed as redundant are treated as belonging to the same actuator redundancy group during topology processing and solver scheduling.
If the array is empty or omitted, the actuator is treated as a single-actuator group.

\section{Experiment Details}
\label{appendix:experiment-details}

This appendix summarizes the reproducibility settings used in the randomized experiments.
The full industrial MRDF benchmark files used for the quantitative experiments contain proprietary geometric parameters, actuator stroke limits, linkage dimensions, and commercial robot data; therefore, these files and the exact industrial actuator bounds cannot be publicly released.
To support the reproducibility of the proposed framework within these restrictions, we provide a lightweight open-source implementation of the core MineRobot pipeline, an online interactive demo, and representative simplified MRDF examples, as described in the Data Availability Statement.
The public MRDF examples include non-proprietary actuator bounds and preserve the main kinematic structures used in this work, including linear actuators, closed-chain topology, planar four-bar substructures, redundant actuator groups, and coupled actuator mechanisms.
They are intended to support reproduction of MRDF parsing, topology processing, FK/IK solving, and bound-constrained actuator interaction, although they do not reproduce the proprietary numerical benchmark tables exactly.
Table~\ref{tab:reproducibility_general} summarizes the common settings used in the randomized experiments unless otherwise stated.

\begin{table}[htbp]
  \caption{General settings used in the randomized experiments.}
  \label{tab:reproducibility_general}
  \centering
  \resizebox{\textwidth}{!}{
    \begin{tabular}{ll}
      \toprule
      \textbf{Item} & \textbf{Setting} \\
      \midrule
      Random seed & 42 for all randomized experiments \\
      Trial count & 100 trials per robot or per target set unless otherwise stated \\
      Actuator sampling & Uniform sampling within MRDF-defined bounds \([\ell_i,u_i]\) \\
      FK tolerance & \(10^{-6}\) actuator-length and loop-closure residual tolerance \\
      FK max iterations & 50 per ITEP solve \\
      IK tolerance & \(10^{-6}\) objective-change tolerance \\
      IK max outer iterations & 50 \\
      IK time budget & \(1\,\mathrm{s}\) per trial \\
      \bottomrule
    \end{tabular}
  }
\end{table}
For the IK optimizer comparison in Section~\ref{sec:ik-experiments}, all methods use the same objective function \(\Psi(\mathcal{L},\mathbf{T}^*)\), actuator bounds, initial actuator lengths, target poses, random seed, tolerance, maximum iteration budget, and time budget.
For BFGS, conjugate gradient (CG), and gradient descent, gradients are computed by finite differences, and actuator bounds are enforced by projection, i.e., candidate actuator lengths are clamped to \([\ell_i,u_i]\) after each update.
For CMA-ES, the initial mean is set to the current actuator-length vector, the initial step size is \(\sigma=0.5\), the population size is \(\lambda=6\), and the best \(\mu=3\) candidates are used to update the mean, covariance matrix, and global step size; actuator bounds are enforced by clamping, and the maximum number of iterations is 200.
For (1+1)-ES, the initial state is the current actuator-length vector and the initial mutation scale is \(\sigma=0.5\).
A successful mutation expands the mutation scale by \(1.2\), whereas an unsuccessful mutation shrinks it by \(0.8\), with \(\sigma\in[10^{-8},10^4]\).
Candidate actuator lengths are clamped to their bounds, and the maximum number of iterations is 200.

For the Drake and MuJoCo reference workflows in Section~\ref{comparison-runtime}, the robot geometries and actuator bounds are kept consistent with the MRDF models.
Drake models are encoded in SDF, and MuJoCo models are encoded in MJCF.
The closed-loop constraints and actuator-length tasks required by the benchmark mechanisms are explicitly specified in the reference workflows.
Both reference workflows use the current actuator state as the initial guess, a constraint/solver tolerance of \(10^{-6}\), a maximum iteration count of 50, and a time budget of \(1\,\mathrm{s}\) per trial.
These comparisons are workflow-level references under matched settings, rather than claims about the intrinsic performance limits of Drake or MuJoCo.
For the MuJoCo IK reference entries marked with \(^{*}\) in Table~\ref{tab:engine_runtime_comparison}, we use an external reference workflow because MuJoCo does not provide a native closed-chain IK pipeline.
The workflow optimizes the full robot configuration, evaluates forward kinematics in MuJoCo at each iteration, and minimizes a weighted residual combining end-effector pose error, closed-loop constraint error, and regularization using finite-difference Levenberg--Marquardt updates with line search and bound projection.



\begin{thebibliography}{999}

\bibitem[Peng(2019)]{peng2019longwall}
Peng, S.
\newblock {\em Longwall Mining}; CRC Press: \hl{Boca Raton, FL, USA}, 2019.

\bibitem[Ralston et~al.(2017)Ralston, Hargrave, and Dunn]{ralston2017longwall}
Ralston, J.C.; Hargrave, C.O.; Dunn, M.T.
\newblock Longwall automation: Trends, challenges and opportunities.
\newblock {\em Int. J. Min. Sci. Technol.} {\bf
  2017}, {\em 27},~733--739.

\bibitem[Deshmukh et~al.(2020)Deshmukh, Raina, Murthy, Trivedi, and
  Vajre]{deshmukh2020roadheader}
Deshmukh, S.; Raina, A.; Murthy, V.; Trivedi, R.; Vajre, R.
\newblock Roadheader--A comprehensive review.
\newblock {\em Tunn. Undergr. Space Technol.} {\bf 2020}, {\em
  95},~103148.

\bibitem[Yan et~al.(2025)Yan, Zhao, Hou, and Lu]{Yan2025roadheader}
Yan, C.; Zhao, G.; Hou, S.; Lu, X.
\newblock Advancing unmanned roadheader-based tunneling: A framework for
  challenging underground environments.
\newblock {\em Geo-Spat. Inf. Sci.} {\bf 2025}, \hl{1--18}
.

\bibitem[P{\'e}rez et~al.(2019)P{\'e}rez, Diez, Usamentiaga, and
  Garc{\'\i}a]{perez2019industrial}
P{\'e}rez, L.; Diez, E.; Usamentiaga, R.; Garc{\'\i}a, D.F.
\newblock Industrial robot control and operator training using virtual reality
  interfaces.
\newblock {\em Comput. Ind.} {\bf 2019}, {\em 109},~114--120.

\bibitem[Guan et~al.(2019)Guan, Miao, Li, Liu, and Zhao]{guan2019dynamic}
Guan, E.; Miao, H.; Li, P.; Liu, J.; Zhao, Y.
\newblock Dynamic model analysis of hydraulic support.
\newblock {\em Adv. Mech. Eng.} {\bf 2019}, {\em
  11},~1687814018820143.

\bibitem[Guo et~al.(2024)Guo, Wang, and Liu]{guo2024adaptive}
Guo, X.; Wang, H.; Liu, H.
\newblock Adaptive sliding mode control with disturbance estimation for
  hydraulic actuator systems and application to rock drilling jumbo.
\newblock {\em Appl. Math. Model.} {\bf 2024}, {\em 136},~115637.

\bibitem[Prebil et~al.(2002)Prebil, Kra{\v{s}}na, and
  Ciglari{\v{c}}]{prebil2002synthesis}
Prebil, I.; Kra{\v{s}}na, S.; Ciglari{\v{c}}, I.
\newblock Synthesis of four-bar mechanism in a hydraulic support using a global
  optimization algorithm.
\newblock {\em Struct. Multidiscip. Optim.} {\bf 2002}, {\em
  24},~246--251.

\bibitem[Merlet(2006)]{merlet2006parallel}
Merlet, J.P.
\newblock {\em Parallel Robots}; Springer Science \& Business Media: \hl{Berlin/Heidelberg, Germany}, 2006; Volume 128.

\bibitem[Taghirad(2013)]{taghirad2013parallel}
Taghirad, H.D.
\newblock {\em Parallel Robots: Mechanics and Control}; CRC Press: \hl{Boca Raton, FL, USA},  2013.

\bibitem[Mei et~al.(2025)Mei, Wang, Xie, Li, and Liu]{mei2025sensing}
Mei, Z.; Wang, X.; Xie, J.; Li, S.; Liu, J.
\newblock A sensing system and solving method for dynamic detection of relative
  pose of hydraulic support group.
\newblock {\em Measurement} {\bf 2025}, {\em 243},~116145.

\bibitem[Ge et~al.(2020)Ge, Xie, Wang, Liu, and Shi]{ge2020virtual}
Ge, X.; Xie, J.; Wang, X.; Liu, Y.; Shi, H.
\newblock A virtual adjustment method and experimental study of the support
  attitude of hydraulic support groups in propulsion state.
\newblock {\em Measurement} {\bf 2020}, {\em 158},~107743.

\bibitem[Xie et~al.(2022)Xie, Liu, and Wang]{xie2022framework}
Xie, J.; Liu, S.; Wang, X.
\newblock Framework for a closed-loop cooperative human Cyber-Physical System
  for the mining industry driven by VR and AR: MHCPS.
\newblock {\em Comput. Ind. Eng.} {\bf 2022}, {\em
  168},~108050.

\bibitem[Hou et~al.(2023)Hou, Lu, Gao, Jiang, and Zhang]{hou2023interactive}
Hou, S.; Lu, X.; Gao, W.; Jiang, S.; Zhang, X.
\newblock Interactive physically based simulation of roadheader robot.
\newblock {\em Arab. J. Sci. Eng.} {\bf 2023}, {\em
  48},~2441--2454.

\bibitem[Xie et~al.(2022)Xie, Ge, Cui, and Wang]{xie2022virtual}
Xie, J.; Ge, F.; Cui, T.; Wang, X.
\newblock A virtual test and evaluation method for fully mechanized mining
  production system with different smart levels.
\newblock {\em Int. J. Coal Sci. Technol.} {\bf 2022},
  {\em 9},~41.

\bibitem[Quigley et~al.(2009)Quigley, Conley, Gerkey, Faust, Foote, Leibs,
  Wheeler, Ng, et~al.]{quigley2009ros}
Quigley, M.; Conley, K.; Gerkey, B.; Faust, J.; Foote, T.; Leibs, J.; Wheeler,
  R.; Ng, A.Y.;  Berger, E.
\newblock ROS: An open-source Robot Operating System.
\newblock In Proceedings of the ICRA Workshop on Open Source Software, Kobe, \hl{Japan, 12--17 May} 2009; Volume~3, p.~5.

\bibitem[{Open Source Robotics Foundation}(2019)]{sdformat}
{Open Source Robotics Foundation}. {SDFormat: Simulation Description Format}. 2019.   Available online: \url{http://sdformat.org}.

\bibitem[Todorov et~al.(2012)Todorov, Erez, and Tassa]{todorov2012mujoco}
Todorov, E.; Erez, T.; Tassa, Y.
\newblock Mujoco: A physics engine for model-based control.
\newblock In \emph{Proceedings of the 2012 IEEE/RSJ International Conference on
  Intelligent Robots and Systems}; IEEE: \hl{New York, NY, USA}, 2012; pp. 5026--5033.

\bibitem[Tedrake and the Drake Development~Team(2019)]{drake}
Tedrake, R.; the Drake Development~Team.
\newblock \hl{Drake: Model-based design and verification for robotics. Available online: https://drake.mit.edu.}  



\bibitem[Tian et~al.(2018)Tian, Wang, and Wu]{tian2018kinematic}
Tian, J.; Wang, S.; Wu, M.
\newblock Kinematic models and simulations for trajectory planning in the
  cutting of Spatially-Arbitrary crosssections by a robotic roadheader.
\newblock {\em Tunn. Undergr. Space Technol.} {\bf 2018}, {\em
  78},~115--123.

\bibitem[Gosselin and Schreiber(2018)]{gosselin2018redundancy}
Gosselin, C.; Schreiber, L.T.
\newblock Redundancy in parallel mechanisms: A review.
\newblock {\em Appl. Mech. Rev.} {\bf 2018}, {\em 70},~010802.

\bibitem[Liu et~al.(2024)Liu, Zhou, Qiao, and Zhu]{liu2024relative}
Liu, P.; Zhou, H.; Qiao, X.; Zhu, Y.
\newblock On the Relative Kinematics and Control of Dual-Arm Cutting Robots for
  a Coal Mine.
\newblock {\em Actuators} {\bf 2024}, {\em 13},~157.
\newblock {\url{https://doi.org/10.3390/act13050157}}.

\bibitem[Wang et~al.(2024)Wang, Xie, Zheng, Wang, and Wang]{wang2024method}
Wang, Y.; Xie, J.; Zheng, Z.; Wang, Y.; Wang, X.
\newblock Method for Collision Relationship of Hydraulic Supports Considering
  Multilayer Sensing Errors.
\newblock {\em IEEE Sensors J.} {\bf 2024}, \hl{\emph{24}, 30895--30908}.

\bibitem[Nordmann et~al.(2014)Nordmann, Hochgeschwender, and
  Wrede]{nordmann2014survey}
Nordmann, A.; Hochgeschwender, N.; Wrede, S.
\newblock A survey on domain-specific languages in robotics.
\newblock In \emph{Proceedings of the International Conference on Simulation,
  Modeling, and Programming for Autonomous Robots}; Springer: \hl{Berlin/Heidelberg, Germany},  2014; pp. 195--206.

\bibitem[Qiu et~al.(2024)Qiu, Song, and Wan]{qiu2024describing}
Qiu, N.; Song, C.; Wan, F.
\newblock Describing Robots from Design to Learning: Towards an Interactive
  Lifecycle Representation of Robots.
\newblock In \emph{Proceedings of the 2024 International Conference on Advanced
  Robotics and Mechatronics (ICARM)}; IEEE: \hl{New York, NY, USA},  2024; pp. 1081--1086.

\bibitem[Ivanou et~al.(2021)Ivanou, Mikhel, and Savin]{ivanou2021robot}
Ivanou, M.; Mikhel, S.; Savin, S.
\newblock Robot description formats and approaches.
\newblock In \emph{Proceedings of the 2021 International Conference "Nonlinearity,
  Information and Robotics" (NIR)}; IEEE: \hl{New York, NY, USA},  2021; pp. 1--5.

\bibitem[Tola and Corke(2023)]{tola2023understanding}
Tola, D.; Corke, P.
\newblock Understanding urdf: A survey based on user experience.
\newblock In \emph{Proceedings of the 2023 IEEE 19th International Conference on
  Automation Science and Engineering (CASE)}; IEEE: \hl{New York, NY, USA},  2023; pp. 1--7.

\bibitem[Macenski et~al.(2022)Macenski, Foote, Gerkey, Lalancette, and
  Woodall]{macenski2022robot}
Macenski, S.; Foote, T.; Gerkey, B.; Lalancette, C.; Woodall, W.
\newblock Robot operating system 2: Design, architecture, and uses in the wild.
\newblock {\em Sci. Robot.} {\bf 2022}, {\em 7},~eabm6074.

\bibitem[Chignoli et~al.(2024)Chignoli, Slotine, Wensing, and
  Kim]{chignoli2024urdf+}
Chignoli, M.; Slotine, J.J.; Wensing, P.M.; Kim, S.
\newblock URDF+: An Enhanced URDF for Robots with Kinematic Loops.
\newblock In \emph{Proceedings of the 2024 IEEE-RAS 23rd International Conference on
  Humanoid Robots (Humanoids)}; IEEE: \hl{New York, NY, USA},  2024; pp. 197--204.

\bibitem[Batto et~al.(2025)Batto, De~Matte{\"\i}s, and
  Mansard]{batto2025extended}
Batto, V.; De~Matte{\"\i}s, L.; Mansard, N.
\newblock Extended URDF: Accounting for parallel mechanism in robot
  description.
\newblock In \emph{Proceedings of the International Conference on Robotics in
  Alpe-Adria Danube Region}; Springer: \hl{Berlin/Heidelberg, Germany}, 2025; pp. 332--340.

\bibitem[Siciliano and Khatib(2016)]{siciliano2016robotics}
Siciliano, B.; Khatib, O.
\newblock Robotics and the handbook. In {\em Springer Handbook of Robotics};
  Springer: \hl{Berlin/Heidelberg, Germany},  2016; pp. 1--6.

\bibitem[Aristidou et~al.(2018)Aristidou, Lasenby, Chrysanthou, and
  Shamir]{aristidou2018inverse}
Aristidou, A.; Lasenby, J.; Chrysanthou, Y.; Shamir, A.
\newblock Inverse kinematics techniques in computer graphics: A survey.
\newblock In \emph{Proceedings of the Computer Graphics Forum}; Wiley Online Library: \hl{Hoboken, NJ, USA}, 2018; Volume~37, pp. 35--58.

\bibitem[Le~Naour et~al.(2019)Le~Naour, Courty, and Gibet]{le2019kinematics}
Le~Naour, T.; Courty, N.; Gibet, S.
\newblock Kinematics in the metric space.
\newblock {\em Comput. Graph.} {\bf 2019}, {\em 84},~13--23.

\bibitem[Staicu(2019)]{staicu2019dynamics}
Staicu, S.
\newblock {\em Dynamics of Parallel Robots}; Springer: \hl{Berlin/Heidelberg, Germany},  2019.

\bibitem[Choi et~al.(2021)Choi, Lee, Kim, Lim, and Kwon]{choi2021kinematic}
Choi, K.B.; Lee, J.; Kim, G.; Lim, H.; Kwon, S.
\newblock Kinematic Analysis of a Parallel Manipulator Driven by Perpendicular
  Linear Actuators.
\newblock {\em Actuators} {\bf 2021}, {\em 10},~262.
\newblock {\url{https://doi.org/10.3390/act10100262}}.

\bibitem[Zhao and Badler(1994)]{zhao1994inverse}
Zhao, J.; Badler, N.I.
\newblock Inverse kinematics positioning using nonlinear programming for highly
  articulated figures.
\newblock {\em ACM Trans. Graph. (TOG)} {\bf 1994}, {\em
  13},~313--336.

\bibitem[Erleben and Andrews(2019)]{erleben2019solving}
Erleben, K.; Andrews, S.
\newblock Solving inverse kinematics using exact Hessian matrices.
\newblock {\em Comput. Graph.} {\bf 2019}, {\em 78},~1--11.

\bibitem[Koenig and Howard(2004)]{koenig2004gazebo}
Koenig, N.; Howard, A.
\newblock Design and use paradigms for gazebo, an open-source multi-robot
  simulator.
\newblock In \emph{Proceedings of the 2004 IEEE/RSJ International Conference on
  Intelligent Robots and Systems (IROS) (IEEE Cat. No. 04CH37566)}; IEEE: \hl{New York, NY, USA},  2004;   Volume~3, pp. 2149--2154.

\bibitem[Coumans and Bai(2021)]{coumans2021pybullet}
Coumans, E.; Bai, Y.
\newblock PyBullet, a Python Module for Physics Simulation for Games, Robotics
  and Machine Learning.  2021. Available online: \url{http://pybullet.org}.


\bibitem[{NVIDIA}(2021)]{nvidia_isaac_sim}
{NVIDIA}.
\newblock {NVIDIA Isaac Sim}: Robotics Simulation and Synthetic Data
  Generation.  2021. Available online: \url{https://developer.nvidia.com/isaac/sim}.

\bibitem[Lee et~al.(2018)Lee, Grey, Ha, Kunz, Jain, Ye, Srinivasa, Stilman, and
  Liu]{lee2018dart}
Lee, J.; Grey, M.X.; Ha, S.; Kunz, T.; Jain, S.; Ye, Y.; Srinivasa, S.S.;
  Stilman, M.; Liu, C.K.
\newblock {DART}: Dynamic Animation and Robotics Toolkit.
\newblock {\em J. Open Source Softw.} {\bf 2018}, {\em 3},~500.
\newblock {\url{https://doi.org/10.21105/joss.00500}}.

\bibitem[Maloisel et~al.(2023)Maloisel, Schumacher, Knoop, Grandia, and
  B{\"a}cher]{maloisel2023optimal}
Maloisel, G.; Schumacher, C.; Knoop, E.; Grandia, R.; B{\"a}cher, M.
\newblock Optimal design of robotic character kinematics.
\newblock {\em ACM Trans. Graph. (TOG)} {\bf 2023}, {\em 42},~1--15.

\bibitem[Sciavicco and Siciliano(2012)]{sciavicco2012modelling}
Sciavicco, L.; Siciliano, B.
\newblock {\em Modelling and Control of Robot Manipulators}; Springer Science
  \& Business Media: \hl{Berlin/Heidelberg, Germany},  2012.

\bibitem[Norton and Han(2007)]{norton2007design}
Norton, R.L.; Han, J.
\newblock {\em Design of Machinery}; Higher Education Press: \hl{Beijing, China},  2007.

\bibitem[McCarthy and Soh(2010)]{mccarthy2010geometric}
McCarthy, J.M.; Soh, G.S.
\newblock {\em Geometric Design of Linkages}; Springer Science \&
  Business Media: \hl{Berlin/Heidelberg, Germany},  2010; Volume~11.

\bibitem[Ebrahimi and Payvandy(2015)]{ebrahimi2015efficient}
Ebrahimi, S.; Payvandy, P.
\newblock Efficient constrained synthesis of path generating four-bar
  mechanisms based on the heuristic optimization algorithms.
\newblock {\em Mech. Mach. Theory} {\bf 2015}, {\em 85},~189--204.

\bibitem[Tarjan(1972)]{tarjan1972depth}
Tarjan, R.
\newblock Depth-first search and linear graph algorithms.
\newblock {\em SIAM J. Comput.} {\bf 1972}, {\em 1},~146--160.

\bibitem[Pang and Shi(2025)]{pang2025intelligent}
Pang, Y.; Shi, Y.
\newblock Intelligent control algorithms for posture and height control of
  four-leg hydraulic supports.
\newblock {\em Sci. Rep.} {\bf 2025}, {\em 15},~3010.

\bibitem[Mu et~al.(2024)Mu, Xie, and Yang]{mu2024research}
Mu, M.; Xie, B.; Yang, Y.
\newblock Research on attitude analysis of hydraulic support based on column
  length.
\newblock {\em Stroj.-Vestn.-J. Mech. Eng.} {\bf
  2024}, {\em 70},~293--310.

\bibitem[Khadim et~al.(2023)Khadim, Hagh, Jiang, Pyrh{\"o}nen, Jaiswal,
  Zhidchenko, Yu, Kurvinen, Handroos, and Mikkola]{khadim2023experimental}
Khadim, Q.; Hagh, Y.S.; Jiang, D.; Pyrh{\"o}nen, L.; Jaiswal, S.; Zhidchenko,
  V.; Yu, X.; Kurvinen, E.; Handroos, H.; Mikkola, A.
\newblock Experimental investigation into the state estimation of a forestry
  crane using the unscented Kalman filter and a multiphysics model.
\newblock {\em Mech. Mach. Theory} {\bf 2023}, {\em 189},~105405.
\newblock {\url{https://doi.org/10.1016/j.mechmachtheory.2023.105405}}.

\bibitem[Lynch and Park(2017)]{lynch2017modern}
Lynch, K.M.; Park, F.C.
\newblock {\em Modern Robotics}; Cambridge University Press: \hl{Cambridge, UK},  2017.

\bibitem[Kiefer(1953)]{kiefer1953sequential}
Kiefer, J.
\newblock Sequential Minimax Search for a Maximum.
\newblock {\em Proc. Am. Math. Soc.} {\bf 1953},
  {\em 4},~502--506.
\newblock {\url{https://doi.org/10.1090/S0002-9939-1953-0055639-3}}.

\bibitem[Conn et~al.(2009)Conn, Scheinberg, and Vicente]{conn2009introduction}
Conn, A.R.; Scheinberg, K.; Vicente, L.N.
\newblock {\em Introduction to Derivative-Free Optimization};  {MPS-SIAM Series on Optimization}; Society for Industrial and Applied
  Mathematics: Philadelphia, PA, \hl{USA}, 2009; Volume~8.
\newblock {\url{https://doi.org/10.1137/1.9780898718768}}.


\bibitem[Guennebaud et~al.(2010)Guennebaud, Jacob, et~al.]{eigenweb}
Guennebaud, G.; Jacob, B.
\newblock Eigen, 2010, Version 3.4.0. Available online: https://eigen.tuxfamily.org. 

\bibitem[Li et~al.(2023)Li, Xu, Ye, Ren, and Liu]{li2023difffr}
Li, Z.; Xu, Q.; Ye, X.; Ren, B.; Liu, L.
\newblock Difffr: Differentiable sph-based fluid-rigid coupling for rigid body
  control.
\newblock {\em ACM Trans. Graph. (TOG)} {\bf 2023}, {\em 42},~1--17.

\bibitem[Geilinger et~al.(2020)Geilinger, Hahn, Zehnder, B{\"a}cher,
  Thomaszewski, and Coros]{geilinger2020add}
Geilinger, M.; Hahn, D.; Zehnder, J.; B{\"a}cher, M.; Thomaszewski, B.; Coros,
  S.
\newblock Add: Analytically differentiable dynamics for multi-body systems with
  frictional contact.
\newblock {\em ACM Trans. Graph. (TOG)} {\bf 2020}, {\em 39},~1--15.

\bibitem[{Pixar Animation Studios}(2016)]{usd_pixar}
{Pixar Animation Studios}. {Universal Scene Description (OpenUSD)}. 2016.  Available online:  \url{https://openusd.org/} .

\bibitem[Brent(1973)]{brent1973algorithms}
Brent, R.P.
\newblock {\em Algorithms for Minimization without Derivatives}; Prentice-Hall:
  Englewood Cliffs, NJ, \hl{USA}, 1973.

\bibitem[Hou et~al.(2021)Hou, Thomas, and Lu]{hou2021vrmenudesigner}
Hou, S.; Thomas, B.H.; Lu, X.
\newblock VRMenuDesigner: A toolkit for automatically generating and modifying
  VR menus.
\newblock In \emph{Proceedings of the 2021 IEEE International Conference on
  Artificial Intelligence and Virtual Reality (AIVR)}; IEEE: \hl{New York, NY, USA},  2021; pp.~154--159.

\end{thebibliography}
\end{document}